\documentclass[a4paper,11pt]{article} 

\usepackage[utf8]{inputenc}
\usepackage{amsmath,amsfonts,amssymb,amsthm,mathtools}
\usepackage{dsfont}
\usepackage[ngerman,english]{babel}
\usepackage[inline]{enumitem}
\usepackage[all]{xy}
\usepackage{microtype}
\usepackage{color}
\usepackage[sortcites, maxbibnames=10]{biblatex}
\numberwithin{equation}{section}

\DeclareMathOperator{\tr}{Tr}

\DeclareMathOperator{\erf}{erf}
\DeclareMathOperator{\Span}{Span}

\newtheorem{theorem}{Theorem}[section]
\newtheorem{proposition}[theorem]{Proposition}
\newtheorem{lemma}[theorem]{Lemma}
\newtheorem{corollary}[theorem]{Corollary}
\theoremstyle{definition}
\newtheorem{remark}[theorem]{Remark}

\newcommand{\dda}{\mathrm{d}}
\newcommand{\de}{\,\dda}

\renewcommand\rho\varrho
\renewcommand\epsilon\varepsilon

\definecolor{darkred}{rgb}{0.9,0,0.3}
\definecolor{darkblue}{rgb}{0,0.3,0.9}

\definecolor{vdarkred}{rgb}{0.7,0,0.2}
\definecolor{vdarkblue}{rgb}{0,0.2,0.7}
\usepackage[pdftex, colorlinks, linkcolor=vdarkblue,citecolor=vdarkred]{hyperref}

\def\eps{\varepsilon}
\def\cchi{\raisebox{0.4ex}{$\chi$}}
\def\ph{\varphi}

          \let\ph=\varphi

\theoremstyle{definition}

\newcommand{\beq}{\begin{equation}}
\newcommand{\eeq}{\end{equation}}
\usepackage{pdfsync}

\begin{document}

\title{Upper bound for the grand canonical free energy of the Bose gas in the Gross-Pitaevskii limit for general interaction potentials}

\author{Marco Caporaletti, Andreas Deuchert}

\date{\today}

\maketitle

\begin{abstract} 
We consider a homogeneous Bose gas in the Gross--Pitaevskii limit at temperatures that are comparable to the critical temperature for Bose--Einstein condensation. Recently, an upper bound for the grand canonical free energy was proved in \cite{BocDeuStock2023} capturing two novel contributions. First, the free energy of the interacting condensate is given in terms of an effective theory describing the probability distribution of the number of condensed particles. Second, the free energy of the thermally excited particles equals that of a temperature-dependent Bogoliubov Hamiltonian. We extend this result to a more general class of interaction potentials, including interactions with a hard core. Our proof follows a different approach than the one in \cite{BocDeuStock2023}: we model microscopic correlations between the particles by a Jastrow factor, and exploit a cancellation in the computation of the energy that emerges due to the different length scales in the system.   
\end{abstract}

\setcounter{tocdepth}{3}
\tableofcontents

\section{Introduction and main result}
\label{sec:model}

\subsection{Background and summary}

Since the first experimental realizations of Bose--Einstein condensation (BEC) in cold alkali gases in 1995 \cite{Anderson95, Ketterle95}, the dilute Bose Gas has become a prominent topic of experimental and theoretical research. The most relevant parameter regime to describe experiments with trapped quantum gases theoretically is the Gross--Pitaevskii (GP) limit. Here the scattering length of the interaction potential is scaled with the number of particles $N$ in such a way that the interaction energy per particle is of the same order of magnitude as the spectral gap in the trap, as $N \to \infty$. Many rigorous mathematical results about the GP limit of interacting Bose gases have been proved over the past twenty years. In the foundational works \cite{LiebSeiringer2002, LiebSeYng2000, LiebSeYng2002} it was shown that the ground state energy per particle can be approximated by the minimum of the GP energy functional, and that approximate ground states display BEC and superfluidity. These results have later been extended in \cite{LiebSeiringer2006, NamRouSei2016, Seiringer2003} to the case of rotating Bose gases.  

Condensation with an optimal rate was, in the GP limit, first proven in \cite{BBCS_OptimalRate2020} for approximate ground states of a Bose gas captured in a three-dimensional flat torus. In \cite{BBCS_Specturm2019} the same authors show that the second-order correction to the ground state energy, the low-lying eigenvalues of the many-body Hamiltonian and the corresponding eigenfunctions are well approximated by related quantities of a quadratic Hamiltonian, called Bogoliubov Hamiltonian. This confirms predictions of Bogoliubov from 1947 \cite{Bogoliubov47}. Similar results have later been obtained for the trapped Bose gas in the GP limit \cite{BrenSchleinSchra_OptimalRate2022, BreSchleiSchra_Spectrum2022}, for the homogeneous gas in a Thomas--Fermi limit \cite{AdiBreSchlein2021, BreCapSchlein2022}, and for the homogeneous gas in the GP limit in two space dimensions \cite{CaraCenaSchlein2021, CaraCenaSchlein2023}. More recently, a second-order upper bound for the ground state energy of a hard sphere Bose gas has been proven in \cite{basti2022second}. The homogeneous gas in a box with Neumann boundary conditions has been studied in \cite{BoccaSei2023}. 

While low-energy eigenstates of the Hamiltonian accurately describe the dilute Bose gas at (or near) zero temperature, understanding the system at positive temperature is crucial to describe modern experiments. In this setting the natural analogues of the ground state energy and its corresponding eigenfunction are the free energy and the Gibbs state associated to the many-body Hamiltonian. In the article \cite{DeuSeiYng2019} the trapped Bose gas is studied in a combination of thermodynamic limit in the trap and GP limit. It is proven that the free energy of the system minus that of the ideal gas is well approximated by the minimum of a GP energy functional. Moreover, the one-particle density matrix of any approximate minimizer of the free energy is, to leading order, given by the one of the ideal gas, where the condensate wavefunction has been replaced by the minimizer of the GP energy functional. This, in particular, establishes the existence of a BEC phase transition in the system. Comparable results for the homogeneous Bose gas have been obtained in \cite{DeuchertSeiringerHom2020}.

The GP limit is appropriate to describe experiments with atomic clouds containing $10^2-10^6$ particles. To describe truly macroscopic samples with particle numbers of the order of the Avogadro constant $N_{\mathrm{A}} \approx 6.022 \times 10^{23}$, one needs to consider a thermodynamic limit followed by a dilute limit. The leading order asymptotics of the ground state energy per particle in this regime has been established in the influential works \cite{Dyson57, LiebYng98} (three space dimensions) and \cite{LiebYng2001} (two space dimensions). The one-dimensional case has been studied in \cite{AgerReuSolo2022}. Recently, also the second-order correction predicted by Lee, Huang and Yang (LHY) in 1957 \cite{LHY57} could be justified, see \cite{YY2009, BaCeSchlein2021} for upper bounds and \cite{FouSol2020, FouSolo2023} for matching lower bounds. It is interesting to note that, to this date, there is no upper bound available that captures the LHY correction for a gas of hard spheres (the lower bound in \cite{FouSolo2023} applies in this case). A comparable second order expansion for the two-dimensional Bose has been obtained in \cite{FGJMO2022}. For the dilute Bose gas at positive temperature, asymptotic expansions capturing the leading order correction to the free energy caused by the interaction between the particles have been proved in \cite{Seiringer2008, Yin2010} (three space dimensions) and \cite{DeuMaySei2020, MaySei2020} (two space dimensions). In \cite{HabHaiNamSeiTriay2023} a LHY-type lower bound for the three-dimensional gas at suitably low temperatures is established. 

In the recent work \cite{BocDeuStock2023} the authors consider a grand canonical homogeneous Bose gas in the GP limit at temperatures that are comparable to the critical temperature for BEC in the ideal gas\footnote{The critical temperature in the interacting gas is expected to be the same, to leading order in $N$, but this has so far been proven only in the canonical setting, see \cite{DeuchertSeiringerHom2020}.}. Under the assumption that the interaction potential is of class $L^3$, they establish an upper bound for the grand canonical free energy that contains two novel contributions: the free energy of the interacting condensate is given in terms of an effective theory describing its particle number fluctuations. Moreover, the free energy of the thermally excited particles equals that of a temperature-dependent Bogoliubov Hamiltonian. In the present article, we extend this result to systems with interactions in a more general class, including the hard-core potential. Our proof is based on the use of a trial state that is similar in spirit to the one in \cite{BocDeuStock2023}. However, due to the lack of regularity of the interaction potential, we are forced to implement microscopic correlations between the particles by a full Jastrow factor. Because of this, the computation of the energy requires different arguments. A crucial step in our proof is a cancellation in the computation of the energy that emerges due to the different length scales in the problem.

\subsection{Notation}
Given two functions $a, b$ of the particle number $N$ and other parameters of the system, we write $a \lesssim b$ if there exists a constant $C$, independent of $N$, such that $a \le C b$. If we want to highlight the dependency of the constant on some ($N$-independent) parameter $k$, we use the notation $a \lesssim_k b$. We write $a \sim b$ if $a \lesssim b$ and $b \lesssim a$, and $a\simeq b$ means that $a/b \to 1$ in the limit considered. The letters $C, c$ denote generic positive constants, whose values may change from line to line.

The Fourier coefficients of a function $f : \Lambda = [-1/2, 1/2]^3 \to \mathbb C$ are denoted by $\hat f(p) = \int_{\Lambda} e^{-\mathrm{i} p\cdot x} f(x) \, \mathrm dx$ and, given a sequence $g: \Lambda^* = 2\pi\mathbb Z^3 \to \mathbb C$, the inverse Fourier transformation reads $\check g(x) = \sum_{p\in\Lambda^*} g(p) e^{\mathrm{i} p\cdot x}$. Standard $L^p(\Lambda)$ and $\ell^p(\Lambda^*)$ norms are denoted by $\|\cdot\|_p$. If $H$ is a (separable, complex) Hilbert space we denote by $\langle \cdot, \cdot \rangle$ its inner product, and by $\mathcal L^1(H)$ the space of trace-class operators on $H$. If $A$ is an operator on $H$ and $\psi\in H$ belongs to the form domain of $A$, we use the notation $\langle A\rangle_{\psi} = \langle \psi, A\psi\rangle$.

\subsection{The grand canonical free energy and the Gibbs variational principle}

We consider a Bose gas captured in the three-dimensional box $\Lambda = [-1/ 2 , 1/ 2 ]^3$ with periodic boundary conditions. Since we are interested in a system with a fluctuating particle number, its Hilbert is given by the bosonic Fock space 
\begin{equation}
	\label{eq:Fock_def}
	\mathfrak F = \bigoplus_{n\ge 0} L^2_{\mathrm s}(\Lambda^n).
\end{equation} 
Here $L^2_{ \mathrm s}(\Lambda^n)$ denotes the space of permutation symmetric functions in $L^2(\Lambda^n)$. That is, the closed linear subspace of $L^2(\Lambda^n)$, whose elements $\Psi(x_1,...,x_n)$ are invariant under any permutation of its $n$ particle coordinates $x_1,...,x_n$. 

The Hamiltonian of the system in the GP scaling reads
\begin{equation}
	\label{eq:Hamiltonian_second_quant}
	\mathcal H_N =  \bigoplus_{n\ge 0} H_N^{(n)}, \quad \text{ with } \quad H_N^{(n)} = - \sum_{i = 1}^n \Delta_i + N^2 \sum_{1 \le i < j \le n}V(N(x_i-x_j)).
\end{equation} Here, $\Delta_i$ denotes the Laplacian with periodic boundary conditions acting on the $i$-th coordinate, and $V(N(x_i-x_j))$ is a multiplication operator. We assume that the interaction potential $V:\mathbb R^3\to [0,\infty]$ is measurable, radial and compactly supported. The parameter $N$ will be chosen such that it coincides with the expected number of particles in the system. Our assumptions on $V$ guarantee that its scattering length $\mathfrak a$ is well-defined: this is a combined measure of the range and strength of the interaction potential $V$. For a precise definition of the scattering length, we refer the reader to Appendix~\ref{app:scattering}. By scaling, the scattering length $\mathfrak a_N$ of $V_N = N^2 V(N\cdot)$ satisfies $\mathfrak a_N = \mathfrak a / N$. The fact that $V\ge 0$ allows us to define the Hamiltonian $\mathcal H_N$ in \eqref{eq:Hamiltonian_second_quant} as a self-adjoint operator via Friedrichs extension. 

For states $\Gamma \in \Sigma \coloneqq \{ \Gamma\in \mathcal L^1(\mathfrak F) \mid \Gamma\ge 0, ~ \tr\Gamma=1 \}$ we define the Gibbs free energy functional $\mathcal F(\cdot)$ by\footnote{Here and in the following, we interpret $\tr[AB]$ for positive operators $A$ and $B$ as $\tr[A^{1/2}BA^{1/2}]$. By positivity, this expression is always well defined and takes values in $[0, \infty]$. In particular, finiteness of $\tr[AB]$ under this convention requires only that the operator $A^{1/2}BA^{1/2}$ is trace-class, and not necessarily $AB$.}
\begin{equation}
	\mathcal{F}(\Gamma) = \tr [ \mathcal{H}_N \Gamma ] - \beta^{-1} S(\Gamma),
	\label{eq:Gibbsfreenergyfunctional}
\end{equation}
where $ S(\Gamma) = - \tr[\Gamma \ln (\Gamma)] $ denotes the von Neumann entropy and $\beta > 0$ is the inverse temperature of the system. The grand canonical free energy is defined as the minimum of the Gibbs free energy functional among states with expected number of particles equal to $N$:
\begin{equation}
	\label{eq:Interacting_Free_Energy}
	F(\beta, N) = \min \left \{ \mathcal F(\Gamma) \mid \Gamma \in \Sigma, ~ \tr[\mathcal N \Gamma] = N  \right \} = - \beta^{-1} \log \tr [ \exp(-\beta (\mathcal{H}_N - \mu \mathcal{N}) ) ] + \mu N.
\end{equation} Here $\mathcal N = \bigoplus_{n\ge 0} n$ denotes the number of particles operator on $\mathfrak F$. The minimum in \eqref{eq:Interacting_Free_Energy} is attained uniquely at the Gibbs state 
\begin{equation}\label{eq:Gibbs_state_int}
	G = \frac{ \exp(- \beta (\mathcal{H}_N - \mu \mathcal{N}) )}{\tr [ \exp(- \beta (\mathcal{H}_N - \mu \mathcal{N}) ) ]} ,
\end{equation}
where the chemical potential $\mu = \mu(N, \beta) \in \mathbb{R}$ is defined implicitly by the equation $\tr [ \mathcal{N} G ] = N$. 

\subsection{The ideal Bose gas on the torus}

Before we state our main result, we recall some well-known facts about the ideal Bose gas on the unit torus, that is, the system described by the Hamiltonian in \eqref{eq:Hamiltonian_second_quant} with $V=0$. In this case the chemical potential $\mu_0 = \mu_0(\beta, N)< 0$ can be defined by the equation
\begin{equation} \label{eq:mu0_definition_implicit}
	N = \sum_{p\in\Lambda^*} \frac{1}{\exp( \beta(|p|^2 - \mu_0(\beta, N))) -1 }.
\end{equation} The expected number of particles with momentum $p = 0$ is given by 
\begin{equation}\label{eq:n0_def}
	 N_0(\beta, N) = \frac 1 {\exp(-\beta\mu_0) -1 } 
\end{equation} and satisfies
\begin{equation} \label{eq:Tc_ideal}
	\frac{N_0(\beta, N)}{N} \simeq \left[ 1- \frac{\beta_\mathrm{c}}{\beta} \right]_+, \quad \text{ with } \quad \beta_{\mathrm{c}} = \frac{1}{4\pi}\left( \frac{N}{\zeta(3/2)} \right)^{-2/3}
\end{equation} in the limit $N\to\infty$. Here $\zeta$ denotes the Riemann zeta function and $[\cdot]_+ = \max \{ \cdot, 0 \}$. Equation \eqref{eq:Tc_ideal} implies that the ideal gas displays a BEC phase transition with critical inverse temperature $\beta_{\mathrm{c}}$. More precisely, if $\beta = \kappa \beta_{\mathrm{c}}$ with $\kappa \in (1,\infty)$, we have $N_0 \simeq N (1-1/\kappa)$ and $|\mu_0| \sim N^{-1/3}$. If, in contrast, $\kappa \in (0,1)$ then $N_0 \sim 1$ and $|\mu_0| \sim N^{2/3}$. 

The free energy of the ideal Bose gas reads $F_0(\beta, N) = F_0^{\mathrm{BEC}} + F_0^+$, where
\begin{equation}\label{eq:free_energy_ideal_condensate}
		F_0^{\mathrm{BEC}} = \frac{1}{\beta} \log( 1- \exp(\beta\mu_0) ) + \mu_0 N_0
\end{equation} denotes the free energy of the condensate and
\begin{equation}\label{eq:free_energy_ideal_cloud}
	F_0^{+} = \frac{1}{\beta} \sum_{p\in\Lambda_+^*} \log( 1- \exp( - \beta( |p|^2 - \mu_0 ) ) ) + \mu_0 ( N - N_0 )
\end{equation} that of the thermally excited particles.

\subsection{Main results}

The following theorem is the main result of this article. 
\begin{theorem}\label{thm:main_thm}
	Let $V:\mathbb R^3 \to [0,\infty]$ be measurable, spherically symmetric and compactly supported. In the limit $N\to\infty$, with $\beta = \kappa \beta_{\mathrm{c}}$, $\kappa\in(0,\infty)$ and $\beta_{\mathrm{c}}$ in \eqref{eq:Tc_ideal}, the free energy in \eqref{eq:Interacting_Free_Energy} satisfies 
	\begin{equation}\label{eq:main_result}
\begin{split}
			F(\beta, N) \le & F_0^+(\beta, N) + 8\pi \mathfrak{a}_N N^2 + \min\{ F^{\mathrm{BEC}} - 8\pi \mathfrak{a}_N N_0^2, F_0^{\mathrm{BEC}} \} \\
			& - \frac{1}{2\beta} \sum_{p\in\Lambda_+^*} \left[ \frac{16\pi \mathfrak{a}_N N_0 }{ |p|^2 } - \log\left( 1+ \frac{16\pi \mathfrak{a}_N N_0 }{ |p|^2 }   \right) \right] + \mathcal{O}(N^{11/18}),
\end{split}
	\end{equation} with $N_0$, $F_0^{\mathrm{BEC}}$ and $F_0^+$ defined respectively in \eqref{eq:n0_def}, \eqref{eq:free_energy_ideal_condensate} and \eqref{eq:free_energy_ideal_cloud}, and
\begin{equation}\label{eq:free_energy_interacting_condensate}
	F^{\mathrm{BEC}} = 	F^{\mathrm{BEC}} (\beta, N_0, \mathfrak{a}_N) = -\frac{1}{\beta} \log \left( \int_{\mathbb C} \exp\left( -\beta \left( 4\pi \mathfrak{a}_N |z|^4 - \mu|z|^2 \right) \right) \,\mathrm d z \right) + \mu N_0(\beta, N).
\end{equation} Here, $\mathrm d z = \pi^{-1} \mathrm d x \, \mathrm d y$, where $\mathrm dx \, \mathrm dy$ denotes the Lebesgue measure on $\mathbb C$, and $\mu$ is chosen as the unique solution of the equation 
\begin{equation}\label{eq:mu_definition}
	\int_{\mathbb C} |z|^2 g(z) \, \mathrm d z= N_0(\beta, N),
\end{equation} with the probability density 
\begin{equation}\label{eq:condensate_density}
	g(z) = \frac{ \exp\left( -\beta \left( 4\pi \mathfrak{a}_N |z|^4 - \mu|z|^2 \right) \right) }{ \int_{\mathbb C} \exp\left( -\beta \left( 4\pi \mathfrak{a}_N |z|^4 - \mu|z|^2 \right) \right)\,\mathrm d z }. 
\end{equation} 
\end{theorem}

 The terms on the right-hand side of \eqref{eq:main_result} appear in descending order according to their order of magnitude in the limit $N\to\infty$. The free energy $F_0^+(\beta, N)$ of the thermal cloud of the ideal gas is proportional to $N^{5/3}$. The second term is a density-density interaction, which is of order $N$. The third term represents the free energy of the interacting condensate. If $\kappa >1$, it contributes two terms, one of order $N$ and one of order $N^{2/3} \log N$. For $\kappa<1$ it is proportional to $N^{2/3}$. Finally, the term on the second line of \eqref{eq:main_result} is a correction to the free energy of the thermally excited particles coming from Bogoliubov theory. It is of order $N^{2/3}$ in the presence of a macroscopic condensate occupation ($\kappa>1$), and of order $ N^{-4/3}$ if $\kappa<1$. More details concerning the last two terms can be found in Remark~\ref{rmk:main_result} below.

The following proposition, which is proved in \cite[Proposition 1.2]{BocDeuStock2023}, allows us to simplify the right-hand side of \eqref{eq:main_result} in the parameter regimes strictly above and strictly below the critical point.

\begin{proposition}\label{prop:f_bec}
	We consider the limit $N\to\infty$, with $\beta = \kappa \beta_{\mathrm{c}}$, $\kappa\in(0,\infty)$ and $\beta_{\mathrm{c}}$ in \eqref{eq:Tc_ideal}. The following statements hold for given $\epsilon > 0$:
	\begin{enumerate}
		\item Assume that $N_0 \gtrsim N^{5/6 + \epsilon}$ and that $\mathfrak{a}_N > 0$. There exists a constant $c>0$ such that
		\begin{equation}\label{eq:f_bec_expansion_condensed}
			F^{ \mathrm{ BEC } } = 4\pi \mathfrak a_N N_0^2 + \frac 1 {2\beta} \log( 4 \beta \mathfrak a_N ) + \mathcal O( \exp(-cN^{\eps})).
		\end{equation}
	\item If $N_0 \lesssim N^{5/6 - \eps}$, then 
	\begin{equation}\label{eq:f_bec_expansion_noncondensed}
			F^{ \mathrm{ BEC } } = - \frac{1}{\beta} \log(N_0) - \frac 1 \beta + \mathcal O( N^{2/3 - 2\eps} ).
	\end{equation}
	\end{enumerate}
\end{proposition}

Proposition \ref{prop:f_bec} describes a transition in the behavior of the effective theory of the interacting condensate. If $N_0 \gtrsim N^{5/6 + \eps}$, the free energy in \eqref{eq:free_energy_interacting_condensate} is given, up to a small remainder, by the usual density-density interaction plus a contribution of order $N^{2/3} \log(N)$, which is related to the free energy of the fluctuations of the number of condensed particles. We refer to part 4 of Remark~\ref{rmk:main_result} for further details. Both contributions are caused by the self-interaction of the condensate, as is evident from their dependence on the scattering length $\mathfrak a_N$. If instead $1 \ll N_0 \lesssim N^{5/6- \eps}$, the free energy of the condensate \eqref{eq:f_bec_expansion_noncondensed} equals the one of its non-interacting counterpart, up to $o(N^{2/3})$. The threshold arises from the fact that, when $N_0\sim N^{5/6}$, the interaction energy $ 4\pi \mathfrak a_N N_0^2 \sim N^{2/3}$ of the condensate becomes much smaller than $\beta^{-1}$ times the classical entropy $S^{\mathrm{cl}}$ of $g(z)$ (see \eqref{eq:free_energy_functional_condensate}), which is always of order $N^{2/3} \log N$ if $N^{\eps}\lesssim N_0 \lesssim N^{5/6}$. In the transition regime $N^{5/6 - \eps} \lesssim N_0 \lesssim N^{5/6 + \eps}$, the free energy of the condensate does not have a simple form as in \eqref{eq:f_bec_expansion_condensed} or \eqref{eq:f_bec_expansion_noncondensed}. 

With Proposition \ref{prop:f_bec} at hand, one readily checks that the minimum on the right-hand side of \eqref{eq:main_result} is attained by the first term if $\kappa\in(1, \infty)$ (condensed phase) and by the second if $\kappa\in(0,1)$ (non-condensed phase). This leads to the following reformulation of Theorem \ref{thm:main_thm}, which is better suited for a comparison with the existing literature.

\begin{corollary}\label{cor:main_result_simplified}
	Let $V:\mathbb R^3 \to [0,\infty]$ be a measurable, spherically symmetric and compactly supported function that is strictly positive on a set of positive measure. We consider the limit $N\to\infty$, with $\beta = \kappa \beta_{\mathrm{c}}$, $\kappa\in(0,\infty)$ and $\beta_{\mathrm{c}}$ in \eqref{eq:Tc_ideal}. If $\kappa\in(1,\infty)$, the free energy \eqref{eq:Interacting_Free_Energy} satisfies 
	\begin{equation}\label{eq:main_result_condensed}
		\begin{split}
			F(\beta, N) \le & F_0^+(\beta, N) + 4\pi \mathfrak{a}_N ( 2 N^2 - N_0^2) + \frac 1 {2\beta} \log( 4 \beta \mathfrak a_N )  \\
			& - \frac{1}{2\beta} \sum_{p\in\Lambda_+^*} \left[ \frac{16\pi \mathfrak{a}_N N_0 }{ |p|^2 } - \log\left( 1+ \frac{16\pi \mathfrak{a}_N N_0 }{ |p|^2 }   \right) \right] + \mathcal{O}(N^{11/18})
		\end{split}
	\end{equation} with $N_0$ and $F_0^+$ defined in \eqref{eq:n0_def} and \eqref{eq:free_energy_ideal_cloud}, respectively. If $\kappa\in(0,1)$ we have 
\begin{equation}\label{eq:main_result_non_condensed}
		\begin{split}
			F(\beta, N) \le & F_0(\beta, N) + 8\pi \mathfrak{a}_N  N^2 + \mathcal{O}(N^{11/18}),
\end{split}
\end{equation} with the free energy of the ideal gas $F_0(\beta, N)$ above \eqref{eq:free_energy_ideal_condensate}.
\end{corollary}

At the critical point, corresponding to $\kappa = 1$, Proposition \ref{prop:f_bec} does not apply, and the minimum in \eqref{eq:main_result} is needed. We have the following remarks concerning the above results.

\begin{remark}\label{rmk:main_result}
	\begin{enumerate}
		\item The first two terms in \eqref{eq:main_result_condensed} were first identified for the dilute Bose gas in the thermodynamic limit in \cite{Yin2010} (upper bound) and \cite{Seiringer2008} (lower bound). An asymptotic expansion for the canonical free energy of the Bose gas in the GP limit was given, up to remainders of order $o(N)$, in \cite{DeuchertSeiringerHom2020}. The same expansion is, however, expected to hold in the grand canonical setting, and it coincides with the first two terms on the right-hand sides of \eqref{eq:main_result_condensed} and \eqref{eq:main_result_non_condensed}. The first upper bound capturing the third and fourth term on the right-hand side of \eqref{eq:main_result_condensed} was proved in \cite{BocDeuStock2023} for Bose gases interacting through sufficiently regular interaction potentials (of class $L^3$). In contrast to \cite{BocDeuStock2023}, we make no such regularity assumption, and our result applies, in particular, to the case of the hard-core interaction 
		\begin{equation} \label{eq:hardcore}
			V(x) = \begin{cases} 
				+\infty & \text{if } |x| \le \mathfrak a, \\
				0 & \text{otherwise}. 
			\end{cases}
		\end{equation} This generalization is the main contribution of the present article.
	
	\item Our proof of Theorem~\ref{thm:main_thm} is based on a trial state that is similar to the one used in \cite{BocDeuStock2023}. In particular, we use (up to technicalities) the same uncorrelated trial state. However, in the absence of integrability assumptions on the interaction potential, we need to describe the correlations between particles by a full Jastrow factor.  As a consequence our proof of the upper bound in \eqref{eq:main_result} is not an adaption of the one in \cite{BocDeuStock2023} and requires different arguments. One key step in our proof is a cancellation in the computation of the interaction energy that is similar in spirit to a cancellation observed in \cite{BCOPS2023} in the computation of the Lee--Huang--Yang correction to the ground state energy of the hard-sphere gas. To see this cancellation, we exploit the fact that the interaction between the particles lives on a much smaller length scale than the thermal wavelength $\beta^{1/2}$. Moreover, precise pointwise bounds for the reduced densities of our trial state without correlations and of its eigenfunctions are needed. In contrast, in \cite{BocDeuStock2023} it was possible to implement correlations between the particles with a truncated quartic (in creation and annihilation operators) transformation in Fock space. In combination with the use of suitable momentum cutoffs in the trial state, this allowed the authors of \cite{BocDeuStock2023} to obtain an upper bound for the free energy in a more direct way.
	\item The third term on the right-hand side of \eqref{eq:main_result_condensed} is the free energy of the fluctuations of the number of particles in the condensate. To explain this, we describe the condensate with a trial state of the form 
	\begin{equation}
		G_0 = \int_{\mathbb C} | z \rangle \langle z | \rho(z) \, \mathrm dz,
		\label{eq:A1}
	\end{equation} where $	| z \rangle = \exp( z a_0^* - \overline z a_0 ) \Omega$ is the usual coherent state on the $p=0$ mode and $\rho(z)$ is a probability density with respect to the measure $\mathrm d z$ introduced below \eqref{eq:free_energy_interacting_condensate}. This is motivated by the fact that a $c$-number substitution for one momentum mode is known to introduce only small corrections to the free energy, see for instance \cite{LSY2005}. If we take $4\pi \mathfrak a_N a_0^* a_0^* a_0 a_0$ as the effective interaction Hamiltonian of the condensate (i.e., we replace the potential by a renormalized one proportional to the scattering length), we can write the free energy of $G_0$ as
	\begin{equation}
	\label{eq:free_energy_functional_condensate}
	\mathcal F^{\mathrm{BEC}}(G_0) = 4\pi \mathfrak{a}_N  \int_{\mathbb C} \rho(z)|z|^4\, \mathrm d z - \frac 1 \beta S(G_0) \le 4\pi \mathfrak{a}_N  \int_{\mathbb C} \rho(z)|z|^4\, \mathrm d z - \frac 1 \beta S^{\mathrm{cl}}(\rho),
\end{equation} where $S^{\mathrm{cl}}(\rho) = - \int_{\mathbb C} \rho(z) \log( \rho(z) ) \, \mathrm d z $ denotes the classical entropy of $\rho$. The inequality in \eqref{eq:free_energy_functional_condensate} is a consequence of the Berezin-Lieb inequality, see \cite{Berezin1972, Berezin75, Lieb73}. If we minimize the right-hand side of \eqref{eq:free_energy_functional_condensate} under the constraint $\int_{\mathbb C} |z|^2 \rho(z)\, \mathrm d z = N_0$, we find $F^{\mathrm{BEC}}$ in \eqref{eq:free_energy_interacting_condensate}, with the unique minimizer $g(z)$ in \eqref{eq:condensate_density}. Using Proposition \ref{prop:f_bec}, we thus see that 
\begin{equation}\label{eq:fluctuation_term}
	\frac 1 {2\beta} \log(4\beta \mathfrak a_N) = 4\pi a_N \Big( \int_{\mathbb C} |z|^4 g(z) \, \mathrm d z - \Big(  \int_{\mathbb C} |z|^2 g(z)\, \mathrm dz \Big)^2 \Big) -\frac{1}{\beta} S^{ \mathrm{ cl } }(g) + \mathcal O(\exp(-cN^{\eps}))
\end{equation}
	if $N_0\gtrsim N^{5/6+\eps}$ for some $\eps>0$. That is, according to this effective theory, the third term on the right-hand side of \eqref{eq:main_result_condensed} indeed equals the free energy of the fluctuations of the number of particles in the condensate. 
	\item While the variance of the number of particles in the condensate of the ideal gas is of order $N^2$, for the Gibbs distribution $g$ we have 
	\begin{equation}\label{eq:g_variance}
		\mathrm{Var}_g(|z|^2) = \int_{\mathbb C} |z|^4 g(z) \, \mathrm d z - \Big(  \int_{\mathbb C} |z|^2 g(z) \Big)^2 \sim N^{5/3},
	\end{equation} provided $\kappa>1$. This decrease of the fluctuations of the number of condensed particles caused by the repulsive interaction between them is a well-known effect, see e.g. \cite{BuffetPule83, Davies72}. Motivated by the recent experimental realization \cite{SchmittWeitz2014} of a system with grand canonical number statistics, a discrete version of $g$ has been used in \cite{WurffStoof2014} to compute the size of the fluctuations of the number of condensed particles for a trapped gas. The computations in \cite{WurffStoof2014} could be rigorously justified by showing that $g(z)$ approximates $\tr [ |z \rangle \langle z | G]$, where $|z \rangle$ is the coherent state defined below \eqref{eq:A1} and $G$ is the interacting Gibbs state in \eqref{eq:Gibbs_state_int}. This interesting mathematical problem is, however, beyond the scope of our present investigation.
\item The term on the second line of \eqref{eq:main_result} is a correction to the free energy of the thermally excited particles, which is related to Bogoliubov theory. This can be seen with the following heuristic computation. In the first step we assume, for the sake of simplicity, that $V \in L^1(\Lambda)$. We start by writing the Hamiltonian in \eqref{eq:Hamiltonian_second_quant} in terms of
the usual creation and annihilation operators $a_p^*, a_p$ of a particle with momentum $p \in \Lambda^*$ as
\begin{equation*}
	\mathcal H_N = \sum_{p\in\Lambda^*_+} |p|^2 a_p^* a_p + \sum_{p,q,r\in\Lambda^*} \widehat V_N(r) a_{p+r}^* a_q^* a_p a_{q+r}.
\end{equation*} Replacing $a_0^*, a_0$ with $\sqrt N_0$, the potential $\widehat V_N(r)$ with its renormalized version $4\pi \mathfrak a_N$, and neglecting cubic and quartic terms in $a_p^*, a_p$, we arrive at the Bogoliubov Hamiltonian
\begin{equation*}
	\mathcal H^{\mathrm{Bog}} = \sum_{p\in\Lambda^*_+} |p|^2 a_p^* a_p + 4\pi \mathfrak a_N N_0 \sum_{p \in \Lambda_+^*} \Big( 2 a_p^* a_p + a_p^* a_{-p}^* + a_p a_{-p} \Big).
\end{equation*} A careful analysis shows that the grand potential $\Phi^{ \mathrm{ Bog } }(\beta, \mu_0)$ associated to $\mathcal H^{ \mathrm{ Bog} } - \mu_0 \mathcal{N}$ with the chemical potential $\mu_0$ in \eqref{eq:mu0_definition_implicit} satisfies (compare to Lemma \ref{lem:bogoliubov_free_energy_expansion})
\begin{equation*}
\begin{split}
		\Phi^{ \mathrm{ Bog } }(\beta, \mu_0) = & \frac 1 \beta \sum_{ p\in\Lambda_+^* }  \log \Big( 1 - e^{ \beta \sqrt{ |p|^2 - \mu_0} \sqrt{ |p|^2 - \mu_0 + 16\pi\mathfrak a_NN_0 } } \Big)\\
		=  & \frac 1 \beta \sum_{ p\in \Lambda_+^* } \log \big( 1- e^{-\beta(|p|^2-\mu_0)}  \big) + 8\pi \mathfrak a_N N_0( N - N_0 ) \\
	& - \frac{1}{2\beta} \sum_{ p\in \Lambda_+^*} \Bigg[ \frac{16\pi  \mathfrak a_N N_0}{|p|^2} - \log\Big( 1 +  \frac{16\pi \mathfrak a_N N_0}{|p|^2}   \Big) \Bigg] + o(N^{2/3}).
\end{split}
\end{equation*}  The first term on the right-hand side contributes to $F_0^+$, the second term to the density-density interaction $4\pi \mathfrak a_N(2N^2-N_0^2)$, and the third term appears on the second line of \eqref{eq:main_result}. 
\item Let us denote by $H_N^{(N)}$ the restriction of $\mathcal H_N$ to the $N$-particle sector of Fock space (see \eqref{eq:Hamiltonian_second_quant}), and by $E_0$ the ground state of $H_N^{(N)}$. It has been shown in \cite{BBCS2020} that the eigenvalues $E$ of $H_N^{(N)} - E_0$ that satisfy $E \ll N^{1/8}$ are well approximated by those of a Bogoliubov Hamiltonian. If we compare this threshold to the energy scale $\beta^{-1} \sim N^{2/3}$, which represents the energy per particle in our system, we see that the results in \cite{BBCS2020} are far from being sufficient to draw conclusions on the free energy in our setting. 
\item If we replace the torus $\Lambda = [-1/2,1/2]^3$ by $\Lambda_{L} = [-L/2,L/2]^3$ with fixed $L>0$, Theorem~\ref{thm:main_thm} and a scaling argument imply a similar upper bound for the grand canonical free energy in this setting. In this case, the term on the second line of \eqref{eq:main_result} reads 
	\begin{equation*}\label{eq:main_result_L}
	\begin{split}
		- \frac{1}{2\beta} \sum_{p\in \frac {2\pi} L \mathbb Z^3\setminus \{0\} } \left[ \frac{16\pi \mathfrak{a}_N \rho_0(\beta, N, L) }{ |p|^2 } - \log\left( 1+ \frac{16\pi \mathfrak{a}_N \rho_0(\beta, N, L) }{ |p|^2 }   \right) \right] .
	\end{split}
\end{equation*} If we replace $\mathfrak a_N$ by $\mathfrak a$, divide the term in \eqref{eq:main_result_L} by $L^3$ and take a formal thermodynamic limit (i.e. letting $N, L\to\infty$ with $\rho = N/ L^3$ fixed), we obtain
\begin{equation*}
	-\frac{1}{2\beta ( 2\pi )^3 } \int_{\mathbb R^3} \left[ \frac{16\pi \mathfrak{a} \rho_0 }{ |p|^2 } - \log\left( 1+ \frac{16\pi \mathfrak{a} \rho_0 }{ |p|^2 }  \right) \right]  \, \mathrm d p = -\frac {16\sqrt \pi} {3\beta} (\mathfrak a\rho_0)^{3/2}.
\end{equation*} The right-hand side has been conjectured to appear in the asymptotic expansion of the specific free energy in the dilute limit, see \cite[Theorem 11]{NapReuSol2017}.
\item The minimum on the right-hand side of \eqref{eq:main_result} is needed, because $F^{ \mathrm{ BEC } }$ does not accurately describe the free energy of the condensate if $N_0 \sim 1$, see \eqref{eq:bounds_minimum}. This is related to the fact that we approximate the discrete random variable associated with the operator $a_0^* a_0$ with a continuous one. 
\item We state and prove Theorem \ref{thm:main_thm} with $\kappa \in (0, \infty)$ fixed. However, a straightforward adaptation of our proof applies to the case in which $\kappa$ depends on $N$, as long as $\kappa \gtrsim 1$. In particular, it is possible to take a zero-temperature limit (corresponding to $\kappa\to\infty$). 
\item We expect the upper bound in Theorem \ref{thm:main_thm} to be sharp. That is, we expect it to be possible to prove a matching lower bound, up to remainders of order $o(N^{2/3})$.
\item A similar expansion as in \eqref{eq:main_result_condensed} is expected to hold for the interacting canonical free energy if $F_0^+$ is replaced by the canonical free energy of the ideal gas and $F^{ \mathrm{ BEC } }$ by $4\pi \mathfrak a_N N_0^2$. The reason for the latter replacement lies in the fact that, in the canonical ensemble, the variance of the number of particles in the condensate is expected to be of the order $N^{4/3}$ if $\beta = \kappa \beta_{\mathrm{c}}$ with $\kappa>1$ (This is the order of magnitude of these fluctuations in the ideal gas.). When we compare this to \eqref{eq:fluctuation_term} and \eqref{eq:g_variance}, we see that these fluctuations are too small (when compared to $\beta^{-1}$ times the entropy) to generate a contribution to the free energy of the order $N^{2/3}$ or $N^{2/3} \ln(N)$. We refer the reader to \cite{ChaDia2014} for a detailed analysis of the condensate fluctuations in the canonical ideal gas.
\end{enumerate}
\end{remark}

\subsection{Outline of the article}
To prove Theorem~\ref{thm:main_thm}, we apply a trial state argument with two distinct trial states corresponding to the regimes of high and low occupation of the condensate, respectively. The analysis of the former parameter regime is considerably more difficult, and we therefore focus on it. The adaption of the proof to the (simpler) case of low condensate occupation is discussed at the end of Section~\ref{sec:proof_main_theorem}.

In Section~\ref{sec:trial_state} we define the trial state and we prove some of its properties that are needed for the computation of the free energy. In particular, we prove pointwise bounds for the two- and four-body reduced densities of our trial state and of its eigenfunctions. These estimates determine their leading order behavior as $N \to \infty$. In Section~\ref{sec:energy} we provide an upper bound for the energy. One main step in our proof is the use of the pointwise bounds for the reduced densities from Section~\ref{sec:trial_state} to establish in Section~\ref{sec:analysisoftherenormalizedinteraction} a cancellation between the numerator and the denominator of the effective interaction energy. An upper bound for the entropy is provided in Section~\ref{sec:entropy}. In Section~\ref{sec:proof_main_theorem} we collect the results from Sections~\ref{sec:energy} and \ref{sec:entropy} and give a proof of Theorem \ref{thm:main_thm}. 

To not disrupt the main line of the argument, we defer some technical lemmas to the Appendix. In Appendix \ref{app:scattering} we recall some properties of the solution to the scattering equation. Appendix~\ref{app:condensate_functional} contains useful estimates for the effective chemical potential in the BEC. Finally, the expected number of particles in our trial state is computed in Appendix~\ref{app:perturbation_number_particles}. 

\section{The trial state} \label{sec:trial_state}

In this section we define our trial state, which consists of the following parts: (a) the Gibbs state of a temperature-dependent Bogoliubov Hamiltonian that describes the thermally excited particles, (b) a suitable convex combination of coherent states describing the BEC, and (c) a correlation structure given by a Jastrow factor. In Section~\ref{sec:PropertiesTrialState} we state and prove several lemmas that are needed in the computation of the free energy of our trial state. Before constructing the trial state, we recall some definitions concerning the formalism of second quantization, which also allows us to set some notation.

\subsection{Second quantization}

An important class of operators on $\mathfrak F$ is given by the creation and annihilation operators $a^*(f)$, $a(f)$ of a one-particle wave function $f$. They satisfy the canonical commutation relations (CCR) 
\begin{equation*}
	[ a(f), a(g)^* ] = \langle f, g\rangle, \quad \quad  [a(f),a(g)] = 0 = [a^*(f),a^*(g)]
\end{equation*}
for every $f, g\in L^2 (\Lambda)$. 
In the special case $f (x)= \ph_p(x) = e^{\mathrm{i}p\cdot x}$ with $p\in 2\pi \mathbb Z^3$ we write $a_p = a(\ph_p)$. We also introduce the operator-valued distributions $a_x^*, a_x$ creating and annihilating a particle at a point $x \in \mathbb{R}^3$, respectively, which satisfy the CCR $ [ a_x, a_y^* ] = \delta(x-y) $, $[ a_x, a_y ] = 0 = [ a_x^*, a_y^* ]$ for $x,y\in\Lambda$. Here $\delta(x)$ denotes Dirac's delta distribution with unit mass at the origin.

To be able to distinguish between the condensate and the thermally excited particles, we introduce the Fock spaces
\begin{equation*}
	\mathfrak F_0 = \bigoplus_{n\ge 0} \Span\{\varphi_0\} ^{\otimes^n},	\qquad \mathfrak{F}_+ = \bigoplus_{n\ge 0} L^2_\perp(\Lambda)^{\otimes_{\mathrm {s}}^n},
\end{equation*} where $\varphi_0$ is the (normalized) constant function on $\Lambda$, and $L^2_\perp(\Lambda)$ denotes the orthogonal complement of $\Span\{ \varphi_0 \}$ in $L^2(\Lambda)$. We denote by $\Omega_0$ and $\Omega_+$ the vacuum vectors in $\mathfrak F_0$ and $\mathfrak F_+$, respectively. The full Fock space can be identified with the tensor product $\mathfrak{F}_0\otimes\mathfrak{F}_+$ thanks to the unitary equivalence $\mathfrak F = \mathcal U (\mathfrak F_0 \otimes \mathfrak F_+)$ defined by $\Omega = \mathcal U (\Omega_0 \otimes \Omega_+)$, where $\Omega$ denotes the vacuum vector in $\mathfrak F$, and 
\begin{equation}\label{eq:definitionU}
	\begin{split}
		\mathcal U^* a(\varphi_0\oplus 0) \mathcal U = & a(\varphi_0) \otimes \mathds 1, \\
		\mathcal U^* a(0\oplus f) \mathcal U = & \mathds 1 \otimes a(f), 
	\end{split}
\end{equation} for every $f\in L^2(\Lambda)$.

\subsection{Definition of the trial state}

We are now prepared to give the definition of our trial state. On the excitation Fock space $\mathfrak F_+$, we define the Bogoliubov Hamiltonian
\begin{equation}\label{eq:bogoliubov_ham}
	\mathcal H_{\mathrm{B}} (z) = \sum_{p\in\Lambda^*_+} (|p|^2 -\mu_0) a_p^* a_p + 4\pi \mathfrak a_N N_0 \sum_{ p\in P_{\mathrm B} } \left[  2 a_p^* a_p + \frac{z^2}{|z|^2} a_p^* a_{-p}^* + \frac{\overline{z}^2}{|z|^2} a_{p} a_{-p} \right]
\end{equation}  
with $z\in \mathbb C$, $\mu_0(\beta,N)$ in \eqref{eq:mu0_definition_implicit} and $N_0(\beta,N)$ in \eqref{eq:n0_def}. It is important to note that $\mathcal H_{\mathrm{B}} (z)$ depends on $\beta$ via the latter two quantities. The momentum set $P_{\mathrm{B}}$ in the second sum is defined by 
\begin{equation*}
	P_{\mathrm B} = \{ p\in \Lambda_+^* \mid |p| \le N^{\delta_{\mathrm{Bog}}} \}
\end{equation*} with some $\delta_{ \mathrm{ Bog } }>0$ that will be chosen later (independently of $N$). The Hamiltonian in \eqref{eq:bogoliubov_ham} can be diagonalized with a (unitary) Bogoliubov transformation $T_z$, that is,
\begin{equation}\label{eq:HbogDiagonal}
	T_z^* \mathcal H_{ \mathrm B }(z) T_z = \mathcal H^{ \mathrm {diag} } = E_0 + \sum_{p\in\Lambda_+^*} \eps(p) a_p^* a_p,
\end{equation} with $E_0, \eps(p) \in \mathbb R$ (precise definitions will be given later in Section \ref{subsec:bogo_trafo}).  

The thermally excited particles will be described by the following Gibbs states related to $\mathcal H_{\mathrm{B}} (z)$:
\begin{equation}\label{eq:gibbs_state_bog_diag}
	G^{\mathrm{diag}} \coloneqq \frac{ \exp(- \beta \mathcal H ^{\mathrm{diag}} ) }{\tr_{\mathfrak{F}_+} \left[  \exp(- \beta  \mathcal H ^{\mathrm{diag}}  ) \right]} , \qquad 	\widetilde G^{\mathrm{diag}} \coloneqq \frac{ \mathds P_{\widetilde c} \exp(- \beta \mathcal H ^{\mathrm{diag}} ) }{\tr_{\mathfrak{F}_+} \left[  \mathds P_{\widetilde c} \exp(- \beta \mathcal H ^{\mathrm{diag}} ) \right]}.
\end{equation} 
Here the spectral projection $\mathds P_{\widetilde c}$ is defined as
\begin{equation}
	\label{eq:cutoff_definition}
	\mathds P_{\widetilde c} \coloneqq \mathds{1}_{ \{ \mathcal N^< \le \widetilde c \beta^{-1} N^{ \delta_{ \mathrm{ Bog } } } \} } \mathds{1}_{ \{ \mathcal N^> \le \widetilde c N \} },
\end{equation} 
with some $\widetilde c>1$ to be specified later and with the number operators
\begin{equation}
	\label{eq:N<>_def}
	\mathcal N^{<} \coloneqq \sum_{p\in P_{ \mathrm{ B } } } a_p^* a_p \quad \text{ and } \quad \mathcal N^> \coloneqq \sum_{p\in ( \Lambda_+^* \setminus P_{ \mathrm{ B } } ) } a_p^* a_p.
\end{equation}
Here and in the following we introduce all states once with and once without a particle number cutoff. This is motivated by the fact that both objects appear frequently in the computation of the free energy of our final trial state. To make the identification of quantities with cutoff easier, we always denote them with a tilde. The states in \eqref{eq:gibbs_state_bog_diag}, when transformed by $T_z$, are denoted by
\begin{equation}\label{eq:gibbs_state_bog}
	G(z) \coloneqq \frac{ \exp(- \beta \mathcal H_{\mathrm{B}}(z) ) }{\tr_{\mathfrak{F}_+} \left[   \exp(- \beta  \mathcal H_{\mathrm{B}}(z)  ) \right]} = T_z G^{\mathrm{diag}} T_z^* \quad \text{ and } \quad 	\widetilde G(z) \coloneqq T_z \widetilde G^{\mathrm{diag}} T_z^*.
\end{equation}  
Finally, the uncorrelated trial states read
\begin{equation} \label{eq:gamma_0_def} 
	\Gamma_0 =  \mathcal U \left( \int_{\mathbb C} \zeta(z)  | z \rangle \langle z | \otimes G(z) \, \mathrm d z \right) \mathcal U^*, \quad \quad \widetilde \Gamma_0 = \mathcal U \left( \int_{\mathbb C} \zeta(z)  | z \rangle \langle z | \otimes \widetilde G(z) \, \mathrm d z \right) \mathcal U^*, \\ 
\end{equation} 
with $\mathcal{U}$ in \eqref{eq:definitionU} and the coherent state $|z\rangle = W_z \Omega_0 = \exp(z a^*(\ph_0) - \overline z a(\ph_0))\Omega_0 \in \mathfrak F_0$. The probability density $\zeta(z)$ on $\mathbb{C}$ with respect to the measure $\de z = \de x \de y/\pi$ with $z = x+\mathrm{i} y$ is given by 
\begin{equation}\label{eq:condensate_density_cutoff}
	\zeta(z) = \frac {\mathds{1}_{\{|z|^2\le \widetilde c N\}}\exp( -\beta(4\pi \mathfrak a_N |z|^4 - \widetilde \mu |z|^2 ) )} { \int_{ \{ |z|^2 \le \widetilde c N \}}\exp( -\beta(4\pi \mathfrak a_N |z|^4 - \widetilde\mu |z|^2 ) ) \, \mathrm dz }. 
\end{equation} The chemical potential $\widetilde \mu = \widetilde \mu(N) \in \mathbb{R}$ will be chosen later such that our final trial state has the correct particle number. All particle number cutoffs we have introduced so far are needed for technical reasons. The restriction of the momenta in the sum in the interaction terms in \eqref{eq:bogoliubov_ham} to the set $P_{\mathrm{B}}$ is very convenient from a mathematical point of view and still allows us to obtain the term in the second line in \eqref{eq:main_result}. This is possible because $\epsilon(p) \simeq p^2-\mu_0$ for $|p| \gg 1$. In \cite{BocDeuStock2023} the state $\Gamma_0$ in \eqref{eq:gamma_0_def} has been used as uncorrelated trial state, while we are forced to work with $\widetilde \Gamma_0$ instead. This is related to the fact that we implement correlations between the particles differently than in \cite{BocDeuStock2023}. 

It remains to add microscopic correlation induced by the interaction $v_N$ to our trial state. To that end, we first apply the spectral theorem and write $G^{\mathrm{diag}} = \sum_{\alpha\in\mathcal A} \lambda_\alpha | \Psi_\alpha \rangle \langle \Psi_\alpha |$. We assume that each $\Psi_{\alpha}$ is a symmetrized product of plane waves with a definite particle number. This choice is possible because of \eqref{eq:HbogDiagonal}, and it is important for our analysis. We highlight that $\{ \Psi_{\alpha} \}_{\alpha}$ is a basis that jointly diagonalizes $\mathcal H^{ \mathrm {diag} }$, $\mathcal{N}$, $\mathcal{N}^<$ and $\mathcal{N}^>$, i.e.
\begin{equation}
	\label{eq:Eigenvalues_Psi_alpha}
	\mathcal H^{\mathrm {diag} } \Psi_\alpha = E_\alpha \Psi_\alpha , \qquad \mathcal N \Psi_\alpha = N_\alpha \Psi_\alpha , \qquad \mathcal N^< \Psi_\alpha = N_\alpha^< \Psi_\alpha , \qquad \mathcal N^> \Psi_\alpha = N_\alpha^> \Psi_\alpha,
\end{equation} 
where $N_\alpha = N_\alpha^< + N_\alpha^>$. In this representation the particle number cutoff in the definition of $\widetilde G^{\mathrm{diag}}$ amounts to restricting the sum over $\alpha$ to the set 
\begin{equation} \label{eq:A_tilde_def}
	\widetilde {\mathcal A} = \{ \alpha \in \mathcal A \mid N_\alpha^< \le \widetilde c \beta^{-1} N^{ \delta_{ \mathrm{ Bog } } } \text{ and } N_\alpha^> \le \widetilde c N \}
\end{equation} 
and to normalizing by the factor $\kappa_0 \coloneqq \sum_{\alpha\in\widetilde {\mathcal A}} \lambda_\alpha$. The eigenvalues of $\widetilde G^{\mathrm{diag}}$ therefore read $\widetilde{\lambda}_{\alpha} = \kappa_0^{-1} \lambda_{\alpha}$. 

We define the correlation structure in terms of the solution $f(|x|)$ to the zero energy scattering equation $\Delta f(|x|) = V_N(x) f(|x|)/2$ in $\mathbb{R}^3$ with the boundary condition $\lim_{|x| \to \infty} f(|x|) = 1$. Let
\begin{equation}
	f_{\ell}(x) = \begin{cases}
		f(|x|)/f(\ell) & \text{ for } |x| < \ell, \\ 1 & \text{ for } |x| \geq \ell,
	\end{cases}
	\label{eq:Definitionfl}
\end{equation}
where the parameter $\ell > 0$ is required to be strictly larger than the radius of the support of $V_N$. In the following, we assume that $\ell$ is at least twice as large as that radius. This, in particular, implies $\ell \geq 2 \mathfrak{a}_N$. We also define the operator $F$ on $\mathfrak{F}$ by
\begin{align}
	(F \Psi)^{(n)}(x_1,...,x_n) &= F_n(x_1,...,x_n) \Psi^{(n)}(x_1,...,x_n) \quad \text{ with } \nonumber \\ 
	F_n(x_1,...,x_n) &= \prod_{1 \leq i < j \leq n} f_{\ell}(x_i - x_j).
	\label{eq:jastrow_operator_def}
\end{align}
That is, $F$ multiplies each $n$-particle component of a Fock space vector $\Psi$ by a Jastrow factor. This should be compared with \cite{Jastrow1955,Dyson57,ErdSchlYau2010,DeuchertSeiringerHom2020,MaySei2020}. Finally, our trial state with correlations is given by
\begin{equation}\label{eq:trial_state}
	\Gamma =  \int_{\mathbb C } \zeta(z)  \sum_{\alpha\in \widetilde{ \mathcal A}}\frac {  \widetilde \lambda_\alpha } { \| F \mathcal U \phi_{z,\alpha} \|^2 }  | F \mathcal U \phi_{z,\alpha} \rangle \langle  F \mathcal U \phi_{z,\alpha}  |\, \mathrm d z
\end{equation} 
with $\mathcal{U}$ in \eqref{eq:definitionU} and $ \phi_{z,\alpha} = |z \rangle \otimes T_z \Psi_{\alpha} $.  

For $\Gamma$ to be an admissible trial state in the Gibbs variational principle, we require that $\tr[\mathcal N \Gamma] = N$ holds. This is not a trivial matter because (a) the chemical potential in the definition of $\widetilde{G}(z)$ is fixed and (b) the correlation structure changes the particle number of $\Gamma$ with respect to that of $\widetilde \Gamma_0$ because $\widetilde \Gamma_0$ and $\mathcal{N}$ do not commute. Under suitable assumptions on the parameters, the following lemma guarantees the existence of $\widetilde \mu \in\mathbb R$ such that $\Gamma$ is an admissible trial state.
\begin{lemma}
\label{lem:perturbation_number_particles}
We consider the combined limit $N\to\infty$, $\beta \gtrsim \beta_{\mathrm{c}}$ with $\beta_{\mathrm{c}}$ in \eqref{eq:Tc_ideal} and assume $N_0 \geq N^{2/3}$, $\delta_{ \mathrm{ Bog } }<1/12$ as well as that $\widetilde{c}$ is sufficiently large. Then there are constants $c, M > 0$ such that if $2 \mathfrak a_N \le \ell \le c N^{-7/12}$ and $N \geq M$ the following holds: There exists $\widetilde \mu \in \mathbb{R}$ such that the state $\Gamma$ in \eqref{eq:trial_state} satisfies $\tr[\mathcal N \Gamma] = N$ and we have the bound 
\begin{equation}\label{eq:difference_particle_number}
	\big | \tr [ \mathcal N \Gamma ] - \tr[ \mathcal N \widetilde \Gamma_0] \big | \lesssim N^{3}\ell^4 + N^{ 1 + \delta_{ \mathrm{ Bog } }}\ell^2 ( \beta^{-1} + 1 )
\end{equation} 
with $\widetilde \Gamma_0$ in \eqref{eq:gamma_0_def}.
\end{lemma}
In the proof of the above lemma we use a simpler version of a cancellation that we observe in the computation of the energy of $\Gamma$ in the proof of Proposition \ref{prop:renormalized_potential}. To not dilute the main line of the argument, we therefore defer it to Appendix \ref{app:perturbation_number_particles}.

\begin{remark}
As is apparent from the assumption in the above lemma, the trial state $\Gamma$ is only well defined for inverse temperatures such that $N_0(\beta,N) \geq N^{2/3}$ holds. To obtain a proof of Theorem \ref{thm:main_thm} for all inverse temperatures satisfying $\beta \gtrsim \beta_{\mathrm{c}}$, we use a second and much simpler trial state in the parameter regime defined by $N_0(\beta,N) \leq N^{2/3}$. More details can be found in Section~\ref{sec:proof_main_theorem}.
\end{remark}

\begin{remark} \label{rem:correclationStructure}
The way correlations are implemented in \cite{BocDeuStock2023} differs from our approach in two ways. The first difference is that, instead of a Jastrow factor, the authors of \cite{BocDeuStock2023} use a certain quartic (in creation and annihilation operators) transformation. Up to technicalities, their transformation amounts to multiplying an uncorrelated $n$-particle wave function by a factor $1-\sum_{1 \leq i < j \leq n} w(x_i-x_j)$, where $w(x) = 1 - f_{\ell}(x)$. This approach is very convenient from a computational perspective, but it suffers from the disadvantage that the resulting trial state is not an element of the form domain of $\mathcal{H}_N$, if $V$ is chosen as in \eqref{eq:hardcore}. The trial state $\Gamma$ in \eqref{eq:trial_state} does not have this problem. 

The second main difference between our approaches lies in the level at which correlations are introduced. In \cite{BocDeuStock2023} correlations are added to the eigenfunctions of $\Gamma_0$. With this choice, it is easier (when compared to our case) to estimate the effect of the correlation structure on the entropy of the trial state. However, the eigenfunctions of $\Gamma_0$ have a somewhat complicated structure, which makes it difficult to access useful properties of the eigenfunctions of the Gibbs state $G(z)$ in computing the energy and number of particles of the trial state. This is not a problem in \cite{BocDeuStock2023}, thanks to the special form of the correlation structure chosen there (however, it causes additional difficulties in proving the existence of a chemical potential $\widetilde{\mu}$ such that the trial state has the correct particle number, see the discussion below Lemma~2.1 in \cite{BocDeuStock2023}). In contrast, we add correlations to the eigenfunctions of the state $|z \rangle \langle z | \otimes G(z)$. This allows us to harness properties of a (suitably chosen) basis of eigenfunctions of $G(z)$, which turns out to be crucial in computing the energy of our trial state. We therefore had to find a different way to estimate the influence of the correlations on the entropy. More details can be found in Section~\ref{sec:entropy}.
\end{remark} 

To simplify the notation we will, by a slight abuse of notation, drop the isomorphism $\mathcal U$ from all formulas, and identify vectors in $\mathfrak F$ and $\mathfrak F_0\otimes \mathfrak F_+$. In the remaining part of the article we prove an upper bound for the free energy of $\Gamma$ that implies Theorem~\ref{thm:main_thm}.

\subsection{Properties of the trial state}\label{sec:PropertiesTrialState}

To compute the free energy of our trial state, precise information about its properties is needed. In this section we prove the relevant statements to not interrupt the main line of the argument later.

The following lemma, which allows us to estimate momentum sums in terms of integrals, will be used frequently in our analysis. A proof can be found in \cite[Lemma 3.3]{DeuchertSeiringerHom2020}.

\begin{lemma}\label{lem:sum_integral_approx}
	Let $f:[0,\infty) \to \mathbb R$ be nonnegative and monotone decreasing, and let $\lambda \ge 0$. Then we have 
	\begin{equation*}
		\sum_{p\in \Lambda_+^*} f(p) \mathds 1_{ [\lambda, \infty) } (|p|) \le (2\pi)^{-3}  \int_{ |p| \ge [ \lambda -  2\pi \sqrt 3 ]_+ } f( |p| ) \left( 1 + \frac{2\pi}{ |p| } + \frac{6\pi}{ |p|^2 } \right) \, \mathrm d p .
	\end{equation*}
\end{lemma}

In the following subsection we recall some properties of the operators $T_z$ and $W_z$.

\subsubsection{The Bogoliubov and Weyl transformations}
\label{subsec:bogo_trafo}

For $p\in\Lambda_+^*$ we define
\begin{equation}\label{eq:tau_def}
	\tau_p \coloneqq - \frac 1 4 \log\left[ 1 + \frac{16\pi \mathfrak  a_N N_0 \mathds 1_{ P_ {\mathrm B} }(p) }{|p|^2 - \mu_0} \right]
\end{equation}  
as well as the Bogoliubov transformation
\begin{equation}\label{eq:Tz_def}
	T_z \coloneqq \exp\Big( \sum_{p\in\Lambda^*_+} \tau_p \Big( \frac{ z^2}{|z|^2}a_p^* a_{-p}^* - \mathrm{h.c.} \Big) \Big).
\end{equation} Its action on creation and annihilation is given by
\begin{equation} \label{eq:action_bogoliubov}
	T_z^* a_p^* T_z = u_p a_p^* + v_p \frac{\overline{z}^2}{|z|^2} a_{-p} , \quad \quad T_z^* a_p T_z = u_p a_p + v_p \frac{z^2}{|z|^2} a_{-p}^*,
\end{equation} with $u_p \coloneqq \cosh(\tau_p)$ and $v_p \coloneqq \sinh(\tau_p) $, that is,
\begin{equation}\label{eq:up_vp_def}
	\begin{split}
		u_p = &\frac 1 2 \left( \frac { |p|^2 - \mu_0 } { |p|^2 - \mu_0 + 16 \pi \mathfrak a_N N_0  \mathds 1_{ P_ {\mathrm B} }(p) } \right)^ {1/4} +  \frac 1 2 \left( \frac { |p|^2 - \mu_0 } { |p|^2 - \mu_0 + 16 \pi \mathfrak a_N N_0 \mathds 1_{ P_ {\mathrm B} }(p) } \right)^ {-1/4},  \\
		v_p = & \frac 1 2 \left( \frac { |p|^2 - \mu_0 } { |p|^2 - \mu_0 + 16 \pi \mathfrak  a_N N_0\mathds 1_{ P_ {\mathrm B} }(p) } \right)^ {1/4} -  \frac 1 2 \left( \frac { |p|^2 - \mu_0 } { |p|^2 - \mu_0 + 16 \pi \mathfrak  a_N N_0 \mathds 1_{ P_ {\mathrm B} }(p)} \right)^ {-1/4}.
	\end{split}
\end{equation} 
The coefficients $u_p, v_p$ satisfy the following bounds.

\begin{lemma}\label{lem:up_vp_norms}
	 The coefficient $v_p$ in \eqref{eq:up_vp_def} satisfies the bounds
	\begin{equation}\label{eq:up_vp_norm_bounds1}
		\begin{split}
			 v_p \lesssim \frac {N_0}{N |p|^2} \lesssim \frac{N_0}{N}, \qquad \| v \|_2 \lesssim \frac {N_0} {N} \quad \text{ and } \quad \sum_{p\in \Lambda_+^*} |p|^k |v_p| \lesssim \frac {N_0} N N^{(k+1)\delta_{\mathrm {Bog}}} 
		\end{split}		
	\end{equation}
	for $k \in \{0,1,2\}$. For $u_p$ we have
	\begin{equation}\label{eq:up_vp_norm_bounds2}
		\| u -1 \|_\infty\lesssim \frac{N_0^2}{N^2}.
	\end{equation}
\end{lemma}
\begin{proof}
	We first observe that \eqref{eq:up_vp_def} implies
	\begin{equation*}
		v_p^2 =  \frac 1 4 \left( \frac { |p|^2 - \mu_0 } { |p|^2 - \mu_0 + 16 \pi \mathfrak a_N N_0  \mathds 1_{ P_ {\mathrm B} }(p)  } \right)^ {1/2} +  \frac 1 4 \left( \frac { |p|^2 - \mu_0 } { |p|^2 - \mu_0 + 16 \pi  \mathfrak  a_N N_0 \mathds 1_{ P_ {\mathrm B} }(p)  } \right)^ {-1/2} - \frac 1 2.
	\end{equation*} Using  $\mu_0<0$, $ \mathfrak a_N =  \mathfrak a / N$ and the inequality $0 \le (1+x)^{1/2} + (1+x)^{-1/2} - 2 \le x^2/4$ for $x\ge 0$, we find
	\begin{equation}\label{eq:vp_pointwise_bound}
		v_p^2 \lesssim \frac {N_0^2}{N^2 |p|^4}. 
	\end{equation} The remaining estimates in \eqref{eq:up_vp_norm_bounds1} and \eqref{eq:up_vp_norm_bounds2} follow from \eqref{eq:vp_pointwise_bound}, $|p|\ge 2\pi$ for  $p\in\Lambda_+^*$, the identity $u_p^2 = 1+v_p^2$ and the inequality $\sqrt{1+x} -1 \le x$ for $x \ge 0$. 
\end{proof} 

As we mentioned earlier, the unitary $T_z$ diagonalizes the Bogoliubov Hamiltonian. The precise statement is the following.
\begin{lemma}
	Let $z\in\mathbb C$, $\mathcal H_{\mathrm{B}} (z)$ in \eqref{eq:bogoliubov_ham} and $\tau_p$ in \eqref{eq:tau_def}. We have
	\begin{equation}\label{eq:bogoliubov_ham_diagonalization}
		T_z^* \mathcal H_{\mathrm{B}} (z) T_z = \mathcal H^{ \mathrm{diag} } \coloneqq E_0 + \sum_{p\in\Lambda^*_+} \eps(p) a_p^* a_p,
	\end{equation} 
	with the ground state energy
	\begin{equation}\label{eq:E_0}
		E_0 \coloneqq -\frac 1 2 \sum_{ p\in P_{\mathrm B}} \left[ |p|^2-\mu_0 + 8\pi \mathfrak  a_N N_0 -\eps(p) \right] 
	\end{equation}
	and the Bogoliubov dispersion relation
	\begin{equation} \label{eq:bogoliubov_dispersion}
		\eps(p) \coloneqq \begin{cases}
			\sqrt{|p|^2 - \mu_0} \sqrt{|p|^2 -\mu_0 + 16\pi \mathfrak  a_N N_0 }, &  p\in P_{\mathrm B} ,\\
			|p|^2 - \mu_0, & p\in\Lambda_+^*\setminus P_{\mathrm B}. 
		\end{cases}
	\end{equation} 
\end{lemma} The proof is a standard computation based on \eqref{eq:action_bogoliubov} and \eqref{eq:up_vp_def}, which can be found for instance in \cite[Lemma 5.2]{BBCS2020}.

Next, we recall the definition of the Weyl operator $W_z \coloneqq \exp(za^*(\ph_0) - \overline z a(\ph_0))$ and the well-known formulas
\begin{equation}
	\label{eq:action_weyl_position}
	W_z^* a_x W_z = a_x + z, \qquad W_z^* a_x^* W_z = a_x^* + \overline{z},
\end{equation} 
\begin{equation}
	\label{eq:action_weyl_momentum}
	W_z^* a_p W_z = a_p + z\delta_{p,0}, \qquad W_z^* a_p^* W_z = a_p^* + \overline{z}\delta_{p,0},
\end{equation} for every $x\in\Lambda$ and $p\in\Lambda^*$. The Bogoliubov and the Weyl transformations cannot change the particle number too much. The precise statement is captured in the following lemma.
\begin{lemma}
	\label{lem:Tz_N_bound}
	The operator inequalities
	\begin{equation} \label{eq:Tz_N_bound}
		T_z^* (\mathcal N^< + 1)^k T_z \lesssim_k ( \mathcal N^< + 1 )^{k}
	\end{equation} 
	and 
	\begin{equation} \label{eq:Wz_N_bound}
		W_z^* (\mathcal N^< + 1)^k W_z \lesssim_k  \mathcal ( N^< + |z|^2 + 1  )^{k}
	\end{equation} 
	hold for all $k\in\mathbb N$.
\end{lemma}
We omit the proof, which is a standard application of Gr\"onwall's inequality, compare for instance with \cite[Lemma~3.1]{BreSchlein2019}.
%
We are now prepared to prove several important properties of the Bogoliubov Gibbs state in the diagonal and in the non-diagonal representation. 

\subsubsection{The Bogoliubov Gibbs state in the diagonal representation}

We start by computing the 1-pdm of $G^{ \mathrm{ diag } }$ in \eqref{eq:gibbs_state_bog_diag}. 
\begin{lemma}\label{lem:integral_gamma_0} 
	For every $p, q\in\Lambda_+^*$, we have 
	\begin{equation}\label{eq:correlations_Gdiag_def}
		\tr_{\mathfrak{F}_+} [a_p^* a_q G^{ \mathrm{ diag } }] = \gamma^{ \mathrm {diag} }_p \delta_{p,q}, \qquad 	\tr_{\mathfrak{F}_+} [a_p a_q G^{ \mathrm{ diag } }] = 0,  
	\end{equation} with 
	\begin{equation} \label{eq:gamma_Bog_def}
		\gamma^{ \mathrm {diag} }_p \coloneqq  \frac{1}{\exp( \beta \eps(p)) - 1}
	\end{equation}
	and $\eps(p)$ in \eqref{eq:bogoliubov_dispersion}. The sequence of eigenvalues $\gamma_p^{ \mathrm { diag } }$ satisfies 
	\begin{equation}\label{eq:gamma_Bog_bound_momentum}
		\sum_{p \in \Lambda_+^*} \gamma_p^{ \mathrm { diag } } \lesssim \beta^{-3/2} + \beta^{-1}, \qquad \sum_{p \in \Lambda_+^*} p^2 \gamma_p^{ \mathrm { diag } } \lesssim \beta^{-5/2} + \beta^{-3/2}, \qquad \sum_{p \in \Lambda_+^*} | \gamma_p^{ \mathrm { diag } } |^2 \lesssim \beta^{-2}.
	\end{equation}
	Moreover, 
	\begin{equation}\label{eq:gamma_Bog_bound_position}
		  \int_{\Lambda} | \check \gamma^{ \mathrm {diag} }(x)|\, \mathrm d x \lesssim \beta^{-1}. 
	\end{equation}
\end{lemma}
\begin{proof}
	The formulas \eqref{eq:correlations_Gdiag_def}, \eqref{eq:gamma_Bog_def} follow directly from the definition of $G^{ \mathrm{ diag } }$. We apply Lemma \ref{lem:sum_integral_approx} with $\lambda \in (2\pi\sqrt 3 , 4\pi)$, use $\eps(p) \ge |p|^2$ and $(\exp(x)-1)^{-1} \le x^{-1}$ to see that 
	\begin{equation*}
		\sum_{p \in \Lambda_+^*} \gamma_p^{ \mathrm { diag } } \le \sum_{p \in \Lambda_+^*} \frac { \mathds 1_{[\lambda, \infty)}(|p|) } { \exp( \beta |p|^2 ) -1 } + C\beta^{-1} \lesssim \int_{ |p| \ge \lambda -  2\pi \sqrt 3 }  \frac {1 + |p|^{-2}} { \exp( \beta |p|^2 ) -1 } \, \mathrm d p + \beta^{-1} \lesssim \beta^{-3/2} + \beta^{-1}
	\end{equation*} 
	holds. With the same argument we prove the second bound in \eqref{eq:gamma_Bog_bound_momentum}. We also have
	\begin{equation*}
		\begin{split}
			\sum_{p\in\Lambda_+^*} |\gamma_p^{ \mathrm{diag} } |^2 \le \frac{1}{ \beta^2 }  \sum_{p\in\Lambda_+^*} \frac 1 { |p|^4 } \lesssim \beta^{-2},
		\end{split}
	\end{equation*} which completes the proof of \eqref{eq:gamma_Bog_bound_momentum}. 
	Finally, to obtain \eqref{eq:gamma_Bog_bound_position}, we apply Cauchy-Schwartz and \eqref{eq:gamma_Bog_bound_momentum}:
	\begin{equation*}
		\int_{\Lambda} | \check \gamma^{ \mathrm {diag} }(x)|\, \mathrm d x \leq \left( \int_{\Lambda} | \check \gamma^{ \mathrm {diag} }(x)|^2\, \mathrm d x \right)^{1/2} \lesssim \beta^{-1}.
	\end{equation*}
	%
\end{proof} 

The following lemma allows us to control the expectation of powers of $\mathcal N_+$ (the restriction of $\mathcal N$ to $\mathfrak F_+$) in the state $G^{\mathrm{ diag }}$. 
\begin{lemma}\label{cor:expectation_powers_N}
	For every $k \ge 2$, we have 
	\begin{equation}\label{eq:expectation_powers_N}
		\tr_{\mathfrak{F}_+}[\mathcal N_+^k G^{ \mathrm{ diag} }] - \big( \tr_{\mathfrak{F}_+}[\mathcal N_+ G^{ \mathrm{ diag} }]\big) ^k \lesssim_k \beta^{-\frac{3(k-2)} 2}\beta^{-2} + \beta^{-(k-1)}.
	\end{equation}
\end{lemma}
\begin{proof} 
	Let us introduce the notation $\mathrm d X^k = \mathrm d x_1, ... , \, \mathrm d x_k$. Using the CCR we see that
	\begin{equation}\label{eq:tr_N_Gdiag_1}
		\begin{split}
			\tr_{\mathfrak{F}_+} & [\mathcal N^{k} G^{ \mathrm{ diag} }] \\
			= & \int_{\Lambda^{k+1}}  \tr_{\mathfrak{F}_+}[a_{x_1}^* ... a_{x_{k}}^* a_{x_1} ... a_{x_{k}} G^{ \mathrm{ diag} }] \, \mathrm d X^{k} +\tr_{\mathfrak{F}_+}[ \big( \mathcal N^{k} - \mathcal N \cdot ... \cdot (\mathcal N-k + 1) \big) G^{ \mathrm{ diag} }].
		\end{split}
	\end{equation} 
	An application of Wick's theorem and the identity $\tr_{\mathfrak F_+}[ a_x^* a_x G^{\mathrm {diag} } ] = \tr_{\mathfrak F_+}[ \mathcal N_+ G^{\mathrm {diag} } ]$ show that the first term on the right-hand side equals
	\begin{equation*}
		\big( \tr_{\mathfrak F_+}[ \mathcal N_+ G^{\mathrm {diag} } ] \big)^k + \sum_{\pi \in S_k \setminus \{\mathrm{Id}_k \}} \int_{\Lambda^k} \prod_{i = 1}^k \tr_{ \mathfrak{F}_+ }[ a^*_{ x_i } a_{ x_{ \pi(i) } } G^{ \mathrm{ diag } } ] \mathrm d X^k,
	\end{equation*} where $S_k$ denotes the set of permutations of $\{ 1, ..., k \}$ and $\mathrm{Id_k} \in S_k$ is the identity. To bound the sum on the right-hand side, we first observe that an application of Cauchy--Schwartz and Lemma~\ref{lem:integral_gamma_0} imply
	\begin{equation*}
		| \tr_{\mathfrak F_+}[ a_x^* a_y G^{\mathrm {diag} } ] | \leq \sqrt{\tr_{\mathfrak F_+}[ a_x^* a_x G^{\mathrm {diag} } ] \tr_{\mathfrak F_+}[ a_y^* a_y G^{\mathrm {diag} } ]} \lesssim \beta^{-3/2} + \beta^{-1}.
	\end{equation*}
	Moreover, if $\pi\in S_k \setminus \{ \mathrm{Id}_k \}$, there exist distinct $i, j \in \{1, ..., k\}$ such that $\pi(i)\neq i$, $\pi(j) \neq j$. We can thus estimate
	\begin{equation*}
		\begin{split}
			\sum_{\pi \in S_k \setminus \{\mathrm{Id}_k \}}  \int_{\Lambda^k} \prod_{i = 1}^k & \tr_{ \mathfrak{F}_+ }[ a^*_{ x_i } a_{ x_{ \pi(i) } } G^{ \mathrm{ diag } }] \mathrm d X^k \\
			\lesssim	& k! ( \beta^{-3/2} + \beta^{-1} )^{k-2} \Big ( \int_{\Lambda^4} | \check \gamma^{\mathrm{diag}} (x_1 - x_{ 3 }) \check \gamma^{\mathrm{diag}} (x_2 - x_{ 4 }) | \, \mathrm d x_1 \, \mathrm d x_2 \, \mathrm d x_{3} \, \mathrm d x_{4} \\
			& \qquad +  \int_{\Lambda^4} | \check \gamma^{\mathrm{diag}} (x_1 - x_{ 3 }) \check \gamma^{\mathrm{diag}} (x_2 - x_{ 3 }) | \, \mathrm d x_1 \, \mathrm d x_2 \, \mathrm d x_{3} \Big) \\
			\lesssim_k & ( \beta^{-3/2} + \beta^{-1} )^{k-2} \beta^{-2}, 
		\end{split}
	\end{equation*} where we used \eqref{eq:gamma_Bog_bound_position} in the last step.
	
	To obtain a bound for the second term on the right-hand side of \eqref{eq:tr_N_Gdiag_1}, we argue as above and find
	\begin{equation*}
		\tr_{\mathfrak{F}_+}[ \big( \mathcal N^{k} - \mathcal N \cdot ... \cdot (\mathcal N-k + 1) \big) G^{ \mathrm{ diag} }] \lesssim_k \tr_{\mathfrak{F}_+}[ \mathcal N^{k-1} G^{ \mathrm{ diag} }] \lesssim_k (\beta^{-3/2} +\beta^{-1})^{(k-1)}.
	\end{equation*} In combination, these considerations prove \eqref{eq:expectation_powers_N}.
\end{proof}

To compare the energy of our trial state to the expression in \eqref{eq:main_result}, we need to remove the cutoff on the number of particles in several terms. The next lemma allows us to control the corresponding errors.
\begin{lemma} \label{lem:cutoff_a_remainder}
	Let $\beta \gtrsim \beta_{\mathrm{c}}$. There exist constants $\widetilde c, c>0$ such that, for every $m\in\mathbb N$, we have
	\begin{equation} \label{eq:cutoff_a_remainder_bound}
	\tr_{ \mathfrak{F}_+ }[ ( \mathds 1 - \mathds P_{ \widetilde c } ) \mathcal N_+^m G^{\mathrm{diag}} ] = \sum_{\alpha\in\mathcal A \setminus \widetilde{ \mathcal A}} \lambda_\alpha N_\alpha^m \lesssim_m \exp(-c N^{ \delta_{ \mathrm{ Bog } } } )
	\end{equation} 
	with $\mathds P_{ \widetilde c }$ in \eqref{eq:cutoff_definition}. In particular, $1\ge \kappa_0 \ge 1- C\exp(-c N^{ \delta_{ \mathrm{ Bog } } })$. Moreover, 
	\begin{equation} \label{eq:cutoff_a_remainder_bound_energy}
		\tr_{ \mathfrak{F}_+ }[ ( \mathds 1 - \mathds P_{ \widetilde c } ) \mathcal H^{\mathrm{diag}} G^{\mathrm{diag}} ] =	\sum_{\alpha\in\mathcal A \setminus \widetilde{ \mathcal A}} \lambda_\alpha E_\alpha \lesssim \exp(-c N^{ \delta_{ \mathrm{ Bog } } } ).
	\end{equation}
\end{lemma}
\begin{proof}
	We start by proving the first statement in the case $m = 0$. For notational convenience, let us introduce $\widetilde N^< = \widetilde c \beta^{-1} N^{\delta_{ \mathrm{ Bog } }}$ and $\widetilde N = \widetilde c N$.
	Using the fact that 
	\begin{equation*}
		\mathds 1 - \mathds P_{\widetilde c} \le \mathds 1_{ \{ \mathcal N^< > \widetilde N^< \}} + \mathds 1_{ \{ \mathcal N^> > \widetilde N \} } \le \mathds 1_{ \{ \mathcal N^< > \widetilde N^< \}} + \mathds 1_{ \{ \mathcal N_+ > \widetilde N \} },
	\end{equation*} we can bound 
	\begin{equation}\label{eq:cutoff_remainder_first_bound}
	\tr_{ \mathfrak{F}_+ }[ ( \mathds 1 - \mathds P_{ \widetilde c } ) G^{\mathrm{diag}} ] \le \tr_{\mathfrak F_+}[ \mathds 1_{ \{ \mathcal N^< > \widetilde N^< \}}  G^{ \mathrm{ diag } } ] + \tr_{\mathfrak F_+}[  \mathds 1_{ \{ \mathcal N_+ > \widetilde N \} } G^{ \mathrm{ diag } } ]. 
	\end{equation} 
	With the inequality $\mathds 1_{\{ x > 0 \}} \le e^{\eta x}$ valid for any $x \in \mathbb R$ and $\eta>0$, we estimate the second term on the right-hand side by 
	\begin{equation}\label{eq:cutoff_intermediate_bound_whole_1}
		\tr_{\mathfrak F_+}[  \mathds 1_{ \{ \mathcal N_+ > \widetilde N \} } G^{ \mathrm{ diag } } ] \le  \tr_{\mathfrak F_+}[ e^{\eta (\mathcal N_+ - \widetilde N)} G^{ \mathrm{ diag } } ]	= e^{-\eta \widetilde N} \frac { \tr_{\mathfrak F_+}[ e^{-\beta \mathcal H^{ \mathrm{ diag } }  + \eta \mathcal N_+}  ] } { \tr_{\mathfrak F_+}[ e^{-\beta \mathcal H^{ \mathrm{ diag } }  } ] }.
	\end{equation}
	For the choice $\eta = 2\pi^2 \beta $, which ensures $ \eta / \beta < 4\pi^2$, we have 
	\begin{equation*}
		\tr_{\mathfrak F_+}[ e^{-\beta \mathcal H^{ \mathrm{ diag } }  + \eta \mathcal N_+} ] = \exp\Bigg( - \sum_{p\in\Lambda_+^*} \log \Big( 1-e^{-\beta\eps(p) + \eta } \Big)\Bigg).
	\end{equation*} 
	We expand the logarithm to second order and use the fact that the function $x\mapsto (\cosh(x)-1)^{-1}$ is decreasing and satisfies $(\cosh(x)-1)^{-1}\le 2x^{-2}$ for $x > 0$, to see that
	\begin{equation} \label{eq:large_deviations_intermediate_1}
		\begin{split}
			\tr_{\mathfrak F_+}[ e^{-\beta\mathcal H^{\mathrm{ diag }} + \eta \mathcal N_+} ]	\le & 	\tr_{\mathfrak F_+}[ e^{-\beta\mathcal H^{\mathrm{ diag }} } ] \exp \Bigg( \eta \sum_{p\in\Lambda_+^*} \frac 1 {e^{\beta\eps(p)} -1 } \Bigg) \exp\Bigg( \frac 1 4 \eta^2 \sum_{p\in\Lambda_+^*} \frac{1}{\cosh(\frac 1 2 \beta\eps(p)) - 1 }\Bigg) \\
			\le & 	\tr_{\mathfrak F_+}[ e^{-\beta\mathcal H^{\mathrm{ diag }} } ] \exp \Bigg( \eta \sum_{p\in\Lambda_+^*} \frac 1 {e^{\beta\eps(p)} -1 } \Bigg) \exp\Bigg( 2 \eta^2 \sum_{p\in\Lambda_+^*} \frac{1}{\beta^2|p|^4 }\Bigg) \\
			\lesssim & 	\tr_{\mathfrak F_+}[ e^{-\beta\mathcal H^{\mathrm{ diag }} } ] \exp \Bigg( \eta \sum_{p\in\Lambda_+^*} \frac 1 {e^{\beta\eps(p)} -1 } \Bigg).
		\end{split}
	\end{equation} It follows from \eqref{eq:mu0_definition_implicit} that
	\begin{equation*}
		\sum_{p\in\Lambda_+^*} \frac 1 {e^{\beta\eps(p)} -1 } \le \sum_{p\in\Lambda_+^*} \frac 1 {e^{\beta( |p|^2 - \mu_0 )} -1 } \le N .
	\end{equation*} 
	In combination with \eqref{eq:cutoff_intermediate_bound_whole_1} and \eqref{eq:large_deviations_intermediate_1}, this bound implies
	\begin{equation} \label{eq:cutoff_a_remainder_intermediate_bound}
		\begin{split}
			\tr_{\mathfrak F_+}[  \mathds 1_{ \{ \mathcal N_+ > \widetilde N \} } G^{ \mathrm{ diag } } ] \lesssim e^{-\eta ( \widetilde N - N)}.
		\end{split}
	\end{equation} 
	
	Next, we prove a similar estimate for the first term on the right-hand side of \eqref{eq:cutoff_remainder_first_bound}. We have
	\begin{equation*}
		\tr_{\mathfrak F_+}[ \mathds 1_{ \{ \mathcal N^< > \widetilde N^< \}}  G^{ \mathrm{ diag } } ] \le e^{-\eta \widetilde N^<} \frac { \tr_{\mathfrak F_+}[ e^{-\beta \mathcal H^{ \mathrm{ diag, < } }  + \eta \mathcal N^<}  ] } { \tr_{\mathfrak F_+}[ e^{-\beta \mathcal H^{ \mathrm{ diag , < } }  } ] }
	\end{equation*} with $\mathcal H^{\mathrm {diag}, <} = E_0 + \sum_{p \in P_{\mathrm B}}\eps(p) a_p^* a_p$. Arguing as in \eqref{eq:large_deviations_intermediate_1} with the same choice of $\eta$, we also see that 
	\begin{equation*}
		\begin{split}
			\tr_{\mathfrak F_+}[ e^{-\beta\mathcal H^{\mathrm{ diag }, <} + \eta \mathcal N^<} ]	\lesssim	\tr_{\mathfrak F_+}[ e^{-\beta\mathcal H^{\mathrm{ diag }, <} } ] \exp \Bigg( \eta \sum_{p\in P_{\mathrm{ B }}} \frac 1 {e^{\beta\eps(p)} -1 } \Bigg).	
			\end{split}
	\end{equation*} The sum in the above exponential is bounded by
	\begin{equation*}
		\beta^{-1} \sum_{p\in P_{\mathrm{ B }}} \frac 1 { |p|^2 } \lesssim \beta^{-1} N^{\delta_{ \mathrm{ Bog } }},
	\end{equation*} and hence 
	\begin{equation}
		\label{eq:cutoff_a_remainder_intermediate_bound_2}
		\tr_{\mathfrak F_+}[ \mathds 1_{ \{ \mathcal N^< > \widetilde N^< \}}  G^{ \mathrm{ diag } } ] \lesssim e^{ -\eta ( \widetilde N^< - c_1 \beta^{-1} N^{ \delta_{ \mathrm{ Bog } } } ) },
	\end{equation} for some $c_1>0$ independent of $N$. In combination, \eqref{eq:cutoff_remainder_first_bound}, \eqref{eq:cutoff_a_remainder_intermediate_bound} and \eqref{eq:cutoff_a_remainder_intermediate_bound_2} imply \eqref{eq:cutoff_a_remainder_bound} for $m = 0$ and $\widetilde c > \max\{1, c_1\}$.
	
	If $m\ge 1$, we write
	\begin{equation*}
		\begin{split}
			\tr_{ \mathfrak{F}_+ }[ ( \mathds 1 - \mathds P_{ \widetilde c } ) \mathcal N_+^m G^{\mathrm{diag}} ] \le & \tr_{\mathfrak F_+}[\mathcal N_+^m\mathds 1_{\{\mathcal N_+ > \widetilde N\}} G^{ \mathrm{ diag } } ] + \tr_{\mathfrak F_+}[\mathcal N_+^m\mathds 1_{\{\mathcal N^< > \widetilde N^<\}} G^{ \mathrm{ diag } } ] \\
			\lesssim_m &  \tr_{\mathfrak F_+}[( \mathcal N_+ - \widetilde N)^m\mathds 1_{\{\mathcal N_+ > \widetilde N\}}  G^{ \mathrm{ diag } } ] + \widetilde N^m \tr_{\mathfrak F_+}[\mathds 1_{\{\mathcal N_+ > \widetilde N\}}  G^{ \mathrm{ diag } } ] \\
			& + \tr_{\mathfrak F_+}[(\mathcal N^<  - \widetilde N^< )^m\mathds 1_{\{\mathcal N^< > \widetilde N^<\}} G^{ \mathrm{ diag } } ] + (\widetilde N^<)^m \tr_{\mathfrak F_+}[\mathds 1_{\{\mathcal N^< > \widetilde N^<\}}  G^{ \mathrm{ diag } } ] \\ 
			& + \tr_{\mathfrak F_+}[(\mathcal N^>)^m\mathds 1_{\{\mathcal N^< > \widetilde N^<\}} G^{ \mathrm{ diag } } ].
		\end{split}
	\end{equation*} 
	To estimate the first and the third term on the right-hand side, we apply the inequality $x^m\mathds 1_{\{x\ge 0\}}\le m! e^{\eta x}/\eta^m$. The rest of the argument is the same as in the case $m=0$. With \eqref{eq:cutoff_a_remainder_intermediate_bound}, \eqref{eq:cutoff_a_remainder_intermediate_bound_2} and $\widetilde N, \widetilde N^< \lesssim N$ we see that the second and the fourth term are bounded by a constant times $\exp(-c N^{ \delta_{ \mathrm{ Bog } } })$. To obtain a bound for the last term, we observe that it equals 
	\begin{equation*}
		\tr_{\mathfrak F_+}[\mathds 1_{\{\mathcal N^< > \widetilde N^<\}} G^{ \mathrm{ diag } } ] 	\tr_{\mathfrak F_+}[(\mathcal N^>)^m G^{ \mathrm{ diag } } ] \lesssim e^{-cN^{ \delta_{ \mathrm{ Bog } } } }	\tr_{\mathfrak F_+}[ \mathcal N_+^m G^{ \mathrm{ diag } } ] \lesssim_m N^m e^{-cN^{ \delta_{ \mathrm{ Bog } } } } ,
	\end{equation*} where we used $[\mathcal N^<, \mathcal N^>] = 0$, \eqref{eq:cutoff_a_remainder_intermediate_bound} and \eqref{eq:cutoff_a_remainder_intermediate_bound_2}.

	It remains to prove \eqref{eq:cutoff_a_remainder_bound_energy}. The Cauchy--Schwarz inequality and \eqref{eq:cutoff_a_remainder_bound} for $m=0$ imply
	\begin{equation}\label{eq:cutoff_a_energy_intermediate}
		\begin{split}
			\tr_{ \mathfrak{F}_+ }[ ( \mathds 1 - \mathds P_{ \widetilde c } ) \mathcal H^{\mathrm{diag}} G^{\mathrm{diag}} ] \le & \big(\tr_{ \mathfrak{F}_+ }[ ( \mathds 1 - \mathds P_{ \widetilde c } ) G^{\mathrm{diag}} ] \big)^{1/2} \big(\tr_{ \mathfrak{F}_+ }[ (\mathcal H^{ \mathrm{ diag } })^2 G^{\mathrm{diag}} ] \big)^{1/2} \\
			\lesssim & e^{-cN^{ \delta_{ \mathrm{ Bog } } }} \big(\tr_{ \mathfrak{F}_+ }[ (\mathcal H^{ \mathrm{ diag } })^2 G^{\mathrm{diag}} ] \big)^{1/2}.
		\end{split}
	\end{equation} 
	To estimate the second factor on the right-hand side we apply Wick's theorem:
	\begin{equation*}
		\tr_{ \mathfrak{F}_+ }[ (\mathcal H^{ \mathrm{ diag } })^2 G^{\mathrm{diag}} ]  = \Big(\sum_{p \in\Lambda_+^*} \eps(p)\gamma_p^{\mathrm{diag}} \Big)^2 + \sum_{p \in\Lambda_+^*} \eps(p)^2\gamma_p^{\mathrm{diag}} ( 1 +  \gamma_p^{\mathrm{diag}}) \lesssim \beta^{-5} + \beta^{-5/2}.
	\end{equation*} 
	In the last step we used the monotonicity of $x \mapsto (e^x-1)^{-1}$, $\eps(p) \geq |p|^2$ and Lemma \ref{lem:sum_integral_approx}. Inserting this bound into \eqref{eq:cutoff_a_energy_intermediate} we get \eqref{eq:cutoff_a_remainder_bound_energy}.

\end{proof}

\subsubsection{The Bogoliubov Gibbs state in the original representation}

We now compute and estimate the correlation functions of $G(z)$.

\begin{lemma}\label{lem:1pdm_pairing_bogoliubov}
	For every $z\in\mathbb C$, we have 
	\begin{equation}\label{eq:correlations_GB_momentum}
		\tr_{\mathfrak F_+}[a_q^*a_p G(z) ] = \gamma_p \delta_{p,q}, \qquad \tr_{\mathfrak F_+}[a_qa_p G(z) ] = (z/|z|)^2 \alpha_p \delta_{p,-q},
	\end{equation} 
with 
	\begin{equation}
		\label{eq:1pdm_pairing_bogoliubov_momentum}
		\begin{split}
			\gamma_p = (1+2v_p^2)  \gamma_p^{ \mathrm{ diag } } + v_p^2, \qquad  \alpha_p =u_pv_p \Big(  2 \gamma_p^{ \mathrm{ diag } }  + 1 \Big)
		\end{split}
	\end{equation} 
	and $u_p, v_p$ in \eqref{eq:up_vp_def}. Moreover, we have the bounds 
	\begin{equation} \label{eq:pointwise_bounds_gamma_alpha}
		\gamma_p \lesssim \frac  { N_0^2 \mathds 1_{ P_ {\mathrm B} }(p)  } { N^2 |p|^4} + \frac {1} {\exp(\beta|p|^2) - 1 }, \qquad \alpha_p \lesssim \frac  { N_0 \mathds 1_{ P_ {\mathrm B} }(p)  } {N |p|^2} \left(1 + \frac 1 {\beta |p|^2}\right),
	\end{equation} as well as 
	\begin{align}
	\sum_{ p\in P_{\mathrm B} } \gamma_p \lesssim & \beta^{-1} N^{\delta_{\mathrm {Bog}}} + \frac{N_0^2}{N^2}, &	\sum_{ p\in \Lambda_+^* } \gamma_p \lesssim &  \beta^{-3/2} + \beta^{-1/2} + \frac{N_0^2}{N^2}, \nonumber \\
		\sum_{ p\in \Lambda_+^* } \gamma_p^2 \lesssim & \beta^{-2} + \frac{N_0^4}{N^4} , & \sum_{p\in\Lambda_+^*} |p| \gamma_p \lesssim & \beta^{-2} + \beta^{-1} + \frac {N_0^2}  {N^2} \log N , \label{eq:norm_bounds_gamma} \\
		\sum_{ p\in \Lambda_+^* } | \alpha_p | \lesssim & \frac{N_0}{N} \left( \beta^{-1} + N^{\delta_{\mathrm {Bog}}}  \right) , & \sum_{ p\in \Lambda_+^*} |\alpha_p|^2 \lesssim  & \frac{N_0^2}{N^2} \left( \beta^{-2} +1 \right) ,  \nonumber \\
	 	\sum_{p\in\Lambda_+^*} |p| |\alpha_p| \lesssim  & \frac{N_0}{N} \left( \beta^{-1} \log N + N^{2\delta_{\mathrm{Bog}}} \right). \label{eq:norm_bounds_alpha}
	\end{align}
\end{lemma}
\begin{proof}
	Using \eqref{eq:action_bogoliubov}, Lemma~\ref{lem:integral_gamma_0}, the fact that $\gamma^{ \mathrm{ diag } }$ and $v$ are even functions of $p$, and the identity $u_p^2 = 1 + v_p^2$, we find
	\begin{equation}\label{eq:gamma_GB_computation}
		\begin{split}
			\tr_{\mathfrak F_+}[a_q^*a_p G(z)] = & \tr_{\mathfrak F_+} \Big [ \Big( u_q a_q^* + v_q \frac{ \overline z^2 }{ |z|^2 } a_{-q} \Big) \Big( u_p a_p +v_p \frac{ z^2 }{ |z|^2 } a_{-p}^* \Big) G^{ \mathrm{ diag} } \Big] \\
			= & u_p^2 \gamma_p^{ \mathrm{ diag } } \delta_{p,q} + v_{p}^2 (1 + \gamma_{p}^{ \mathrm{ diag } } ) \delta_{p,q} = (1+2v_p^2)  \gamma_p^{ \mathrm{ diag } } \delta_{p, q} + v_p^2\delta_{p,q},
		\end{split}
	\end{equation} which is the first identity in \eqref{eq:1pdm_pairing_bogoliubov_momentum}. We find the second identity after an analogous computation for $	\tr_{\mathfrak F_+}[a_q^*a_p^* G(z)] $. In combination, \eqref{eq:1pdm_pairing_bogoliubov_momentum} and Lemma \ref{lem:up_vp_norms} show \eqref{eq:pointwise_bounds_gamma_alpha}. The bounds in \eqref{eq:norm_bounds_gamma} and \eqref{eq:norm_bounds_alpha} are a consequence of \eqref{eq:pointwise_bounds_gamma_alpha} and Lemma~\ref{lem:sum_integral_approx}. 
\end{proof}

Lemma \ref{lem:1pdm_pairing_bogoliubov} allows us to compare the number of particles in $G(z)$ with the one in the Gibbs state
\begin{equation}\label{eq:GibbsStateIdealGas}
	G^{ \mathrm {id} } = \frac{\exp(-\beta \sum_{p \in \Lambda_+} (p^2 - \mu_0) a_p^*a_p)}{\tr_{ \mathfrak{F}_+ } [ \exp(-\beta \sum_{p \in \Lambda_+} (p^2 - \mu_0) a_p^*a_p) ]}
\end{equation}
describing the thermally excited particles in the ideal Bose gas.

\begin{lemma}\label{lem:n_g(z)_ideal_difference}
	We have 
	\begin{equation} \label{eq:n_g(z)_ideal_difference}
			\big | \tr_{ \mathfrak{F}_+ } [ \mathcal N_+  G(z) ] - \tr_{ \mathfrak{F}_+ } [ \mathcal N_+ G^{ \mathrm {id} } ] \big | \lesssim \frac{N_0}{\beta N} + \frac{N_0^2}{N^2}.
	\end{equation}
\end{lemma}
\begin{proof}
With Lemma \ref{lem:1pdm_pairing_bogoliubov} we compute
\begin{equation}\label{eq:gamma_gamma0_difference}
	\begin{split}
		\tr_{ \mathfrak{F}_+ } [ \mathcal N_+  G(z) ] - \tr_{ \mathfrak{F}_+ } [ \mathcal N_+ G^{ \mathrm {id} } ] = & \sum_{ p\in P_{\mathrm B} } \left( \gamma_p - \frac{1}{e^{\beta(|p|^2-\mu_0) -1}} \right) \\
		= &  \sum_{p\in P_{\mathrm B} } \left( \frac{1}{e^{\beta\eps(p)} -1} - \frac{1}{e^{\beta(|p|^2-\mu_0)} -1} \right) +  \sum_{ p\in P_{\mathrm B} }\left( \frac{2 v_p^2}{e^{\beta\eps(p)} -1} + v_p^2 \right).
	\end{split}
\end{equation} 
Using $e^x-1\ge x$, $\eps(p)\ge |p|^2$, and Lemma \ref{lem:up_vp_norms}, it is straightforward to see that 
\begin{equation}\label{eq:n0_n0tilde_difference_1}
	\sum_{ p\in P_{\mathrm B} }\left( \frac{2 v_p^2}{e^{\beta\eps(p)} -1} + v_p^2 \right) \lesssim \frac{N_0^2}{N^2} (1+\beta^{-1}).
\end{equation} To bound the first term on the right-hand side of \eqref{eq:gamma_gamma0_difference}, we write 
\begin{equation*}
	\frac{1}{e^{\beta(|p|^2-\mu_0)} -1} - \frac{1}{e^{\beta\eps(p)} -1} = \int_0^1 \frac{\beta \left( \eps(p) - (|p|^2-\mu_0) \right)}{ 4\sinh^2 \left ( \frac 1 2 ( t\beta\eps(p) + (1-t)\beta(|p|^2-\mu_0) ) \right )  } \, \mathrm d t.
\end{equation*} With
\begin{equation*}
	\left|  \eps(p) - (|p|^2-\mu_0) \right| = (|p|^2-\mu_0) \left | \sqrt{1+\frac {16\pi a_N N_0} {|p|^2 - \mu_0}  }  - 1 \right| \lesssim \frac{N_0}{N},
\end{equation*} we obtain the bound
\begin{equation}\label{eq:n0_n0tilde_difference_2}
	\left|  \sum_{ p\in P_{\mathrm B}  } \left( \frac{1}{e^{\beta\eps(p)} -1} - \frac{1}{e^{\beta(|p|^2-\mu_0)} -1} \right) \right| \lesssim \frac{\beta N_0}{N} \sum_{ p\in P_{\mathrm B}  } \frac{1}{\sinh^2(\beta |p|^2/2)} \lesssim \frac{N_0}{\beta N}.
\end{equation} In the last step we used $\sinh(x) \geq x$ for $x \geq 0$. In combination, \eqref{eq:gamma_gamma0_difference}, \eqref{eq:n0_n0tilde_difference_1} and \eqref{eq:n0_n0tilde_difference_2} prove \eqref{eq:n_g(z)_ideal_difference}. 
\end{proof}

We now consider the eigenfunctions $\phi_{z,\alpha} = |z\rangle  \otimes T_z \Psi_{\alpha} $. For $k\in \mathbb N$, we introduce the notation 
\begin{equation}\label{eq:rho_za_k_def}
	\begin{split}
		\rho_{z, \alpha}^{(k)} (x_1,...,x_k) = &  \langle a_{x_1}^* ... a_{x_k}^* a_{x_1}...a_{x_k}\rangle_{\phi_{z,\alpha}}.
	\end{split}
\end{equation} In the cases $k=2,4$ we have the following bounds. 

\begin{lemma}\label{lem:rho_z_a_uniform_bound}
	We have
	\begin{align}
		\sup_{ x_1, x_2\in\Lambda }  \rho_{z, \alpha}^{(2)} (x_1,x_2) \le & |z|^4 + 4 |z|^2 N_\alpha + 2N_\alpha(N_\alpha-1) \nonumber \\
		&+C ( |z|^2 + N_\alpha )   ( N_\alpha^< + N^{ \delta_{ \mathrm{ Bog } } } ) + C N^{ 3\delta_{ \mathrm{ Bog } }} ( N_\alpha^{<} + 1 )^2 ,  \label{eq:rho2za_bound} \\
		\sup_{ x_1, x_2\in\Lambda } \left| \nabla_2\rho^{(2)}_{z,\alpha}(x_1,x_2) \right| \lesssim & N_\alpha^{3/2} \| \mathcal K^{1/2} \Psi_\alpha \|  + N^{ 4 \delta_{\mathrm {Bog} } } ( N_\alpha^< + 1 )^2 \nonumber \\
		&+   ( N_\alpha^< + N^{ \delta_{ \mathrm{ Bog } } }) \Big( N_\alpha^{1/2}  \| \mathcal K^{1/2} \Psi_\alpha \| + N_{\alpha} N^{\delta_{\mathrm{Bog}}} \Big) \nonumber \\
		&+ |z|^2 \Big( N_\alpha^{1/2} \| \mathcal K^{1/2} \Psi_\alpha \| + N^{ \delta_{ \mathrm{ Bog } } } ( N_\alpha^< +  N^{ \delta_{ \mathrm{ Bog } } }) \Big), \label{eq:nabla_rho2_bound} \\
		\sup_{ x_1, x_2\in\Lambda }  \int_{\Lambda^2}\rho_{z, \alpha}^{(4)} (x_1,x_2,x_3, x_4) \, \mathrm d x_3 \, \mathrm d x_4 \lesssim & ( |z|^2 + N_\alpha + 1 )^2\left ( ( N_\alpha + |z|^2 ) ^2 +N^{ 3 \delta_{\mathrm{Bog}}} ( N_\alpha^< +1 )^2 \right ). \label{eq:rho4_za_integral_bound} 
	\end{align}
\end{lemma}
\begin{proof} The function $\rho_{z, \alpha}^{(2)} (x_1,x_2)$ is bounded by
	\begin{equation}\label{eq:computation_rho_2_z_a}
		\begin{split}
			\rho_{z, \alpha}^{(2)} (x_1,x_2) = & \sum_{p_1,p_2, q_1, q_2\in\Lambda^*} e^{ \mathrm i (p_1-q_1)x_1 + \mathrm i (p_2-q_2) x_2 } \langle a_{p_1}^* a_{p_2}^* a_{q_1}a_{q_2}\rangle_{ \phi_{z,\alpha} } \\
			\le & \sum_{p_1, p_2, q_1, q_2\in\Lambda^*}  \langle a_{p_1}^* a_{p_2}^* a_{q_1}a_{q_2}\rangle_{ \phi_{z,\alpha} } = |z|^4 + |z|^2\sum_{p,q\in \Lambda^*_+} \left( \frac {z^2}{|z|^2} \langle a_{p}^* a_{q}^* \rangle_{T_z\Psi_\alpha} + \mathrm{h.c.} \right) \\
			& + 4 |z|^2\sum_{p,q\in \Lambda^*_+} \langle a_{p}^* a_{q} \rangle_{T_z\Psi_\alpha} + \sum_{p_1, p_2, q_1, q_2 \in\Lambda^*_+} \langle a_{p_1}^* a_{p_2}^* a_{q_1} a_{q_2} \rangle_{ T_z\Psi_\alpha}.
		\end{split}
	\end{equation} In the last step we used \eqref{eq:action_weyl_momentum} and
		$ \langle a_{p}^* a_{q}^* a_r \rangle_{T_z\Psi_\alpha}= \langle a_{p}^* \rangle_{T_z\Psi_\alpha}=0$ for every $p,q,r\in\Lambda^*_+$, which follows from \eqref{eq:action_bogoliubov} and the fact that $\Psi_\alpha$ is an eigenfunction of $\mathcal N_+$.
	
		We first estimate the quadratic terms on the right-hand side of \eqref{eq:computation_rho_2_z_a}. Since $\Psi_\alpha$ is a symmetrized product of plane waves, we have, for every $k\ge 1$,  $\langle a_{p_1}^* ... a_{p_k}^* a_{q_1} ... a_{q_k} \rangle_{\Psi_\alpha} =0$ unless $\{ p_1, ... , p_k\} = \{ q_1, ... q_k \}$. This, in particular, implies
	\begin{equation}\label{eq:quadratic_expectation_psi_a}
		\begin{split}
			\langle a_{p}^* a_{q} \rangle_{T_z\Psi_{\alpha}} = & \, \delta_{p,q} \Big[ u_p^2   \langle a_{p}^* a_{p} \rangle_{\Psi_{\alpha}} + v_p^2  ( \langle a_{-p}^* a_{-p} \rangle_{\Psi_{\alpha}} + 1) \Big], \\
			\langle a_{p}^* a_{q}^* \rangle_{T_z\Psi_{\alpha}} = & \frac{{\overline z}^2}{|z|^2}u_pv_p \delta_{p,-q} \left(  \langle a_{p}^* a_{p} \rangle_{\Psi_{\alpha}} +  \langle a_{-p}^* a_{-p} \rangle_{\Psi_{\alpha}} + 1 \right).
		\end{split}
	\end{equation} Using Lemma \ref{lem:up_vp_norms} and $N_0\le N$ we estimate 
	\begin{equation}\label{eq:rho_2_z_a_quadratic_bound_1}
		\begin{split}
			\sum_{p,q\in \Lambda^*_+} \left( \langle a_{p}^* a_{q}^* \rangle_{T_z\Psi_\alpha} + \mathrm{h.c.} \right)= 2 \frac{\mathfrak{Re}z^2 }{|z|^2} \sum_{p\in\Lambda_+^*} u_pv_p \left( 2 \langle a_{p}^* a_{p} \rangle_{\Psi_{\alpha}} + 1 \right) \lesssim  N_\alpha^{<} + N^{ \delta_{ \mathrm{ Bog } } }  
		\end{split}
	\end{equation} and 
	\begin{equation}\label{eq:rho_2_z_a_quadratic_bound_2}
		\begin{split}
			\sum_{p,q\in \Lambda^*_+} \langle a_{p}^* a_{q} \rangle_{T_z\Psi_\alpha} = \sum_{p\in\Lambda_+^*} \Big[ (1 + 2 v_p^2) \langle a_{p}^* a_{p} \rangle_{\Psi_{\alpha}} + v_p^2 \Big] \le N_\alpha + C ( N_\alpha^{<} + 1 ) .
		\end{split}
	\end{equation} 
	
	We now estimate the last term on the right-hand side of \eqref{eq:computation_rho_2_z_a}. To do so, we observe that, since $T_z$ only acts on low momenta, the expectation in the sum vanishes if $ | \{ p_1, p_2 \} \cap P_{ \mathrm B } | \neq | \{ q_1, q_2 \} \cap P_{ \mathrm B } |$, and hence
	\begin{equation}\label{eq:quartic_expectation_momentum_split}
		\begin{split}
			\sum_{p_1, p_2, q_1, q_2 \in\Lambda^*_+} & \langle a_{p_1}^* a_{p_2}^* a_{q_1} a_{q_2} \rangle_{ T_z\Psi_\alpha} = \sum_{p_1, p_2, q_1, q_2 \in (\Lambda^*_+ \setminus P_{ \mathrm{ B } }) } \langle a_{p_1}^* a_{p_2}^* a_{q_1} a_{q_2} \rangle_{ T_z\Psi_\alpha} \\
			& + \sum_{p_1, p_2, q_1, q_2 \in P_{ \mathrm{ B } } } \langle a_{p_1}^* a_{p_2}^* a_{q_1} a_{q_2} \rangle_{ T_z\Psi_\alpha} + 4 \sum_{\substack{p_1, q_1 \in (\Lambda^*_+ \setminus P_{ \mathrm{ B } }) \\ p_2, q_2 \in P_{ \mathrm{ B } } }} \langle a_{p_1}^* a_{p_2}^* a_{q_1} a_{q_2} \rangle_{ T_z\Psi_\alpha}.
		\end{split}
	\end{equation} 
	
	The first term on the right-hand side equals 
	\begin{equation*}
		\sum_{p_1, p_2, q_1, q_2 \in (\Lambda^*_+ \setminus P_{ \mathrm{ B } }) } \langle a_{p_1}^* a_{p_2}^* a_{q_1} a_{q_2} \rangle_{\Psi_\alpha} = 2 N_\alpha^> ( N_\alpha^> - 1) \le 2 N_\alpha(N_\alpha -1 ). 
	\end{equation*} 
	As for the second term, using the translation invariance of $T_z\Psi_\alpha$, the Cauchy--Schwarz inequality and \eqref{eq:Tz_N_bound}, we find 
	\begin{equation*}
		\begin{split}
			\sum_{p_1, p_2, q_1, q_2 \in P_{ \mathrm{ B } } } \langle a_{p_1}^* a_{p_2}^* a_{q_1} a_{q_2} \rangle_{ T_z\Psi_\alpha} = & \sum_{ \substack{ p_1, p_2, q_1 \in P_{ \mathrm{ B } } \\  p_1+p_2-q_1 \in P_{ \mathrm{ B } } } } \langle a_{p_1}^* a_{p_2}^* a_{q_1} a_{p_1+p_2-q_1} \rangle_{ T_z\Psi_\alpha} \\
			\le & 2 | P_{ \mathrm{ B } } | \sum_{ p_1, p_2 \in P_{ \mathrm{ B } } } \langle a_{p_1}^* a_{p_2}^* a_{p_1} a_{p_2} \rangle_{ T_z\Psi_\alpha} \lesssim N^{ 3 \delta_{\mathrm {Bog} } } ( N_\alpha^< + 1 )^2.
		\end{split}
	\end{equation*}
	The last term on the right-hand side of \eqref{eq:quartic_expectation_momentum_split} equals 
	\begin{equation*}
		N_\alpha^> \sum_{p,q\in P_{ \mathrm B } } \langle a_{p}^* a_{q} \rangle_{ T_z\Psi_\alpha} \lesssim N_\alpha  ( N_\alpha^< +  N^{ \delta_{ \mathrm{ Bog } } } ) ,
	\end{equation*} where we used \eqref{eq:rho_2_z_a_quadratic_bound_2} in the last step. 
	Putting these considerations together, we find
	\begin{equation*}
		\sum_{p_1, p_2, q_1, q_2 \in\Lambda^*_+}  \langle a_{p_1}^* a_{p_2}^* a_{q_1} a_{q_2} \rangle_{ T_z\Psi_\alpha} \le 2N_\alpha(N_\alpha-1) + C  N_\alpha ( N_\alpha^< + N^{ \delta_{ \mathrm{ Bog } } } ) + C N^{ 3\delta_{ \mathrm{ Bog } }}  ( N_\alpha^{<} + 1 )^2,
	\end{equation*} which combined with \eqref{eq:computation_rho_2_z_a}, \eqref{eq:rho_2_z_a_quadratic_bound_1} and \eqref{eq:rho_2_z_a_quadratic_bound_2} implies \eqref{eq:rho2za_bound}.

	We now show \eqref{eq:nabla_rho2_bound}. We start by taking the gradient of \eqref{eq:computation_rho_2_z_a}:
	\begin{equation}\label{eq:nabla_rho2_1}
		\begin{split}
			- \mathrm{i} \nabla_2\rho^{(2)}_{z,\alpha}(x,y) = & \sum_{p_1, p_2, q_1, q_2 \in\Lambda^*} (p_2-q_2) e^{ \mathrm i (p_1-q_1)x +  \mathrm i (p_2-q_2)y} \langle a_{p_1}^* a_{p_2}^* a_{q_1} a_{q_2} \rangle_{\phi_{z,\alpha}} \\ 
			= &  \sum_{p_1, p_2, q_1, q_2 \in\Lambda_+^*}  (p_2-q_2) e^{ \mathrm i (p_1-q_1)x +  \mathrm i (p_2-q_2)y} \langle a_{p_1}^* a_{p_2}^* a_{q_1} a_{q_2} \rangle_{T_z\Psi_\alpha} \\ 
			& +  |z|^2  \sum_{p, q \in\Lambda_+^*} \Bigg( q e^{ \mathrm i p\cdot x +  \mathrm i q\cdot y} \langle a_{p}^* a_{q}^* \rangle _{T_z\Psi_\alpha} \\
			& \qquad - q e^{ \mathrm i p\cdot x -  \mathrm i q\cdot y} \langle a_{p}^* a_{q}  \rangle_{T_z\Psi_\alpha} + p e^{- \mathrm i q\cdot x +  \mathrm i p\cdot y}  \langle a_{p}^* a_{q} \rangle_{T_z\Psi_\alpha} \\
			& \qquad + (p-q)  e^{   \mathrm i (p - q)\cdot y} \langle a_{p}^* a_{q}\rangle_{T_z\Psi_\alpha} - qe^{- \mathrm i p\cdot x -  \mathrm i  q\cdot y} \langle a_{p} a_{q} \rangle_{T_z\Psi_\alpha} \Bigg) . 
		\end{split}
	\end{equation} 
	With a similar argument as in \eqref{eq:quartic_expectation_momentum_split}, we see that the absolute value of first term on the right-hand side is bounded by 
	\begin{equation}\label{eq:quartic_expectation_momentum_split_nabla}
		\begin{split}
			&\sum_{p_1, p_2, q_1, q_2 \in\Lambda^*_+} |p_1|  \langle a_{p_1}^* a_{p_2}^* a_{q_1} a_{q_2} \rangle_{ T_z\Psi_\alpha} \le \sum_{p_1, p_2, q_1, q_2 \in (\Lambda^*_+ \setminus P_{ \mathrm{ B } }) } |p_1| \langle a_{p_1}^* a_{p_2}^* a_{q_1} a_{q_2} \rangle_{ T_z\Psi_\alpha} \\
			& + N^{\delta_{ \mathrm{ Bog } }}  \sum_{p_1, p_2, q_1, q_2 \in P_{ \mathrm{ B } } } \langle a_{p_1}^* a_{p_2}^* a_{q_1} a_{q_2} \rangle_{ T_z\Psi_\alpha} + 2 \sum_{\substack{p_1, q_1 \in (\Lambda^*_+ \setminus P_{ \mathrm{ B } }) \\ p_2, q_2 \in P_{ \mathrm{ B } } }} ( |p_1| + N^{\delta_{\mathrm{ Bog }}} ) \langle a_{p_1}^* a_{p_2}^* a_{q_1} a_{q_2} \rangle_{ T_z\Psi_\alpha}.
		\end{split}
	\end{equation} With the Cauchy--Schwarz inequality we find 
	\begin{equation*}
		\begin{split}
			\sum_{p_1, p_2, q_1, q_2 \in (\Lambda^*_+ \setminus P_{ \mathrm{ B } }) } & |p_1| \langle a_{p_1}^* a_{p_2}^* a_{q_1} a_{q_2} \rangle_{ T_z\Psi_\alpha} \\
			\le & \Big( \sum_{p,q\in\Lambda_+^*} |p|^2 \langle a_{p}^* a_{q}^* a_{p} a_{q} \rangle_{\Psi_\alpha} \Big )^{1/2} \Big ( \sum_{p,q\in\Lambda_+^*} \langle a_{p}^* a_{q}^* a_{p} a_{q} \rangle_{\Psi_\alpha} \Big )^{1/2} \lesssim N_\alpha^{3/2} \| \mathcal K^{1/2} \Psi_\alpha \| .
		\end{split}
	\end{equation*} The remaining terms on the right-hand side of \eqref{eq:quartic_expectation_momentum_split_nabla} are bounded similarly, and we get
	\begin{equation*}
		\begin{split}
			&\sum_{p_1, p_2, q_1, q_2 \in\Lambda^*_+} |p_1|  \langle a_{p_1}^* a_{p_2}^* a_{q_1} a_{q_2} \rangle_{ T_z\Psi_\alpha} \\
			&\lesssim N_\alpha^{3/2} \| \mathcal K^{1/2} \Psi_\alpha \|  + N^{ 4 \delta_{\mathrm {Bog} } } ( N_\alpha^< + 1 )^2 +  ( N_\alpha^< + N^{ \delta_{ \mathrm{ Bog } } }) \Big( N_\alpha^{1/2}  \| \mathcal K^{1/2} \Psi_\alpha \| + N_{\alpha} N^{\delta_{\mathrm{Bog}}} \Big).
		\end{split}
	\end{equation*}
	
	The quadratic terms on the right-hand side of \eqref{eq:nabla_rho2_1} can be estimated easily using the explicit formulas in \eqref{eq:quadratic_expectation_psi_a}. Doing so, we see that they are bounded by 
	\begin{equation*}
		C  |z|^2 \Big( N_\alpha^{1/2} \| \mathcal K^{1/2} \Psi_\alpha \| + N^{ \delta_{ \mathrm{ Bog } } }  ( N_\alpha^< +   N^{ \delta_{ \mathrm{ Bog } } } ) \Big).
	\end{equation*} Combining the previous two bounds proves \eqref{eq:nabla_rho2_bound}.
	
	Finally, we prove \eqref{eq:rho4_za_integral_bound}. By \eqref{eq:action_weyl_momentum} we have 
	\begin{equation}\label{eq:rho_za_4_computation}
		\begin{split}
			\int_{\Lambda^2} & \rho_{z, \alpha}^{(4)} (x_1,x_2,x_3, x_4) \, \mathrm d x_3 \, \mathrm d x_4 = \sum_{p_1,p_2, q_1, q_2 \in\Lambda^*} e^{ \mathrm i(p_1-q_1)x_1 + \mathrm i(p_2-q_2) x_2 } \langle a_{p_1}^* a_{p_2}^* \mathcal N (\mathcal  N - 1) a_{q_1}a_{q_2}\rangle_{\phi_{z,\alpha}} \\
			\le & |z|^4 \langle \mathcal N(\mathcal N -1) \rangle_{\phi_{z, \alpha}} + |z|^2 \Bigg [  4 \sum_{p_1, q_1\in\Lambda^*_+} \langle a_{p_1}^* \mathcal N (\mathcal  N - 1) a_{q_1} \rangle_{\phi_{z,\alpha}}  \\
			& + \Big(\sum_{p_1, p_2\in\Lambda^*_+} \langle a_{p_1}^* a_{p_2}^* \mathcal N (\mathcal  N - 1) \rangle_{\phi_{z,\alpha}}+ \mathrm{h.c.} \Big ) \Bigg] + \sum_{p_1,p_2, q_1, q_2\in\Lambda^*_+}  \langle a_{p_1}^* a_{p_2}^* \mathcal N (\mathcal  N - 1) a_{q_1}a_{q_2}\rangle_{\phi_{z,\alpha}} .
		\end{split}
	\end{equation} The first term on the right-hand side of \eqref{eq:rho_za_4_computation} can be bounded by a constant times $|z|^4( N_\alpha + |z|^2 + 1)^2$ using Lemma \ref{lem:Tz_N_bound}. By the translation invariance of the state $|\phi_{z,\alpha}\rangle \langle \phi_{z, \alpha} |$ and Lemma \ref{lem:Tz_N_bound}, we see that the second term equals
	\begin{equation*}
		4|z|^2\sum_{p \in\Lambda^*_+} \langle a_{p}^* \mathcal N (\mathcal  N - 1) a_{p} \rangle_{\phi_{z,\alpha}}  \lesssim |z|^2(N_\alpha + |z|^2 +1 )^3 ,
	\end{equation*} and the same bound holds for the third term.
	
	To estimate the last term on the right-hand side of \eqref{eq:rho_za_4_computation} we split the sum in momentum sets as in \eqref{eq:quartic_expectation_momentum_split}: 
	\begin{equation*}
		\begin{split}
			\sum&_{p_1, p_2, q_1, q_2 \in\Lambda^*_+}  \langle a_{p_1}^* a_{p_2}^* \mathcal N( \mathcal N -1 ) a_{q_1} a_{q_2} \rangle_{ \phi_{z,\alpha} } = \sum_{p_1, p_2, q_1, q_2 \in (\Lambda^*_+ \setminus P_{ \mathrm{ B } }) } \langle a_{p_1}^* a_{p_2}^*  \mathcal N( \mathcal N -1 )  a_{q_1} a_{q_2} \rangle_{ \phi_{z,\alpha} } \\
			& + \sum_{p_1, p_2, q_1, q_2 \in P_{ \mathrm{ B } } } \langle a_{p_1}^* a_{p_2}^*  \mathcal N( \mathcal N -1 ) a_{q_1} a_{q_2} \rangle_{ \phi_{z,\alpha} } + 4 \sum_{\substack{p_1, q_1 \in (\Lambda^*_+ \setminus P_{ \mathrm{ B } }) \\ p_2, q_2 \in P_{ \mathrm{ B } } }} \langle a_{p_1}^* a_{p_2}^*  \mathcal N( \mathcal N -1 )  a_{q_1} a_{q_2} \rangle_{ \phi_{z,\alpha} }.
		\end{split}
	\end{equation*} The first and the third term are easily seen to be bounded by a constant times $(N_\alpha + |z|^2 +1)^4$, by Lemma \ref{lem:Tz_N_bound}. The second term can be estimated using translation invariance, Cauchy--Schwarz, and Lemma~\ref{lem:Tz_N_bound}: 
	\begin{equation*}
		\begin{split}
			\sum_{p_1, p_2, q_1, q_2 \in P_{ \mathrm{ B } } } & \langle a_{p_1}^* a_{p_2}^*  \mathcal N( \mathcal N -1 ) a_{q_1} a_{q_2} \rangle_{ \phi_{z,\alpha} } = \sum_{p_1, p_2, q_1 \in P_{ \mathrm{ B } } } \langle a_{p_1}^* a_{p_2}^*  \mathcal N( \mathcal N -1 ) a_{q_1} a_{p_1+p_2 - q_1} \rangle_{ \phi_{z,\alpha} } \\
			\lesssim & N^{3 \delta_{ \mathrm{ Bog } }}  \sum_{p_1, p_2 \in P_{ \mathrm{ B } } } \langle a_{p_1}^* a_{p_2}^*  \mathcal N^2 a_{p_1} a_{p_2} \rangle_{ \phi_{z, \alpha } } \le N^{3 \delta_{ \mathrm{ Bog } } } \| ( \mathcal N^<)^2  \phi_{z, \alpha }  \| \cdot \| \mathcal N^2 \phi_{z,\alpha} \| \\
			\lesssim & N^{3\delta_{ \mathrm{ Bog } }} ( N_\alpha^< + 1 )^2 ( N_\alpha + |z|^2 + 1 )^2 .
		\end{split}
	\end{equation*} Collecting the above bounds, we find \eqref{eq:rho4_za_integral_bound}. 
\end{proof} 

The following statement quantifies how the Jastrow factor changes the norm of an eigenfunction of $| z \rangle \langle z | \otimes G(z)$.
\begin{corollary}
	Let $\beta \gtrsim \beta_{\mathrm{c}}$ and $\delta_{\mathrm{ Bog }} \le 2/15$. There exists $C>0$ such that
	\begin{equation}\label{eq:denominator_lower_bound}
		\| F \phi_{\alpha, z} \| ^2 \ge 1 - CN\ell^2,
	\end{equation} for every $|z|^2\le \widetilde c N$, $\alpha\in\widetilde{\mathcal A}$ and $N\in\mathbb N$.
\end{corollary}
\begin{proof} We use the inequality
	\begin{equation}\label{eq:product_lower_bound}
		\prod_{i<j}\left( 1- u_\ell(x_i-x_j) \right) \ge 1 -\sum_{i<j} u_\ell(x_i-x_j),
	\end{equation} for $u_\ell = 1- f_\ell^2 \ge 0$,	to estimate
	\begin{equation} \label{eq:denominator_bound_first}
		\begin{split}
			\| F  \phi_{\alpha,z} \| = \sum_n \int_{\Lambda^n} \big| \phi_{\alpha,z}^{(n)} \big|^2\prod_{i<j}\left( 1- u_\ell(x_i-x_j) \right) \ge 1 - \frac 1 2 \int_{\Lambda^2} \rho_{\alpha, z}^{(2)}(x,y) u_\ell(x-y) \mathrm dx \mathrm dy.  
		\end{split}
	\end{equation} By \eqref{eq:rho2za_bound} and the assumptions $\beta \gtrsim \beta_{\mathrm{c}}$, $\delta_{\mathrm{ Bog }} \le 2/15$, $\alpha \in \widetilde { \mathcal A }$ and $|z|^2 \le \widetilde c N$, we have
	\begin{equation*}
		\rho^{(2)}_{z,\alpha}(x_1, x_2) \lesssim N^2 + N^{5\delta_{ \mathrm{ Bog } }}\beta^{-2} \lesssim N^2.
	\end{equation*} We insert this into \eqref{eq:denominator_bound_first}, use \eqref{eq:ul_integral} and $\mathfrak a_N N \lesssim 1$, and find \eqref{eq:denominator_lower_bound}.
\end{proof} 

\subsubsection{The uncorrelated trial state} 

	Finally, we need some estimates on the densities of the uncorrelated trial state $\Gamma_0$. For $k\in\mathbb N$, we define
	\begin{equation}\label{eq:rho_k_gamma0_def}
		\begin{split}
			\rho_{\Gamma_0}^{(k)}(x_1, ... x_k) = &  \tr[a_{x_1}^*...a_{x_k}^*a_{x_1}...a_{x_k} \Gamma_0].
		\end{split}
	\end{equation} We have the following. 

	\begin{lemma} \label{lem:rho_k_Gamma_bound} Let $\beta \gtrsim \beta_{\mathrm c}$, $k \ge 2$ and $i\in \{1,...,k\}$. Then the bounds
	\begin{equation}\label{eq:rho_k_Gamma0_bounds}
		\begin{split}
			\sup_{x_1, ... , x_k\in\Lambda}|\rho_{\Gamma_0}^{(k)}(x_1, ..., x_k)| \lesssim_k & N^k, \\
				\sup_{x_1, ... , x_k\in\Lambda}| \nabla_i {\rho}_{\Gamma_0}^{(k)} (x_1, ... x_k) | \lesssim_k & \beta^{-1/2} N^k,
		\end{split}
	\end{equation}
	hold.
\end{lemma}
\begin{proof}
	For $x \in \Lambda$ and $z\in \mathbb C$ we introduce the operators $a_{x,z}^* = \frac{z}{|z|}a^*_x$, $ a_{x,z} = \frac{ \overline z }{|z|}a_x $, which satisfy 
	\begin{equation}\label{eq:action_Weyl_z}
		W_z^* a_{x, z}^* W_z = a_{x, z}^* + |z|, \qquad W_z^* a_{x, z} W_z = a_{x, z} + |z|
	\end{equation} 
	by \eqref{eq:action_weyl_position}.  With the notation $y_i = y_{i+k} = x_i$, $\sharp_i = *$, $\sharp_{i+k} = \cdot$ for $i = 1,..., k$, we can write the $k$-body density of $\Gamma_0$ as 
	\begin{equation}\label{eq:rho2_gamma0_expansion}
		\begin{split}
			\rho_{\Gamma_0}^{(k)}(x_1, ... , x_k) = \int_{ \mathbb C} \zeta(z)  \tr_{\mathfrak F} \big[ a_{y_1, z}^{\sharp_1}... a_{y_{2k}, z}^{\sharp_{2k}} |z\rangle\langle z| \otimes G(z) \big] \, \mathrm d z.
		\end{split}
	\end{equation} Using \eqref{eq:action_Weyl_z} and Wick's theorem for the quasi-free state $G(z)$, we find
	\begin{equation*}
\begin{split}
		\tr_{\mathfrak F} \big[ a_{y_1, z}^{\sharp_1}... a_{y_{2k}, z}^{\sharp_{2k}} & |z\rangle\langle z| \otimes G(z) \big] \\
		= & |z|^{2k} + \sum_{h = 1}^k |z|^{2(k - h)} \sum_{ 1 \le i_1 < ... < i_{2h} \le 2k } \, \sum_{ \sigma \in P_{2h} } \prod_{j = 1}^h  \tr_{\mathfrak F_+} \Big[ a^{\sharp_{ i_{\sigma(2j-1)} }}_{y_{ i_{\sigma(2j-1)}, z }} a^{\sharp_{i_{\sigma(2j)}}}_{y_{ i_{\sigma(2j)} , z}} G(z) \Big] ,
\end{split}
	\end{equation*} where 
\begin{equation*}
	P_{2h} = \{ \sigma \in S_{2h} \mid \sigma(2j-1) < \sigma(2j), \, \sigma(2j-1) < \sigma(2j + 1) \}
\end{equation*} denotes the set of pairings of $\{ 1, ... , 2h \}$. Lemma~\ref{lem:1pdm_pairing_bogoliubov} implies that for every $\sharp, \flat \in \{ *, \cdot\}$ there exists $f_z^{(\sharp, \flat)}: \Lambda\to\mathbb C$ such that $ \tr_{ \mathfrak{F}_+ }[a^{\sharp}_{x,z} a^{\flat}_{y, z} G(z)] = f_{z}^{(\sharp, \flat)}( x- y ) $. Moreover, the functions $f_z^{(\sharp, \flat)}$ satisfy the bounds 
\begin{equation*}
	\| f_z^{(\sharp, \flat)} \|_\infty \lesssim N \quad \text{ and } \quad \| \nabla f_z^{(\sharp, \flat)} \|_\infty \lesssim \beta^{-1/2} N,
\end{equation*}
uniformly in $z\in \mathbb C$. Taking into account the particle number cutoff in $\zeta(z)$, we thus find
\begin{equation*}
\begin{split}
	 | \rho_{\Gamma_0}^{(k)}(x_1, ... , x_k) | \lesssim & \int_{\mathbb C} \zeta(z) \Big( |z|^{2k} + \sum_{h = 1}^{k} |z|^{2(k - h)} \binom{2k}{2h} |P_{2h}| N^h \Big) \, \mathrm dz \lesssim_k N^{k}, \\
 | \nabla_i \rho_{\Gamma_0}^{(k)}(x_1, ... , x_k) | \lesssim &  \sum_{h = 1}^{k} \int_{\mathbb C} \zeta(z) |z|^{2(k - h)} \binom{2k}{2h} |P_{2h}|  \beta^{-1/2} N^h \, \mathrm dz \lesssim_k \beta^{-1/2}N^{k},
\end{split}
\end{equation*} which is the claim of the lemma. 
\end{proof}

Let us introduce the notation
\begin{equation} \label{eq:n0_tilde_def}
	\widetilde N_0 = \tr[ a_0^* a_0 \widetilde{\Gamma}_0 ] = \int_{\mathbb C} |z|^2 \zeta(z) \, \mathrm d z
\end{equation} 
for the expected number of particles in the condensate of $\widetilde \Gamma_0$ in \eqref{eq:gamma_0_def}, with $\widetilde \mu$ chosen according to Lemma \ref{lem:perturbation_number_particles}. As a corollary to Lemmas~\ref{lem:perturbation_number_particles}, \ref{lem:cutoff_a_remainder} and \ref{lem:n_g(z)_ideal_difference} we prove the following estimate for the difference between $\widetilde N_0$ and the expected number of particles in the condensate of the ideal gas. 

\begin{corollary} \label{lem:n0_n0tilde}
	Let the assumptions of Lemma \ref{lem:perturbation_number_particles} be satisfied. Then
	\begin{equation}\label{eq:N0_N0tilde_diff}
		| N_0 - \widetilde N_0 | \lesssim \frac{N_0}{\beta N}  + \frac{N_0^2}{N^2}  + N^{3}\ell^4 + N^{ 1 + \delta_{ \mathrm{ Bog } }}\ell^2 (\beta^{-1} + 1)
	\end{equation}
	holds with $N_0$ in \eqref{eq:n0_def}.
\end{corollary}
\begin{proof}
	We have 
	\begin{equation*}
		\begin{split}
			N_0 - \widetilde N_0 = & \tr_{ \mathfrak{F}_+ } [ \mathcal N_+  G(z) ] - \tr_{ \mathfrak{F}_+ } [ \mathcal N_+ G^{ \mathrm {id} } ] + \tr_{ \mathfrak{F}_+ } [ \mathcal N_+ \widetilde G(z) ] - \tr_{ \mathfrak{F}_+ } [ \mathcal N_+ G(z) ] \\
			& + \tr_{ \mathfrak{F} } [ \mathcal N \Gamma ] - \tr_{ \mathfrak{F} } [ \mathcal N \widetilde \Gamma_0 ]
		\end{split}
	\end{equation*} 
	with $G, \widetilde{G}$ in \eqref{eq:gibbs_state_bog}, $\widetilde{\Gamma}_0$ in \eqref{eq:gamma_0_def}, $\Gamma$ in \eqref{eq:trial_state} and $G^{\mathrm{id}}$ in \eqref{eq:GibbsStateIdealGas}. In combination with Lemmas~\ref{lem:perturbation_number_particles}, \ref{lem:cutoff_a_remainder} and \ref{lem:n_g(z)_ideal_difference}, this identity proves the claim.
\end{proof}

\section{Estimate for the energy}\label{sec:energy}

Throughout Sections~\ref{sec:energy},~\ref{sec:entropy} and \ref{sec:proof_main_theorem}, we adopt the assumptions of Lemma \ref{lem:perturbation_number_particles}, and take $N$ sufficiently large so that the trial state $\Gamma$ is well defined. In addition, we suppose that $\beta \sim \beta_{\mathrm{c}}$,  with $\beta_{\mathrm{c}}$ in \eqref{eq:Tc_ideal}, and we fix a sufficiently large cutoff parameter $\widetilde c  >2$ in \eqref{eq:A_tilde_def} such that the statements of Lemma \ref{lem:cutoff_a_remainder} hold. Finally, we assume that $N \ell^2$ is sufficiently small. In this section we will prove the following statement. 

\begin{proposition}\label{prop:energy_upper_bound}
	The energy of $\Gamma$ is given by 
	\begin{equation}\label{eq:energy_upper_bound}
		\begin{split}
			\tr[\mathcal H_N \Gamma] = & \tr [ \mathcal H^{ \mathrm{ diag } } G^{ \mathrm{ diag} }  ] + \mu_0\sum_{p\in \Lambda_+^*} \gamma_p + 4\pi \mathfrak a _N \int_{\mathbb C } \zeta(z) |z|^4 \, \mathrm d z \\
			& +   8\pi \mathfrak a_N  \Bigg[  \Big( \sum_{p\in \Lambda_+^*} \gamma_p\Big)^2 + \widetilde{N}_0  \sum_{ p \in \Lambda_+^* \setminus  P_{\mathrm B} } \gamma(p) + \widetilde{N}_0 \sum_{p\in \Lambda_+^*} \gamma_p  \Bigg]  + \mathcal E_{\mathcal H},
		\end{split}
	\end{equation}
	with 
	\begin{equation}\label{eq:energy_error_bound}
		\mathcal E_{\mathcal H} \lesssim  N^{1/3 + \delta_{ \mathrm{ Bog } }} + N^{2/3 + 3\delta_{ \mathrm{ Bog } }} \ell^{1/2} + N^{5/3 + 3\delta_{ \mathrm{ Bog } }} \ell^2 + N^3\ell^4 + \ell^{-1}.
	\end{equation}
\end{proposition}

To prove the above proposition, we decompose
\begin{equation*}
	\tr[\mathcal H_N \Gamma] = \tr[\mathcal K \Gamma] + \tr[\mathcal V_N \Gamma],
\end{equation*} where
\begin{equation*}
	\mathcal K = \int_{\Lambda} a_x^*(-\Delta) a_x \, \mathrm d x \quad \text{ and } \quad \mathcal{V}_N = \int_{\Lambda^2} v_N(x-y) a_x^*a_y^* a_x a_y\, \mathrm dx\,\mathrm dy
\end{equation*} 
denote the kinetic energy operator and the interaction potential, respectively. We first compute the contribution of the kinetic term.

\subsection{Analysis of the kinetic energy}

The goal of this section is to prove the following proposition. 

\begin{proposition}\label{prop:kinetic_energy_upper_bound}
We have
	\begin{equation}\label{eq:kinetic_energy_contribution}
		\tr[\mathcal K\Gamma] = \tr[\mathcal K \widetilde \Gamma_0 ] + \tr[ \cchi \Gamma] + \mathcal E_{\mathcal K},
	\end{equation} where 
	\begin{equation}\label{eq:chi_operator_def}
	\cchi \coloneqq \int_{\Lambda^2} \frac{|\nabla f_\ell(x-y)|^2}{f_\ell(x-y)^2} a_x^* a_y^* a_x a_y \, \mathrm d x \, \mathrm d y
	\end{equation} 
	with $f_{\ell}$ in \eqref{eq:Definitionfl}. The error term satisfies
	\begin{equation}\label{eq:kinetic_energy_error_bound}
	 \mathcal E_\mathcal K \lesssim N^{2/3 + 3\delta_{\mathrm{Bog}}} \sqrt { \ell(1 + N^2\ell^3) }.
	\end{equation}
\end{proposition}

\begin{remark}
	 The operator $\cchi$ contains the leading-order contribution of the correlation structure to the kinetic energy of the trial state. We will combine it with contributions from the potential $\mathcal V_N$ to obtain the full interaction energy. 
\end{remark}

To prove Proposition \ref{prop:kinetic_energy_upper_bound}, we start by expanding the trace on the left-hand side of \eqref{eq:kinetic_energy_contribution} in particle number sectors. With the notation $\mathrm d X^n = \mathrm d x_1\, ... \, \mathrm d x_n$, we write 
\begin{equation}\label{eq:kinetic_energy_computation_def}
	\begin{split}
		\tr[\mathcal K \Gamma] = & \int_{\mathbb C } \zeta(z)  \sum_{\alpha\in \widetilde{ \mathcal A}}\frac {  \widetilde \lambda_\alpha } { \| F \phi_{z,\alpha} \|^2 } \sum_{n =1}^\infty \sum_{i=1}^n \int_{\Lambda^n} \left| \nabla_i \big( F_n  \phi_{z,\alpha}^n \big)\right |^2 \, \mathrm d X^n \, \mathrm d z,
	\end{split}
\end{equation}
where $\phi_{z,\alpha}^n$ denotes the projection of $\phi_{z, \alpha}$ onto the $n$-particle sector of the Fock space. Integrating by parts, we find
\begin{equation}\label{eq:kinetic_energy_first_partial}
	\begin{split}
		\int_{\Lambda^n}  \big| \nabla_i \big(F_n &  \phi_{z,\alpha}^n \big) \big |^2  \, \mathrm d X^n = - \int_{\Lambda^n} F_n \overline{ \phi_{z,\alpha}^n }\Delta_i \big (F_n \phi_{z,\alpha}^n \big ) \, \mathrm d X^n \\
		= & \int_{\Lambda^n} \Big[ (-\Delta_i F_n ) F_n \big | \phi_{z,\alpha}^n \big|^2  - 2 \nabla_i F_n  \cdot  \nabla_i \phi_{z,\alpha}^n F_n \overline { \phi_{z,\alpha}^n } + F_n^2 \overline { \phi_{z,\alpha}^n  }  (-\Delta_i \phi_{z,\alpha}^n)  \Big ] \, \mathrm d X^n.
	\end{split}
\end{equation} A further integration by parts shows
\begin{equation*}
	\begin{split}
		\int_{\Lambda^n}  (-\Delta_i F_n) F_n \big |  \phi_{z,\alpha}^{n} \big|^2 \, \mathrm d X^n = & 	\int_{\Lambda^n} |\nabla_i F_n |^2 \big |  \phi_{z,\alpha}^{n}  \big |^2 \, \mathrm d X^n \\
		& + 2\mathfrak{Re}	\int_{\Lambda^n} F_n  \left(\nabla_i F_n \cdot \overline { \nabla_i \phi_{z,\alpha}^{n} } \right)  \phi_{z,\alpha}^{n} \, \mathrm d X^n.
	\end{split}
\end{equation*} When we plug this into \eqref{eq:kinetic_energy_first_partial} and take the real part on both sides, we get
\begin{equation*}\label{eq:kinetic_term_integrated_parts}
	\begin{split}
		\int_{\Lambda^n} & \left| \nabla_i \big( F_n  \phi_{z,\alpha}^{n} \big) \right |^2  \, \mathrm d X^n = \int_{\Lambda^n} \left| \nabla_i F_n \right|^2 \big|\phi_{z,\alpha}^{n} \big |^2  \, \mathrm d X^n + \mathfrak{Re} \int_{\Lambda^n} F_n ^2 \overline{ \phi_{z,\alpha}^{n} } \big (- \Delta_i \phi_{z,\alpha}^{n} \big)  \, \mathrm d X^n. 
	\end{split}
\end{equation*}
Inserted into \eqref{eq:kinetic_energy_computation_def}, this yields $\tr[\mathcal K\Gamma] = K_1 + K_2$ with 
\begin{equation*}
\begin{split}
		K_1 \coloneqq & \int_{\mathbb C } \zeta(z)  \sum_{\alpha\in \widetilde{ \mathcal A}}\frac {  \widetilde \lambda_\alpha } { \| F \phi_{z,\alpha} \|^2 } \sum_{n =1}^\infty \sum_{i=1}^n \int_{\Lambda^n}  \left| \nabla_i F_n \right|^2 \left| \phi_{z,\alpha}^{n} \right |^2  \, \mathrm d X^n \, \mathrm d z, \\
		K_2 \coloneqq & \mathfrak{Re} \int_{\mathbb C } \zeta(z)  \sum_{\alpha\in \widetilde{ \mathcal A}}\frac {  \widetilde \lambda_\alpha } { \| F \phi_{z,\alpha} \|^2 } \langle F^2 \phi_{z,\alpha} , \mathcal K\phi_{z,\alpha} \rangle \,\mathrm d z. 
\end{split}
\end{equation*}

In the next two lemmas we provide estimates for $K_1$ and $K_2$.
\begin{lemma}\label{lem:K1_upper_bound}
	We have
	\begin{equation}\label{eq:K1_upper_bound}
		K_1 = \tr[\cchi\Gamma] + \mathcal E^{(1)}_{\mathcal K},
	\end{equation} with 
	\begin{equation}\label{eq:K1_error_bound}
		 \mathcal E^{(1)}_{\mathcal K} \lesssim N\ell^2.
	\end{equation}
\end{lemma}
\begin{proof}
	 Let us introduce the notation 
	 \begin{equation}\label{eq:notation_fij}
	 	f_{ij} = f_\ell(x_i-x_j), \qquad u_{ij} = u_\ell(x_i-x_j), \qquad \nabla f_{ij} = \nabla f_\ell(x_i-x_j),
	 \end{equation} where, we recall, $u_\ell = 1- f_\ell^2$. We have $\nabla_i F_n = F_n \sum_{ j\neq i} \nabla f_{ij} / f_{ij}$,  which implies 
	\begin{equation*}
		|\nabla_i F_n|^2 = F_n ^ 2 \Bigg\{ \sum_{\substack{1\le j \le n \\ j\neq i}} \frac{ | \nabla f_{ij} | ^2 }{f_{ij} ^2 }  + \sum_{\substack{1\le j, k \le n \\ j\neq i, k \neq i, j }} \frac{ \nabla f_{ij} \cdot \nabla f_{ik}}{f_{ij} f_{ik} } \Bigg\} .
	\end{equation*} The first term, inserted into the definition of $K_1$, gives $\tr[\cchi\Gamma]$, thus 
	\begin{equation*}
\begin{split}
		\mathcal E^{(1)}_{\mathcal K} = & \int_{\mathbb C } \zeta(z)  \sum_{\alpha\in \widetilde{ \mathcal A}}\frac {  \widetilde \lambda_\alpha } { \| F \phi_{z, \alpha} \|^2 } \sum_{n =1}^\infty \int_{\Lambda^n} F_n ^ 2  \big| \phi_{z, \alpha}^{n} \big |^2 \sum_{\substack{1\le i, j, k \le n \\ i\neq j, k \neq i, j }} \frac{ \nabla f_{ij} \cdot \nabla f_{ik}}{f_{ij} f_{ik} }  \, \mathrm dX^n \, \mathrm d z. 
\end{split}
	\end{equation*} In combination, \eqref{eq:denominator_lower_bound}, the inequality $(1-x)^{-1} \le 1 + x + 2x^2$, valid for $0 \le x \le 1/2$, and $F_n^2 f_{ij}^{-1} f_{ik}^{-1}\le 1$, imply
\begin{equation}
\begin{split}
  	\mathcal E^{(1)}_{\mathcal K} \le & (1+C N\ell^2) \int_{\mathbb C } \zeta(z)  \sum_{\alpha\in \widetilde{ \mathcal A}} \widetilde \lambda_\alpha \sum_{n=1}^\infty \int_{\Lambda^n} \sum_{\substack{1\le i, j, k \le n \\ i\neq j, k \neq i, j }} |\nabla f_{ij}| | \nabla f_{ik} | \big| \phi_{z, \alpha}^{n} \big|^2  \, \mathrm dX^n \, \mathrm d z \\
	\lesssim & \kappa_0^{-1} \int_{\Lambda^3} \rho_{\Gamma_0}^{(3)}(x_1,x_2,x_3) |\nabla f(x_1-x_2) | |\nabla f(x_1 - x_3)|\mathrm dx_1 \mathrm d x_2 \mathrm d x_3 \le \kappa_0^{-1} \| \rho_{\Gamma_0}^{(3)} \|_\infty \| \nabla f_l \|_1^2,
\end{split}
\end{equation} as long as $N\ell^2$ is sufficiently small. The estimate in \eqref{eq:K1_error_bound} follows using \eqref{eq:rho_k_Gamma0_bounds} with $k=3$, \eqref{eq:ul_integral}, and $\kappa_0^{-1}\lesssim 1 + C\exp(-c\beta N) \lesssim 1$, which follows from Lemma~\ref{lem:cutoff_a_remainder}. 
\end{proof}

\begin{lemma}\label{lem:K2_upper_bound}
		We have
	\begin{equation}\label{eq:K2_upper_bound}
		K_2 = \tr[\mathcal K \widetilde \Gamma_0] + \mathcal E^{(2)}_{\mathcal K},
	\end{equation} where the error term satisfies the bound
	\begin{equation}\label{eq:K2_error_bound}
		\mathcal E^{(2)}_{\mathcal K} \lesssim N^{2/3 + 3\delta_{\mathrm{Bog}}} \sqrt { \ell(1 + N^2\ell^3 ) }.
	\end{equation}
\end{lemma}
\begin{proof}
	Let us decompose the kinetic energy in momentum space as $\mathcal K = \mathcal K^< + \mathcal K^> $ with 
	\begin{equation*}
		\mathcal K^< = \sum_{p\in P_{\mathrm B}} p^2 a_p^* a_p, \qquad \mathcal K^> = \sum_{p\in \Lambda_+^* \setminus P_{\mathrm B}} p^2 a_p^* a_p.
	\end{equation*} As a symmetrized product of plane waves, $\Psi_\alpha$ is an eigenfunction of the localized kinetic energies $\mathcal K^<$ and $\mathcal K^>$, that is $\mathcal K^{ \gtrless } \Psi_{\alpha} = E_\alpha^{\gtrless} \Psi_{\alpha}$. 
	Both, $T_z$ and $W_z$, act trivially on creation/annihilation operators indexed by high momenta $p\in \Lambda_+^*\setminus P_{\mathrm B}$. Hence,
	\begin{equation*}
\begin{split}
			\int_{\mathbb C } \zeta(z)  \sum_{\alpha\in \widetilde{ \mathcal A}}\frac {  \widetilde \lambda_\alpha } { \| F \phi_{z, \alpha} \|^2 } \langle F^2 \phi_{z, \alpha} , \mathcal K^> \phi_{z, \alpha} \rangle \,\mathrm d z = \sum_{\alpha\in \widetilde{ \mathcal A}}  \widetilde \lambda_\alpha E_\alpha^{>} = \tr[\mathcal K^> \widetilde \Gamma_0]. 
\end{split}
	\end{equation*} To extract the largest contributions from the low-momentum part of the kinetic energy, we write 
	\begin{equation}\label{eq:kinetic_low_momenta_expansion}
		\begin{split}
			\langle F^2  \phi_{z, \alpha},  \mathcal K^< \phi_{z, \alpha}\rangle = & \langle \phi_{z, \alpha}, \mathcal K^< \phi_{z, \alpha} \rangle + \langle (F^2 - 1)\phi_{z, \alpha}, \mathcal K^< \phi_{z, \alpha} \rangle.
		\end{split}
	\end{equation} Combining the first term on the right-hand side of \eqref{eq:kinetic_low_momenta_expansion} with the high momentum part of the kinetic energy gives $\tr[\mathcal K \widetilde \Gamma_0]$. It remains to estimate the remainder.
	
	Using Cauchy--Schwarz, \eqref{eq:cutoff_a_remainder_bound}, \eqref{eq:denominator_lower_bound} and the inequality
	\begin{equation*}
		\left| F_n (x_1,...,x_n)^2 -1 \right| \le \sum_{i < j}^n u_\ell(x_i-x_j),
	\end{equation*} which follows from \eqref{eq:product_lower_bound} and $F_n \le 1$, we can estimate the error term as 
	\begin{equation}\label{eq:kinetic_error_2_intermediate_estimate}
		\begin{split}
			\mathcal E_{\mathcal K}^{(2)} = & \mathfrak{Re} \int_{ \mathbb C } \zeta(z)  \sum_{\alpha\in \widetilde{ \mathcal A}}\frac {  \widetilde \lambda_\alpha } { \| F \phi_{z, \alpha} \|^2 } \langle (F^2 - 1)\phi_{z, \alpha}, \mathcal K^< \phi_{z, \alpha} \rangle  \\
			\lesssim & \Big(  \int_{ \mathbb C } \zeta(z)  \sum_{\alpha\in \widetilde{ \mathcal A}} \widetilde \lambda_\alpha \| (F^2 - 1)\phi_{z, \alpha}\|^2 \,\mathrm dz \Big )^{1/2} \Big(  \int_{ \mathbb C } \zeta(z)  \sum_{\alpha\in \widetilde{ \mathcal A}} \widetilde \lambda_\alpha \| \mathcal K^<\phi_{z, \alpha}\|^2 \,\mathrm dz \Big)^{1/2}  \\
			\lesssim &  \Big( \int_{\mathbb C} \zeta(z)  \sum_{\alpha\in \widetilde{ \mathcal A}} \widetilde \lambda_\alpha \sum_{ n } \int_{\Lambda^n} \Big| \sum_{1\le i < j \le n} u_{ij}\Big| ^2 | \phi_{z, \alpha}^{n}|^2 \, \mathrm d X^n \,\mathrm d z  \Big)^{1/2}  \big( \tr_{\mathfrak F_+}[\mathcal (\mathcal K^<)^2 \Gamma_0] \big) ^{1/2} 
		\end{split} 
	\end{equation} with $u_{ij}$ in \eqref{eq:notation_fij}. We have
	\begin{equation*}
		\tr_{\mathfrak F_+}[\mathcal (\mathcal K^<)^2 \Gamma_0] \le N^{4\delta_{\mathrm {Bog}}} \int_{ \mathbb C } \zeta(z) \tr_{\mathfrak F_+}[\mathcal (\mathcal N^<)^2 G(z)] \de z.
	\end{equation*} An application of Wick's theorem and Lemma \ref{lem:1pdm_pairing_bogoliubov} shows
	\begin{equation*}
		\begin{split}
			\tr_{\mathfrak F_+}[\mathcal (\mathcal N^<)^2 G(z)] = & \sum_{ p,q\in P_{\mathrm B} }\tr_{\mathfrak F_+}[a_p^*a_pa_q^*a_q G(z)] \\
			= & \sum_{ p,q\in P_{\mathrm B} } \gamma_p\gamma_q  + \sum_{ p \in P_{\mathrm B} } \Big( \gamma_p(\gamma_p+1) + |\alpha_p|^2 \Big) \lesssim (1+\beta^{-1}N^{\delta_{\mathrm {Bog}}})^2,
		\end{split}
	\end{equation*} and hence
	\begin{equation}
		\label{eq:kinetic_low_bound}
		\big( \tr_{\mathfrak F_+}[\mathcal (\mathcal K^<)^2 \Gamma_0] \big) ^{1/2} \lesssim N^{2\delta_{\mathrm{Bog}}} (1 + \beta^{-1}N^{\delta_{\mathrm{Bog}}}).
	\end{equation} 
	
	We use $\kappa_0^{-1}\lesssim 1$, which follows from Lemma \ref{lem:cutoff_a_remainder}, and expand the square in the integral to see that 
	\begin{equation}\label{eq:E_K2_bound_term_1}
\begin{split}
		 \int_{\mathbb C} & \zeta(z)  \sum_{\alpha\in \widetilde{ \mathcal A}} \widetilde \lambda_\alpha \sum_{ n } \int_{\Lambda^n} \Big| \sum_{1\le i < j \le n} u_{ij}\Big| ^2 \big | \phi_{z, \alpha}^{n} \big|^2 \, \mathrm dX^n \,\mathrm dz\\
		 \lesssim & \frac 1 4 \int_{\Lambda^4} \rho_{\Gamma_0}^{(4)}(x_1,x_2,x_3,x_4) u_\ell(x_1-x_2) u_\ell(x_3 - x_4)\mathrm dx_1 \mathrm d x_2 \mathrm d x_3 \mathrm d x_4 \\
		& \qquad + \int_{\Lambda^3} \rho_{\Gamma_0}^{(3)}(x_1,x_2,x_3) u_\ell(x_1-x_2) u_\ell(x_1-x_3)\mathrm dx_1 \mathrm d x_2 \mathrm d x_3 \\
		& \qquad  \qquad + \frac 1 2 \int_{\Lambda^2} \rho_{\Gamma_0}^{(2)}(x_1,x_2) u_\ell(x_1-x_2)^2 \mathrm dx_1 \mathrm d x_2 .
\end{split}
	\end{equation} 
	Moreover, applications of \eqref{eq:ul_integral} and Lemma \ref{lem:rho_k_Gamma_bound} show
	\begin{equation}\label{eq:bound_rho_2_to_4}
		\begin{split}
			\int_{\Lambda^4} \rho_{\Gamma_0}^{(4)}(x_1,x_2,x_3,x_4) u_\ell(x_1-x_2) u_\ell(x_3 - x_4)\mathrm dx_1 \mathrm d x_2 \mathrm d x_3 \mathrm d x_4 \lesssim &  N^4 \mathfrak a_N^2\ell^4 \lesssim N^2\ell^4, \\
			\int_{\Lambda^3} \rho_{\Gamma_0}^{(3)}(x_1,x_2,x_3) u_\ell(x_1-x_2) u_\ell(x_1-x_3)\mathrm dx_1 \mathrm d x_2 \mathrm d x_3 \lesssim & N^3 \mathfrak a_N^2 \ell^4 \lesssim N\ell^4, \\
			\int_{\Lambda^2} \rho_{\Gamma_0}^{(2)}(x_1,x_2) u_\ell(x_1-x_2)^2 \mathrm dx_1 \mathrm d x_2 \lesssim & N^2 \mathfrak  a_N^2\ell \lesssim \ell.
		\end{split}
	\end{equation}
	With \eqref{eq:E_K2_bound_term_1} and \eqref{eq:bound_rho_2_to_4}, we see that the first factor on the right-hand side of \eqref{eq:kinetic_error_2_intermediate_estimate} is bounded by $\sqrt{ \ell(1 + N^2\ell^3 ) }$. Combined with \eqref{eq:kinetic_low_bound}, this implies \eqref{eq:K2_error_bound}. 
\end{proof}

	Putting together Lemma \ref{lem:K1_upper_bound} and Lemma \ref{lem:K2_upper_bound}, we find 
	\begin{equation*}
			\tr[\mathcal K\Gamma] = \tr[\cchi\Gamma] + \tr[\mathcal K\widetilde \Gamma_0] + \mathcal E_{\mathcal K},
	\end{equation*} with 
\begin{equation*}
	  \mathcal E_{\mathcal K} = \mathcal E_{\mathcal K}^{(1)} + \mathcal E_{\mathcal K}^{(2)}. 
\end{equation*} The claim of Proposition \ref{prop:kinetic_energy_upper_bound} follows from \eqref{eq:K1_error_bound} and  \eqref{eq:K2_error_bound}. \qed

\subsection{Analysis of the renormalized interaction}
\label{sec:analysisoftherenormalizedinteraction}

As shown in Proposition \ref{prop:kinetic_energy_upper_bound}, the expectation of the kinetic energy in our trial states yields two contributions, up to negligible errors. The first contribution is the kinetic energy of the undressed trial state $\widetilde{\Gamma}_0$, which will be combined with the entropy to obtain the free energy of the ideal gas. The second is given by the expectation of the two-body operator $\cchi$ in the state $\Gamma$. This term needs to be combined with the interaction potential $\mathcal V_N$ to replace the integral of $V_N$ by $8 \pi \mathfrak a_N L^{-3}$ in the relevant contributions to the free energy. The precise statement is captured by the following proposition.  
\begin{proposition}\label{prop:renormalized_potential} We have
	\begin{equation}\label{eq:renormalized_potential_statement}
\begin{split}
		 \tr[(\cchi + \mathcal V_N) \Gamma] = & \int_{\mathbb C } \zeta(z) \tr [ \mathcal Q_{\mathrm B} \widetilde G(z) ] \, \mathrm d z + 4\pi \mathfrak a _N \int_{\mathbb C } \zeta(z) |z|^4 \, \mathrm d z \\
	 & +   8\pi \mathfrak a_N  \Bigg[  \Big( \sum_{p \in \Lambda_+^* } \gamma_p \Big)^2 + \widetilde{N}_0  \sum_{ p \in \Lambda_+^* \setminus  P_{\mathrm B} } \gamma_p + \widetilde{N}_0  \sum_{p \in \Lambda_+^* } \gamma_p \Bigg] + \mathcal E_{\mathcal V}, 
\end{split}
	\end{equation} where $\widetilde N_0$ is defined in \eqref{eq:n0_tilde_def} and 
\begin{equation}\label{eq:QB_def}
	\mathcal Q_{\mathrm B } \coloneqq  4\pi \mathfrak a_N N_0 \sum_{ p \in P_{\mathrm B}  } \left[ 2a_p^*a_p + \left( \frac{z^2}{|z|^2} a_p^* a_{-p}^* + \frac{\overline{z}^2}{|z|^2} a_p a_{-p}\right) \right].
\end{equation} The error satisfies the bound
	\begin{equation}
		\mathcal E_{\mathcal V} \lesssim N^{1/3 + \delta_{ \mathrm{ Bog } }} + N^{2/3 + 2\delta_{ \mathrm{ Bog } }} \ell + N^{5/3}\ell^2 + N^3\ell^4 + \ell^{-1}.
	\end{equation}
\end{proposition}
To prove the above proposition, we start by writing the expectation on the left-hand side of \eqref{eq:renormalized_potential_statement} in a more convenient way:
\begin{equation*}
\begin{split}
		\tr[(\cchi + \mathcal V_N) \Gamma]  = &  \int_{\mathbb C } \zeta(z)  \sum_{\alpha\in \widetilde{ \mathcal A}}\frac {  \widetilde \lambda_\alpha } { \| F\phi_{z, \alpha} \|^2 } \\
		& \times \sum_{n = 2}^\infty \sum_{1\le i < j \le n} \int_{\Lambda^n} \Big( 2 \frac{ | \nabla f_\ell (x_i - x_j) | ^2 }{ f_ \ell (x_i - x_j ) ^2 }  + v_n(x_i-x_j) \Big) F_n^2 \left| \phi_{z, \alpha}^{n}\right |^2 \, \mathrm dX^n  \, \mathrm d z\\
		\leq & \int_{ \mathbb C } \zeta(z)  \sum_{\alpha\in \widetilde{ \mathcal A}}\frac {  \widetilde \lambda_\alpha } { \| F\phi_{z, \alpha} \|^2 } \\
		& \times \sum_{n = 2}^\infty \sum_{1\le i < j \le n} \int_{\Lambda^n} 2 \xi_N(x_i-x_j) \prod_{ \substack{ 1\le k < h \le n \\ (k,h)\neq(i,j) } } f_\ell^2(x_k-x_h) \left| \phi_{z, \alpha}^{n}\right |^2 \, \mathrm d X^n \, \mathrm d z,
\end{split}
\end{equation*} where we defined the effective potential 
\begin{equation*}
	\xi_N(x) = |\nabla f_\ell(x)|^2 + \frac 1 2 v_N(x) f_\ell^2(x).
\end{equation*}  Its $L^1$-norm satisfies
\begin{equation}
\label{eq:xi_N_integral}
\int_{\Lambda}\xi_N(x)\mathrm dx = \mathcal E_\ell (f_\ell) = \frac{4\pi \mathfrak a_N}{1 - \frac{ \mathfrak a_N}{\ell}},
\end{equation} see \eqref{eq:energy_f_ell}. For $\ell \ge 2 \mathfrak a_N$, the denominator is bounded from below by $1/2$, and thus $ \| \xi_N \|_1 \lesssim \mathfrak a_N $. 

Using the pointwise bound 
\begin{equation}\label{eq:upperBoundJastrow}
	\prod_{ \substack{ k<h \\ (k, h) \neq (i, j) } } f_\ell^2(x_k-x_h) \le 1 - \sum_{ \substack{ k<h \\ (k, h) \neq (i, j) } } u_\ell(x_k-x_h) + \frac 1 2 \sum_{\substack { k<h, r<s \\ (i,j)\neq (k,h) \neq (r,s) \\ (i,j) \neq (r,s)}  }u_\ell(x_k-x_h) u_\ell(x_r-x_s),
\end{equation} we find 
\begin{equation} \label{eq:Chi+V_expectation_computation}
	\begin{split}
		\tr[(\cchi & + \mathcal V_N) \Gamma] \le \int_{ \mathbb C } \zeta(z)  \sum_{\alpha\in \widetilde{ \mathcal A}}\frac {  \widetilde \lambda_\alpha } { \| F\phi_{z, \alpha} \|^2 } \int_{\Lambda^2} \rho_{z, \alpha}^{(2)} \xi_N(x_1-x_2) \, \mathrm d x_1 \, \mathrm d x_2 \, \mathrm d z \\
		& - \frac 1 2  \int_{ \mathbb C } \zeta(z)  \sum_{\alpha\in \widetilde{ \mathcal A}}\frac {  \widetilde \lambda_\alpha } { \| F\phi_{z, \alpha} \|^2 }\int_{\Lambda^4} \rho^{(4)}_{z,\alpha} \xi_N(x_1-x_2) u_\ell(x_3-x_4) \,\mathrm dx_1 \,\mathrm dx_2 \,\mathrm dx_3 \,\mathrm dx_4  \, \mathrm d z\\
		& + \frac 1 8 \int_{ \mathbb C } \zeta(z)  \sum_{\alpha\in \widetilde{ \mathcal A}}\frac {  \widetilde \lambda_\alpha } { \| F\phi_{z, \alpha} \|^2 } \int_{\Lambda^6} \rho^{(6)}_{z,\alpha} \xi_N(x_1-x_2) u_\ell(x_3-x_4) u_\ell(x_5-x_6) \, \mathrm dx_1 ... \, \mathrm dx_6  \, \mathrm d z\\
		= & \tr[ \Xi_N \widetilde \Gamma_0 ] + V_1 + V_2 + V_3.
	\end{split}
\end{equation} Here, we defined
\begin{equation*}
\Xi_N \coloneqq \int_{\mathbb R^6} \xi_N(x-y)a_x^*a_y^*a_xa_y \, \mathrm d x \, \mathrm d y,
\end{equation*} $V_1$ is the first term on the right-hand side of \eqref{eq:Chi+V_expectation_computation}, with $ \| F\phi_{z, \alpha} \|^{-2} $ replaced by $  \| F\phi_{z, \alpha} \|^{-2} -1 $, 
and the densities $\rho_{z, \alpha}^{(k)}$ are defined in \eqref{eq:rho_za_k_def}. Using \eqref{eq:rho_k_Gamma0_bounds}, \eqref{eq:xi_N_integral} and \eqref{eq:ul_integral}, it is easy to see that 
\begin{equation}\label{eq:V3_bound}
	\begin{split}
		V_3 \lesssim \| \rho^{(6)}_{\Gamma_0} \| _ \infty \| \xi_N \|_1 \|u_\ell\|_1^2 \lesssim N^6 a_N^3 \ell^4 \lesssim N^3\ell^4. 
	\end{split}
\end{equation}
In contrast, naive bounds for $V_1$ and $V_2$ are not sufficient to achieve the level of precision necessary to prove Theorem \ref{thm:main_thm}. Indeed, when considered individually, they give contributions to the energy proportional to $N^2\ell^2$. Together with the errors arising from the localization of the minimization problem in \eqref{eq:energy_f_ell}, which are of order $\ell^{-1}$, they add up to a contribution proportional to $N^{2/3}$. This is the level of accuracy of the upper bound for the free energy given in \cite{DeuchertSeiringerHom2020}, and it is not sufficient to resolve the free energy of the interacting BEC and the Bogoliubov corrections we are interested in.  

It turns out that $V_1$ and $V_2$ cancel out to leading order. This important cancellation is the content of the next lemma.
\begin{lemma}\label{lem:V1+V2_cancellation}
The following bound holds:
	\begin{equation}
		\label{eq:V1+V2_cancellation}
		V_1+V_2 \lesssim N^{5/3}\ell^2 + N^3\ell^4.
	\end{equation}
\end{lemma}
\begin{proof}
	We compute the leading order contributions of $V_1$ and $V_2$ separately, and we check that they indeed cancel. We first consider $V_1$. Using \eqref{eq:denominator_bound_first} and the bound $(1-x)^{-1}\le 1+ x + 2x^2$ for $0\le x \le 1/2$, we find 
	\begin{equation*}
		\begin{split}
			V_1 \le & \frac 1 2  \int_{ \mathbb C } \zeta(z)  \sum_{\alpha\in \widetilde{ \mathcal A}}  \widetilde \lambda_\alpha \Big ( \int_{\Lambda^2} \rho_{\alpha, z}^{(2)}(x_1,x_2) u_\ell(x_1-x_2) \, \mathrm dx_1 \, \mathrm dx_2 \Big ) \\
			& \qquad \times \Big ( \int_{\Lambda^2} \rho_{\alpha, z}^{(2)}(x_1,x_2) \xi_N(x_1-x_2) \, \mathrm dx_1 \, \mathrm dx_2 \Big ) \, \mathrm dz\\
			& +  \int_{ \mathbb C } \zeta(z)  \sum_{\alpha\in \widetilde{ \mathcal A}}  \widetilde \lambda_\alpha \Big ( \int_{\Lambda^2} \rho_{\alpha, z}^{(2)}(x_1,x_2) u_\ell(x_1-x_2) \, \mathrm dx_1 \, \mathrm dx_2 \Big )^2 \\
			& \qquad \times \Big ( \int_{\Lambda^2} \rho_{\alpha, z}^{(2)}(x_1,x_2) \xi_N(x_1-x_2) \, \mathrm dx_1 \, \mathrm dx_2 \Big ) \, \mathrm dz \eqqcolon V_{1,1} + V_{1,2}.
		\end{split}
	\end{equation*} Equations \eqref{eq:rho2za_bound}, \eqref{eq:xi_N_integral}, \eqref{eq:ul_integral}, and the particle number cutoffs imposed on $\widetilde \lambda_\alpha$ and $\zeta(z)$, allow us to show that
\begin{equation}\label{eq:V12_bound}
	\begin{split}
	  V_{1,2} \lesssim N ^6 \, \Big ( \int u_\ell \Big )^2\Big (\int \xi_N\Big ) \lesssim N^6 \mathfrak a_N^{3} \ell^4 \lesssim N^3\ell^4.
	\end{split}
\end{equation} 

We recall the definition of $N_{\alpha}$ in \eqref{eq:Eigenvalues_Psi_alpha}. With the same bounds we also obtain
\begin{equation}\label{eq:V11_def}
	V_{1,1} \le \frac 1 2 \Big( \int \xi_N\Big) \Big(\int u_\ell \Big)  \int_{ \mathbb C } \zeta(z)  \sum_{\alpha\in \widetilde{ \mathcal A}}  \widetilde \lambda_\alpha  \big ( |z|^4 + 4|z|^2 N_\alpha + 2 N_\alpha^2 \big )^2 \, \mathrm d z + \mathcal E_{\mathcal V}^{(1)},
\end{equation} where the error term satisfies (recall $\delta_{\mathrm{Bog}} < 1/12$)
\begin{equation*}\begin{split}
	\mathcal E_{\mathcal V}^{(1)} \lesssim & N^{8/3 + \delta_{ \mathrm{ Bog } }} \Big( \int \xi_N\Big) \Big(\int u_\ell \Big) \lesssim N^{2/3 + \delta_{ \mathrm{ Bog } }} \ell^2.
	\end{split}
\end{equation*} It thus follows from Lemma \ref{cor:expectation_powers_N} that
\begin{equation}\label{eq:V11_final_upper_bound}
	\begin{split}
		V_{1,1} \le & \frac 1 { 2 \kappa_0 } \Big ( \int \xi_N \Big ) \Big ( \int u_\ell \Big ) \int_{ \mathbb C } \zeta(z) \Big\{ |z|^8 + 8 |z|^6 \tr_{\mathfrak{F}_+}[\mathcal N_+ G^{ \mathrm{ diag } }] \\
		& + 20 |z|^4 \tr_{\mathfrak{F}_+}[\mathcal N_+ G^{ \mathrm{ diag } }] ^2 +  16 |z|^2 \tr_{\mathfrak{F}_+}[\mathcal N_+ G^{ \mathrm{ diag } }]^3 + 4 \tr_{\mathfrak{F}_+}[\mathcal N_+ G^{ \mathrm{ diag } }] ^4  \Big\} +C N^{4/3} \ell^2.
	\end{split} 
\end{equation}

We now turn to estimating $V_2$. The simple bound $\| F(z\otimes T_z\Psi_{\alpha})\|\le 1$ implies
\begin{equation*}\begin{split}
		V_2 \le &  - \frac 1 2  \int_{ \mathbb C } \zeta(z)  \sum_{\alpha\in \widetilde{ \mathcal A}}  \widetilde \lambda_\alpha \int_{\Lambda^4} \rho^{(4)}_{z,\alpha}(x_1,x_2,x_3,x_4) \xi_N(x_1-x_2) u_\ell(x_3-x_4) \mathrm dx_1 \mathrm dx_2\mathrm dx_3\mathrm dx_4 \, \mathrm d z \\
		= & -  \frac 1 { 2 \kappa_0 } \int_{\Lambda^4} \rho^{(4)}_{\Gamma_0}(x_1,x_2,x_3,x_4) \xi_N(x_1-x_2) u_\ell(x_3-x_4)\, \mathrm dx_1 \,  \mathrm dx_2 \, \mathrm dx_3 \, \mathrm dx_4 + \mathcal E_{ \mathcal V }^{(2)} , 
	\end{split}
\end{equation*} with 
\begin{equation*}
	\mathcal  E_{ \mathcal V }^{(2)} =  \frac 1 2  \int_{ \mathbb C}  \zeta(z)  \sum_{\alpha\in \mathcal A \setminus \widetilde{ \mathcal A}}  \widetilde \lambda_\alpha \int_{\Lambda^4} \rho^{(4)}_{z,\alpha}(x_1,x_2,x_3,x_4) \xi_N(x_1-x_2) u_\ell(x_3-x_4)\, \mathrm dx_1 \,  \mathrm dx_2 \, \mathrm dx_3 \, \mathrm dx_4 \, \mathrm d z.
\end{equation*} Using $0\le u_\ell \le 1$, \eqref{eq:cutoff_a_remainder_bound}, \eqref{eq:rho4_za_integral_bound} and \eqref{eq:xi_N_integral}, we find 
\begin{equation}\label{eq:V2_error_bound_1}
	\mathcal E_{ \mathcal V }^{(2)} \lesssim \mathfrak a_N  \sum_{\alpha\in \mathcal A \setminus \widetilde{ \mathcal A}}  \widetilde \lambda_\alpha (N^4 + N^{3\delta_{ \mathrm{ Bog } }} N_\alpha^2 N^2 + N^{3\delta_{ \mathrm{ Bog } }} N_\alpha^4) \lesssim  \exp(-c N^{ \delta_{ \mathrm{ Bog } }} ).
\end{equation}

To extract the leading order contribution from $V_2$, we expand
\begin{equation}\label{eq:V2_main_term}
	\begin{split}
		-\frac 1 2 \int_{\Lambda^4} & \rho^{(4)}_{\Gamma_0} (x_1,x_2,x_3,x_4)\xi_N(x_1-x_2) u_\ell(x_3-x_4) \, \mathrm dx_1 \,  \mathrm dx_2 \, \mathrm dx_3 \, \mathrm dx_4  \\
		= & -\frac 1 2 \Big (\int u_\ell \Big ) \Big (\int \xi_N \Big)\int_{\Lambda^4} \rho^{(4)} _{\Gamma_0} (x_1,x_1,x_2,x_2) \, \mathrm dx_1 \, \mathrm dx_2 + \mathcal E_{\mathcal V}^{(3)},
	\end{split}
\end{equation} where
\begin{equation}\label{eq:V2_error_bound_2}
	\begin{split}
		 \mathcal E_{\mathcal V}^{(3)} = &	-\frac 1 2 \int_{\Lambda^4} \int_0^1 [ (x_2-x_1)\cdot \nabla_2 + (x_4-x_3) \cdot \nabla_4 ] \rho^{(4)}_{\Gamma_0} (x_1,x_1 + t(x_2-x_1),x_3,x_3 + t(x_4-x_3)) \\
		& \qquad\qquad \times \xi_N(x_1-x_2) u_\ell(x_3-x_4) \, \mathrm dt \, \mathrm dx_1 \,  \mathrm dx_2 \, \mathrm dx_3 \, \mathrm dx_4.
	\end{split}
\end{equation} With \eqref{eq:xi_N_integral},  \eqref{eq:ul_integral}, Lemma \ref{lem:rho_k_Gamma_bound}, and the fact that the supports of $\xi_N$ and $u_\ell$ are contained in $B_\ell$, we find 
\begin{equation*}
	|  \mathcal E_{\mathcal V}^{(3)} | \lesssim  \ell \Big (\int u_\ell \Big ) \Big (\int \xi_N \Big ) ( \| \nabla_2 \rho^{(4)}_{\Gamma_0} \|_\infty +  \| \nabla_4 \rho^{(4)}_{\Gamma_0} \|_\infty) \lesssim \mathfrak a_N^2\ell^3 N^4 \beta^{-1/2} \lesssim  N^{7/3}\ell^3.
\end{equation*}

Let us now compare \eqref{eq:V2_main_term} with \eqref{eq:V11_final_upper_bound}. We write 
\begin{equation}\label{eq:rho4_gamma0_localized}
	\begin{split}
		\rho^{(4)}_{\Gamma_0 } (x_1,x_1,x_2,x_2) = & \int_{\mathbb C } \zeta(z) \tr_{\mathfrak F} [ a_{x_1}^*a_{x_1}^*a_{x_2}^*a_{x_2}^*a_{x_1}a_{x_1}a_{x_2}a_{x_2} | z \rangle \langle z | \otimes G(z) ] \de z,
	\end{split}
\end{equation} and use \eqref{eq:action_weyl_position} and Wick's theorem to compute the right-hand side. A naive approach generates many terms. To simplify the computation, we observe that, by Lemma \ref{lem:1pdm_pairing_bogoliubov}, $  \sup_{x\in\Lambda}|\check\alpha(x)| \lesssim N^{2/3} $. In particular, when we apply \eqref{eq:action_weyl_position} to \eqref{eq:rho4_gamma0_localized}, we see that all terms that do not contain the same number of creation and annihilation operators are subleading, and we get
\begin{equation}\label{eq:V2_localized_intermediate}
	\begin{split}
		-\frac 1 2 & \Big(\int u_\ell \Big) \Big(\int \xi_N \Big)\int_{\Lambda^4}  \rho^{(4)} _{\Gamma_0} (x_1,x_1,x_2,x_2) \, \mathrm dx_1 \, \mathrm dx_2  \\
		= &  -\frac 1 2 \Big(\int \xi_N\Big) \Big(\int u_\ell\Big) \int  \int_{\mathbb C } \zeta(z) \Big\{ \tr_{\mathfrak F_+} [ a_{x_1}^*a_{x_1}^*a_{x_2}^*a_{x_2}^*a_{x_1}a_{x_1}a_{x_2}a_{x_2} G(z) ] \\
		& + 8 |z|^2 \tr_{\mathfrak F_+} [ a_{x_1}^*a_{x_2}^*a_{x_2}^* a_{x_1}a_{x_2}a_{x_2} G(z) ] + 8 |z|^2 \tr_{\mathfrak F_+} [ a_{x_1}^*a_{x_2}^*a_{x_2}^* a_{x_1}a_{x_1}a_{x_2} G(z) ] \\
		& + 2|z|^4 \tr_{\mathfrak F_+} [ a_{x_1}^* a_{x_1}^* a_{x_1} a_{x_1} G(z) ] + 16 |z|^4 \tr_{\mathfrak F_+} [ a_{x_1}^* a_{x_2}^* a_{x_1} a_{x_2} G(z) ]  \\
		& + 8|z|^4 \tr_{\mathfrak F_+} [ a_{x_1}^* a_{x_1}^* a_{x_1} a_{x_2} G(z) ] + 8|z|^4 \tr_{\mathfrak F_+} [ a_{x_1}^* a_{x_2}^* a_{x_1} a_{x_1} G(z) ]  \\
		& + 2|z|^4 \tr_{\mathfrak F_+} [ a_{x_1}^* a_{x_1}^* a_{x_2} a_{x_2} G(z) ] \\
		& + 8 |z|^6 \tr_{\mathfrak F_+} [ a_{x_1}^* a_{x_1} G(z) ] + 8 |z|^6 \tr_{\mathfrak F_+} [ a_{x_1}^* a_{x_2} G(z) ] + |z|^8\Big\}\, \mathrm d z \, \mathrm dx_1 \, \mathrm dx_2 + CN^{5/3} \ell^2.
	\end{split}
\end{equation}
An application of Lemma \ref{lem:1pdm_pairing_bogoliubov} shows
\begin{equation*}
 \int_{\Lambda} \check\gamma(x)  \, \mathrm d x =  0, \qquad  \int_{\Lambda} | \check\gamma(x) |^2  \, \mathrm d x  \lesssim N^{4/3}.
\end{equation*} When we apply Wick's theorem to compute the right-hand side of \eqref{eq:V2_localized_intermediate}, this allows us to see that only the constant terms give leading-order contributions. More precisely, the right-hand side of \eqref{eq:V2_localized_intermediate} is bounded from above by
\begin{equation*}
	\begin{split}
		 -\frac 1 2 \Big(\int \xi_N\Big) \Big(\int u_\ell\Big) \int_{\mathbb C } \zeta(z) & \Big\{  4 \check \gamma(0)^4 + 16 |z|^2 \check \gamma(0)^3 \\
		& + 20 |z|^4\check \gamma(0)^2 + 8 |z|^6 \check \gamma(0) + |z|^8 \Big\} \, \mathrm d z + CN^{5/3}\ell^2.
	\end{split}
\end{equation*} It follows from \eqref{eq:vp_pointwise_bound} and \eqref{eq:1pdm_pairing_bogoliubov_momentum} that $ \big | \check\gamma(0)- \tr_{\mathfrak{F}_+}[ \mathcal N_+ G^{ \mathrm{ diag} } ] \big | \lesssim N^{2/3}$. The above considerations imply 
\begin{equation} \label{eq:V2_final_uper_bound}
	\begin{split}
		V_2 \le &  -\frac 1 { 2 \kappa_0 } \Big (\int \xi_N\Big ) \Big (\int u_\ell\Big )  \int_{\mathbb C } \zeta(z)  \Big\{  4 \tr_{\mathfrak{F}_+}[ \mathcal N_+  G^{ \mathrm{ diag } }]^4 + 16 |z|^2 \tr_{\mathfrak{F}_+}[ \mathcal N_+  G^{ \mathrm{ diag } }]^3 \\
		& + 20 |z|^4 \tr_{\mathfrak{F}_+}[ \mathcal N_+ G^{ \mathrm{ diag } }]^2 + 8 |z|^6 \tr_{\mathfrak{F}_+}[ \mathcal N_+  G^{ \mathrm{ diag } }] + |z|^8 \Big\} \, \mathrm d z + C ( N^{7/3} \ell^3 + N^{5/3} \ell^2 ).
	\end{split}
\end{equation}  
Combining \eqref{eq:V12_bound}, \eqref{eq:V11_final_upper_bound} and \eqref{eq:V2_final_uper_bound} we find \eqref{eq:V1+V2_cancellation}.
\end{proof}

In the next Lemma, we extract the largest contributions to the energy from the expectation of the effective interaction $\Xi_N$ with respect to the undressed trial state.

\begin{lemma}\label{lem:Xi_expectation_undressed} We have 
 		\begin{equation}\label{eq:Xi_expectation_undressed}
\begin{split}
	 		\tr[\Xi_N \widetilde \Gamma_0] = & \int_{\mathbb C } \zeta(z) \tr [ \mathcal Q_{\mathrm B} \widetilde G(z) ] \, \mathrm d z + 4\pi \mathfrak a _N \int_{\mathbb C } \zeta(z) |z|^4 \, \mathrm d z \\
	 		& +   8\pi \mathfrak a_N  \Bigg[  \Big(\sum_{ p \in \Lambda_+^* } \gamma_p \Big)^2 + \widetilde{N}_0  \sum_{ p \in \Lambda_+^* \setminus  P_{\mathrm B} } \gamma_p + \widetilde{N}_0  \sum_{ p \in \Lambda_+^* } \gamma_p  \Bigg] + \mathcal E_{\Xi},
\end{split} 
 	\end{equation} with $\mathcal Q_{ \mathrm B }$ in \eqref{eq:QB_def} and 
 \begin{equation}\label{eq:E_Xi_bound}
 	\mathcal E_{\Xi} \lesssim N^{2/3 + \delta_{ \mathrm{ Bog } }} ( N^{-1/3} + N^{ \delta_{ \mathrm{Bog} } }\ell + N^{2/3 + \delta_{ \mathrm{ Bog } }} \ell^2 + N^2 \ell^4  )+ \ell^{-1}.
 \end{equation}
\end{lemma}
\begin{proof}
In the momentum space representation, $\Xi_N$ reads
\begin{equation*}
	\Xi_N = \sum_{p,q,r\in\Lambda^*} \widehat \xi_N(r) a_{p+r}^* a_q^* a_p a_{q+r}, 
\end{equation*} and hence
\begin{equation}\label{eq:Xi_N_expansion}
	\langle W_z^* \Xi_N W_z \rangle_{\Omega_0\otimes T_z\Psi_\alpha} = \langle \Xi_N^+ + \mathcal C_N + \mathcal Q_< + \mathcal Q_> \rangle_{\Omega_0\otimes T_z\Psi_\alpha} + 2\widehat \xi_N(0) |z|^2 \langle \mathcal N_+ \rangle_{T_z\Psi_\alpha} + \widehat \xi_N(0) |z|^4,
\end{equation} 
with 
\begin{equation*}
	\begin{split}
		\Xi_N^+ = & \sum_{\substack{p,q,r\in\Lambda^* \\ p, q , p+r, q+r \neq 0}} \widehat \xi_N(r) a_{p+r}^* a_q^* a_p a_{q+r}, \\
	\mathcal C_N = & 2 \sum_{ \substack{p,q,\in\Lambda^*_+ \\  p+q \neq 0} } \widehat \xi_N(p) ( z a_{p+q}^* a_{-p}^* a_q  + \overline z a_{q}^* a_{p + q} a_{-p} ) , \\
	 \mathcal Q_< = & |z|^2 \sum_{ p \in P_{\mathrm B} }\widehat \xi_N(p) \left[ 2a_p^*a_p + \left( \frac{z^2}{|z|^2} a_p^* a_{-p}^* + \frac{\overline{z}^2}{|z|^2} a_p a_{-p}\right) \right], \\
	 \mathcal Q_> = & |z|^2 \sum_{p \in \Lambda_+^* \setminus P_{\mathrm B}  }\widehat \xi_N(p) \left[ 2a_p^*a_p + \left( \frac{z^2}{|z|^2} a_p^* a_{-p}^* + \frac{\overline{z}^2}{|z|^2} a_p a_{-p}\right) \right].
	\end{split}
\end{equation*} Since all $\Psi_\alpha$ are eigenfunctions of $\mathcal N$, the expectation of the cubic term on the basis functions of $ \widetilde G(z) $ vanishes, that is, 
\[ \langle \Omega_0 \otimes T_z\Psi_\alpha, \mathcal C_N \Omega_0 \otimes T_z\Psi_\alpha\rangle = 0.\] 

To bound the expectation of the quartic term, we use $\Xi_N^+\ge 0$, $\widehat \xi_N(p) \le \widehat \xi_N(0)$, which follow from $\xi_N(x)\ge 0$ and $\xi_N(x) = \xi_N(-x)$, Wick's theorem and Lemma \ref{lem:1pdm_pairing_bogoliubov}:
\begin{equation*}
	\begin{split}
		\tr[\Xi_N^+\widetilde \Gamma_0] = &\int_{\mathbb C }  \zeta(z) \sum_{\alpha\in \widetilde{ \mathcal A}} \widetilde\lambda_\alpha \langle T_z\Psi_\alpha, \mathrm \Xi_N^+ T_z\Psi_\alpha \rangle \de z \\
		\le & \kappa_0^{-1}  \int_{\mathbb C } \zeta(z) \sum_{\substack{p,q,r \in \Lambda^*  \\ p, q, p+r, q+r \neq 0 }} \widehat \xi _N (r) \tr_{\mathfrak{F}_+}[ a_{p+r}^*a_q^* a_p a_{q+r} G_{\mathrm B} (z) ] \de z \\
		\le & \kappa_0^{-1} \widehat\xi_N(0) \Big( \sum_{\substack{p,q\in\Lambda^* \\ p, q \neq 0 }} \gamma_p\gamma_q + \sum_{\substack{p,q\in\Lambda^* \\ p, p+q \neq 0 }} \gamma_p\gamma_{p+q} \Big) + \kappa_0^{-1} \sum_{\substack{p,q\in\Lambda^* \\ p, p+q \neq 0 }} \widehat\xi_N(q)  \alpha_p\overline{\alpha_{p+q}}.
	\end{split}
\end{equation*} Using \eqref{eq:norm_bounds_alpha} and \eqref{eq:xi_N_integral}, we see that 
\begin{equation*}
	\begin{split}
		\widehat\xi_N(0) \Big( \sum_{\substack{p,q\in\Lambda^* \\ p, q \neq 0 }} \gamma_p\gamma_q + \sum_{\substack{p,q\in\Lambda^* \\ p, p+q \neq 0 }} \gamma_p\gamma_{p+q} \Big) = & \frac{ 8\pi \mathfrak a_N }{1-\frac{ \mathfrak a_N}{\ell}} \Big(\sum_{ p\in \Lambda_+^* } \gamma_p \Big)^2, \\
		\Big| \sum_{\substack{p,q\in\Lambda^* \\ p, p+q \neq 0 }} \widehat\xi_N(q)  \alpha_p\overline{\alpha_{p+q}} \Big| \lesssim & \mathfrak a_N(N^{ \delta_{ \mathrm {Bog} } } + \beta^{-1})^2 \lesssim N^{1/3}, 
	\end{split}
\end{equation*} and thus
\begin{equation}\label{eq:Xi_N_expectation_quartic}
\tr[\Xi_N^+ \widetilde \Gamma_0]  \le  8\pi \mathfrak a_N \Big(\sum_{ p\in \Lambda_+^* } \gamma_p \Big)^2 + C\big( N^{1/3} + \ell^{-1} \big).
\end{equation} In the last step, we used Lemma \ref{lem:cutoff_a_remainder} to conclude $ \kappa_0^{-1} \le 1 + C\exp(-c N^{\delta_{\mathrm{Bog}}})$, and $(1- \mathfrak a_N/\ell)^{-1} \le 1 + 2 \mathfrak a_N /\ell$, which follows from $\ell \ge 2 \mathfrak a_N$.

We now analyze the quadratic terms, starting with $\mathcal Q_{<}$. We write
\begin{equation}\label{eq:Q<_QB}
\begin{split}
	\int_{\mathbb C } \zeta(z) \tr [ \mathcal Q_< \widetilde G(z) ] \, \mathrm d z = 	\int_{\mathbb C } \zeta(z) \tr [ \mathcal Q_{\mathrm B} \widetilde G(z) ] \, \mathrm d z+ \mathcal E_{ \mathcal Q }
\end{split}
\end{equation} with $Q_{\mathrm{B}}$ in \eqref{eq:QB_def} and
\begin{equation*}\begin{split}
	\mathcal E_{ \mathcal Q } = & \int_{\mathbb C } \zeta(z)  \tr_{\mathfrak F_+} \Bigg[ \sum_{ p \in P_{\mathrm B} } \big( |z|^2 \widehat \xi_N(p) - 4\pi \mathfrak a_N N_0 \big) \Big ( 2a_p^*a_p + \frac{z^2}{|z|^2} a_p^* a_{-p}^* + \frac{\overline{z}^2}{|z|^2} a_p a_{-p} \Big ) \widetilde G(z) \Bigg] \, \mathrm d z .
	\end{split}
\end{equation*} 
It follows from Lemma~\ref{lem:cutoff_a_remainder} and \eqref{eq:norm_bounds_gamma} that
\begin{equation}
	\label{eq:gamma_low_G_z_bound}
	 \sum_{ p\in P_{\mathrm B} } \tr_{\mathfrak{F}_+}[a_p^* a_p \widetilde G_{\mathrm B} (z) ] \le \kappa_0^{-1} \sum_{ p\in P_{\mathrm B} } \gamma_p \lesssim N^{2/3 + \delta_{ \mathrm{ Bog } }}, 
\end{equation} and 
\begin{equation}
	\label{eq:alpha_low_G_z_bound}
\begin{split}
		\Big| \sum_{  p\in P_{\mathrm B} } \tr_{\mathfrak{F}_+}[a_p^* a_{-p}^* \widetilde G_{\mathrm B} (z) ] \Big| \le & \sum_{  p\in P_{\mathrm B} } | \alpha_p | + C \exp(-c N^{\delta_{\mathrm{Bog}}}) \lesssim N^{2/3}. 
\end{split}
\end{equation} 
We also have
\begin{equation}
	\label{eq:xi_n_4pi_a_comparison}
	\sup_{ p\in P_{\mathrm B} } \big|\widehat \xi_N(p) - 4\pi \mathfrak a_N\big| \lesssim \mathfrak a_N N^{ \delta_{ \mathrm{Bog} } } \ell,
\end{equation} which follows from combining 
\begin{equation*} 
	\big| \widehat \xi_N(p) - \xi_N(0) \big| \le |p| \sup_q |(\widehat \xi_N)'(q)| \leq |p| \int_{\Lambda} |x| \xi_N(x) \lesssim |p| \ell \mathfrak a_N
\end{equation*} with \eqref{eq:xi_N_integral}. The last inequality in the above equation follows from \eqref{eq:xi_N_integral} and the fact that the support of $\xi_N$ is contained in the ball of radius $\ell$. Using \eqref{eq:N0_N0tilde_diff}, \eqref{eq:gamma_low_G_z_bound}, \eqref{eq:alpha_low_G_z_bound}, \eqref{eq:xi_n_4pi_a_comparison}, and observing that the quantity
\begin{equation*}
	\tr_{\mathfrak F_+} \Big [ \Big ( 2a_p^*a_p + \frac{z^2}{|z|^2} a_p^* a_{-p}^* + \frac{\overline{z}^2}{|z|^2} a_p a_{-p} \Big )  \widetilde G(z) \Big] 
\end{equation*}   does not depend on $z \in \mathbb C$, we find
\begin{equation} \label{eq:EQ_bound}
	\begin{split}
		\mathcal E_{\mathcal Q} \le & \Big ( \sup_{ p\in P_{ \mathrm B } } \big | \widehat{\xi}_N(p) - 4\pi \mathfrak a_N \big | \widetilde N_0 + 4\pi \mathfrak a_N  \big |\widetilde N_0 - N_0 \big |  \Big ) \\
		& \qquad \times  2 \Big( \sum_{ p\in P_{ \mathrm B } } \tr_{\mathfrak{F}_+}[a_p^* a_p \widetilde G_{ \mathrm B}(z) ] + 	\Big | \sum_{ p\in P_{ \mathrm B } } \tr_{\mathfrak{F}_+}[a_p^* a_{-p}^* \widetilde G_{ \mathrm B}(z) ] \Big |  \Big) \\
		\lesssim & (N^{ \delta_{ \mathrm{Bog} } }\ell + N^{-1/3} +N^2 \ell^4 + N^{2/3 + \delta_{ \mathrm{ Bog } }} \ell^2 ) N^{2/3 + \delta_{ \mathrm{ Bog } }}.
	\end{split}
\end{equation} 

The expectation of the high-momentum term $\mathcal Q_>$ in $\widetilde \Gamma_0$ can be bounded by
\begin{equation} \label{eq:Q>_bound}
	\begin{split}
		\tr[ \mathcal Q_> \widetilde \Gamma_0 ] = & 2 \widetilde N_0 \sum_{ p \in \Lambda_+^* \setminus P_{\mathrm B} } \widehat \xi_N(p) \tr_{ \mathfrak{F} } [ a_p^* a_p \widetilde \Gamma_0 ] \\
		\leq & \frac {8\pi \mathfrak a_N} {1 - \frac{ \mathfrak a_N}{\ell}} \widetilde N_0  \sum_{ p \in \Lambda_+^* \setminus  P_{\mathrm B} }  \tr_{ \mathfrak{F} } [ a_p^* a_p \widetilde \Gamma_0 ] \le 8\pi \mathfrak a_N \widetilde N_0  \sum_{ p \in \Lambda_+^* \setminus  P_{\mathrm B} } \gamma_p + C [ \ell^{-1} + \exp(-c N^{\delta_{\mathrm{Bog}}} )],
	\end{split}
\end{equation} where we used $\widehat\xi_N(p) \le \widehat \xi_N(0)$, $\ell \ge 2 \mathfrak a _N$, $ \sum_{ p \in \Lambda_+^* \setminus  P_{\mathrm B} } \gamma(p) \lesssim N$ and \eqref{eq:cutoff_a_remainder_bound}. Similarly, we estimate the term related to the second-to-last term on the right-hand side of \eqref{eq:Xi_N_expansion} by
\begin{equation}\label{eq:Xi_N_N+}
	2 \widehat \xi_N(0) \tr_{ \mathfrak{F} }[ \mathcal N_+ \widetilde \Gamma_0] \le 8\pi \mathfrak a_N \widetilde N_0 \sum_{  p\in \Lambda_+^* }\gamma_p + C\ell^{-1}.
\end{equation}
Finally, the last term on the right-hand side of \eqref{eq:Xi_N_expansion} satisfies
\begin{equation}\label{eq:Xi_N_constant_term}
	\widehat\xi_N(0) \int_{\mathbb C} \zeta(z) |z|^4 \, \mathrm d z \le 4\pi \mathfrak a _N \int_{ \mathbb C } \zeta(z)|z|^4 \, \mathrm d z+ C \ell^{-1}.
\end{equation} Combining \eqref{eq:Xi_N_expectation_quartic}, \eqref{eq:Q<_QB}, \eqref{eq:EQ_bound}, \eqref{eq:Q>_bound}, \eqref{eq:Xi_N_N+} and \eqref{eq:Xi_N_constant_term}, we find the claim. 
\end{proof}

We can now conclude the proof of Proposition \ref{prop:renormalized_potential}. Inserting \eqref{eq:Xi_expectation_undressed} into \eqref{eq:Chi+V_expectation_computation}, we find \eqref{eq:renormalized_potential_statement}, with $\mathcal E_{\mathcal V} = V_1 + V_2 + V_3 + \mathcal E_{\Xi}$. The result now follows from \eqref{eq:V3_bound}, \eqref{eq:V1+V2_cancellation} and \eqref{eq:E_Xi_bound}. \qed

\subsection{Final upper bound}

We are now prepared to prove Proposition \ref{prop:energy_upper_bound}. Observing that $\mathcal K + \mathcal Q_{ \mathrm{ B } } = \mathcal H_{ \mathrm{ B } } (z) + \mu_0 \mathcal N_+$ and using \eqref{eq:bogoliubov_ham_diagonalization}, we find 
\begin{equation}\label{eq:K_QB_combination}
\begin{split}
		\tr_{\mathfrak{F}}[\mathcal K \widetilde \Gamma_0] +  \int_{\mathbb C } \zeta(z) \tr_{\mathfrak{F}_+} [ \mathcal Q_{\mathrm B} \widetilde G(z) ] \, \mathrm d z = & \tr_{\mathfrak{F}_+} [  \mathcal H^{ \mathrm{ diag } }  \widetilde G^{\mathrm { diag } } ] + \mu_0 \tr_{ \mathfrak{F}_+ } [ \mathcal N_+ \widetilde G(z) ] \\
		\le &\tr_{\mathfrak{F}_+} [   \mathcal H^{ \mathrm{ diag } }  G^{\mathrm { diag } } ] + \mu_0 \sum_{p\in \Lambda_+^*} \gamma_p  + \mathcal E_{\mathrm {Bog}},
\end{split}
\end{equation} with 
\begin{equation*}
 \mathcal E_{\mathrm {Bog}} = (\kappa_0^{-1} -1) \tr[ (   \mathcal H^{ \mathrm{ diag } }  - E_0 ) G^{\mathrm { diag } } ] + \mu_0\big( \tr_{ \mathfrak{F}_+ } [ \mathcal N_+  \widetilde G(z) ] - \tr_{ \mathfrak{F}_+ } [ \mathcal N_+  G(z) ] \big).
\end{equation*} 
The lowest eigenvalue $E_0 < 0$ of $\mathcal{H}_{\mathrm{B}}(z)$ has been defined in \eqref{eq:bogoliubov_ham_diagonalization}. The bounds $|p|^2 -\mu_0 \leq \epsilon(p) \lesssim |p|^2 - \mu_0$ and an application of Lemma~\ref{lem:sum_integral_approx} show
\begin{equation*}
 \tr[ (   \mathcal H^{ \mathrm{ diag } }  - E_0 ) G^{\mathrm { diag } } ] = \sum_{p\in \Lambda_+^*} \frac{ \eps(p) }{ \exp( \beta \eps(p) ) -1 } \lesssim \sum_{p\in \Lambda_+^*} \frac{ (|p|^2-\mu_0) }{ \exp( \beta( |p|^2 - \mu_0 ) ) -1 } \lesssim \beta^{-5/2}.
\end{equation*}
In combination with Lemma \ref{lem:Tz_N_bound}, Lemma \ref{lem:cutoff_a_remainder}, and $ \mu_0 = - \beta^{-1} \log( 1 + N_0 ^{-1})$, this implies
\begin{equation}\label{eq:E_Bog_bound}
\mathcal E _{ \mathrm{ Bog } } \lesssim \exp(-c N^{\delta_{ \mathrm{ Bog } }}),
\end{equation} for some $c>0$. The claim thus follows by combining \eqref{eq:K_QB_combination}, \eqref{eq:E_Bog_bound}, and Propositions \ref{prop:kinetic_energy_upper_bound} and \ref{prop:renormalized_potential}.  \qed

\section{Estimate for the entropy} \label{sec:entropy}

The goal of this section is to prove the lower bound in Proposition~\ref{prop:entropy} for the entropy of the trial state $\Gamma$ defined in \eqref{eq:trial_state}. It should be compared to \cite[Proposition~4.1]{BocDeuStock2023} and to \cite[Lemma~2]{Seiringer2006}. We also refer to the discussion in Remark~\ref{rem:correclationStructure} above. The main improvement is that we can estimate the influence of correlations added to the eigenfunctions of the state $|z \rangle \langle z | \otimes \widetilde{G}(z)$, while this has been possible in \cite{BocDeuStock2023} only for correlations added to those of $\widetilde{\Gamma}_0$. The additional freedom we obtain by doing this is a crucial ingredient for our analysis in Section~\ref{sec:energy}. We recall that the assumptions stated at the beginning of Section~\ref{sec:energy} also apply in this section.

\begin{proposition}
	\label{prop:entropy}
	The entropy of the state $\Gamma$ in \eqref{eq:trial_state} satisfies 
	\begin{equation}\label{eq:entropy_lower_bound}
		S(\Gamma) = S( \widetilde G^{\mathrm{ diag }}) + S^{ \mathrm{cl} }(\zeta) + \mathcal E_{\mathrm{S}},
	\end{equation} with $  \widetilde G^{ \mathrm{ diag } } $ in \eqref{eq:gibbs_state_bog_diag}, $\zeta$ in \eqref{eq:condensate_density_cutoff}, $S^{ \mathrm{cl} }$ defined below \eqref{eq:free_energy_functional_condensate} and
\begin{equation}
	\label{eq:entropy_error_bound}
	 \mathcal E_{\mathrm{S}} \gtrsim - N\ell^2.
\end{equation}
\end{proposition}
\begin{proof}
	We define the function $\ph(x)= - x\log(x)$ for $x \geq 0$ and we assume that $\{ \xi_k \}_k$ is an orthonormal basis of eigenfunctions of $\Gamma$. This allows us to write 
	\begin{equation*}
		\begin{split}
			S(\Gamma) = & \sum_{k} \ph \big( \langle \xi_k, \Gamma\xi_k\rangle \big ) = \sum_k \ph \Big( \int_{\mathbb C } \zeta(z)  \sum_{\alpha\in \widetilde{ \mathcal A}}\frac {  \widetilde \lambda_\alpha } { \| F\phi_{z, \alpha} \|^2 }  \left | \langle \xi_k, F\phi_{z, \alpha} \rangle \right|^2 \, \mathrm d z \Big ) . 
		\end{split}
	\end{equation*} 
Next, we define
	\begin{equation*}
		Z_k = \int_{\mathbb C } \sum_{\alpha\in \widetilde{ \mathcal A}}\frac { \left | \langle \xi_k, F\phi_{z, \alpha} \rangle \right|^2 } { \| F\phi_{z, \alpha} \|^2 } \, \mathrm d z
	\end{equation*} and observe that \eqref{eq:denominator_lower_bound} implies the bound
	\begin{equation}\label{eq:Zk_upper_bound}
		\begin{split}
			Z_k \le & (1+CN\ell^2) \int_{ \mathbb C } \sum_{\alpha\in \mathcal A} \left | \langle \xi_k, F\phi_{z, \alpha} \rangle \right|^2 \, \mathrm d z\\
			= & (1+CN\ell^2)\int_{ \mathbb C } \sum_{\alpha\in \mathcal A}  \langle z \otimes T_z \Psi_\alpha , F \xi_k\rangle \langle F \xi_k,   z\otimes T_z \Psi_\alpha\rangle  \, \mathrm d z \\
			= & (1+CN\ell^2)\int_{ \mathbb C } \langle z, \tr_{ \mathfrak{F}_+ } [F |\xi_k \rangle \langle \xi_k | F  ] z\rangle = (1+CN\ell^2)\tr [ F | \xi_k\rangle \langle \xi_k | F ] \le (1+CN\ell^2),
		\end{split}
	\end{equation} as long as $N\ell^2$ is sufficiently small. To come to the last line, we used the fact that the set $\{ T_z \Psi_{\alpha} \}_\alpha$ is an orthonormal basis of $\mathfrak{F}_+$ for fixed $z \in \mathbb{C}$, the completeness relation $ \int_{\mathbb C}|z\rangle \langle z| \, \mathrm d z = \mathds 1_{\mathfrak{F}_0} $, and the bound $F \leq 1$ for the multiplication operator $F$ on $\mathfrak{F}$. 
	
	Using the identity  $\ph(xy) = x\ph(y) + y\ph(x)$ for $x,y\ge 0$, we can thus write
	\begin{equation}\label{eq:entropy_computation_1}
		\begin{split}
			S(\Gamma) =\sum_k Z_k \ph\Big( \int_{\mathbb C} \sum_{\alpha \in \mathcal A} \zeta(z)\widetilde \lambda_\alpha \frac { \left | \langle \xi_k, F\phi_{z, \alpha} \rangle \right|^2 } { Z_k \| F\phi_{z, \alpha} \|^2 }  \, \mathrm d z \Big)  + \sum_k \ph\left(Z_k \right) \int_{\mathbb C} \sum_{\alpha\in \mathcal A} \zeta(z)\widetilde \lambda_\alpha \frac { \left | \langle \xi_k, F\phi_{z, \alpha} \rangle \right|^2 } { Z_k \| F\phi_{z, \alpha} \|^2 } \, \mathrm d z .
		\end{split}
	\end{equation} With \eqref{eq:Zk_upper_bound} we see that the second term on the right-hand side can be bounded by
	\begin{equation}\label{eq:entropy_error_bound_intermediate_1}
		\begin{split}
			- \sum_k\log(Z_k)  \int_{\mathbb C } \zeta(z) \sum_{\alpha\in \widetilde{ \mathcal A}} &  \widetilde \lambda_\alpha \frac {  \left | \langle \psi_k, F\phi_{z, \alpha} \rangle \right|^2 } { \| F\phi_{z, \alpha} \|^2 }  \\
			\gtrsim & - N\ell^2 \sum_k \int_{\mathbb C } \zeta(z)  \sum_{\alpha\in \widetilde{ \mathcal A}}\widetilde \lambda_\alpha \frac {  \left | \langle \psi_k, F\phi_{z, \alpha} \rangle \right|^2 } { \| F\phi_{z, \alpha} \|^2 }  = - N\ell^2,
		\end{split}
	\end{equation} where the equality in the second line follows from the fact that $\{ \xi_k\}_k$ is an orthonormal basis. Moreover, an application of Jensen's inequality to the strictly concave function $\ph$ shows
	\begin{equation}\label{eq:bound_entropy_main_term}
	\begin{split}
		\sum_k Z_k & \ph\Big( \int_{\mathbb C} \sum_{\alpha \in \widetilde{ \mathcal A}} \zeta(z)\widetilde \lambda_\alpha \frac {  \left | \langle \xi_k, F\phi_{z, \alpha} \rangle \right|^2 } { Z_k \| F\phi_{z, \alpha} \|^2 }  \, \mathrm d z \Big) \ge \sum_k Z_k  \int_{\mathbb C} \sum_{\alpha \in \widetilde{ \mathcal A}} \ph(\zeta(z) \widetilde \lambda_\alpha) \frac {  \left | \langle \xi_k, F\phi_{z, \alpha} \rangle \right|^2 } { Z_k \| F\phi_{z, \alpha} \|^2 }  \, \mathrm d z  \\
		= & \int_{ \mathbb C } \sum_{\alpha\in\widetilde{ \mathcal A}} \ph(\widetilde \lambda_\alpha \zeta(z))  \, \mathrm d z = \sum_{\alpha \in \widetilde{ \mathcal A} } \ph(\widetilde \lambda_\alpha)+ \int_{ \mathbb C } \ph(\zeta(z)) \, \mathrm d z = S( \widetilde G^{\mathrm{diag}}) + S^{\mathrm{cl}}(\zeta).
	\end{split}
\end{equation} We insert \eqref{eq:entropy_error_bound_intermediate_1} and \eqref{eq:bound_entropy_main_term} into \eqref{eq:entropy_computation_1} to conclude the proof. 
\end{proof}

\section{Proof of Theorem \ref{thm:main_thm}} \label{sec:proof_main_theorem}

Let us fix some $0 < \epsilon_0 < 1/12$. To prove \eqref{eq:main_result} we show two distinct upper bounds, corresponding to the two terms in the minimum appearing on the right-hand side. The first upper bound is obtained with the trial state $\Gamma$ in \eqref{eq:trial_state} and yields the term $F^{\mathrm{BEC}} - 8 \pi \mathfrak a_N N_0^2$ in the minimum. It is only valid in the parameter regime in which $N_0 \gtrsim N^{2/3 + \eps_0}$ holds. A second bound that is valid for all $N_0 \leq N$ (but useful only if $N_0 \ll N^{5/6}$) and contributes the term $F_0^{\mathrm{BEC}}$ in the minimum will be obtained at the end of this section with a much simpler trial state. Before we discuss these issues in more detail, we provide the final upper bound for the free energy of $\Gamma$. We recall that for this part of the analysis the assumptions stated at the beginning of Section~\ref{sec:energy} hold.

In combination, Proposition \ref{prop:energy_upper_bound} and \ref{prop:entropy} imply the upper bound
\begin{equation}\label{eq:free_energy_upper_bound_intermediate_1}
\begin{split}
		\mathcal F(\Gamma) = & \tr[\mathcal H_N \Gamma] - \beta^{-1} S(\Gamma) \\
		\le & \tr [ \mathcal H^{ \mathrm { diag } } G^{ \mathrm{ diag} }  ] -\beta^{-1} S( \widetilde G^{\mathrm{ diag }})  + 4\pi \mathfrak a _N \int_{\mathbb C } \zeta(z) |z|^4 \, \mathrm d z +\beta^{-1} S^{ \mathrm{ cl } }(\zeta) + \mu_0 \sum_{p \in \Lambda_+^*} \gamma_p  \\
		& +   8\pi \mathfrak a_N  \Bigg[  \Big( \sum_{p \in \Lambda_+^*} \gamma_p \Big)^2 + \widetilde N_0  \sum_{ p \in \Lambda_+^* \setminus  P_{\mathrm B} } \gamma_p + \widetilde N_0 \sum_{p \in \Lambda_+^*} \gamma_p \Bigg] + \mathcal E_{\mathcal H} - \beta^{-1}\mathcal E_{\mathrm{S}},
\end{split}
\end{equation} where $\mathcal E_{ \mathcal H }$ and $\mathcal E_{\mathrm{S}}$ satisfy \eqref{eq:energy_error_bound} and \eqref{eq:entropy_error_bound}, respectively.

We first remove the particle number cutoff from the entropy of the Bogoliubov Gibbs state. As in Section~\ref{sec:entropy} we denote $\ph(x) = -x\log(x)$ for $x \geq 0$. Using $\kappa_0 = \sum_{\alpha \in \widetilde {\mathcal{A}} } \lambda_\alpha \le 1$, we find 
\begin{equation}\label{eq:entropy_cutoff_removal_1}
	\begin{split}
		S( \widetilde G^{ \mathrm { diag } }) = \sum_{\alpha \in \widetilde{ \mathcal A} } \ph(\widetilde \lambda_\alpha) \ge & \sum_{\alpha \in \widetilde{ \mathcal A} } \ph( \lambda_\alpha) + \ph(\kappa_0^{-1}) \sum_{\alpha \in \widetilde{ \mathcal A} } \lambda_\alpha = S ( G^{ \mathrm{ diag} }) + \log(\kappa_0) - \sum_{\alpha \in \mathcal A \setminus \widetilde{ \mathcal A} } \ph( \lambda_\alpha).
	\end{split}
\end{equation} 
An application of \eqref{eq:cutoff_a_remainder_bound} shows
\begin{equation}\label{eq:log_k0_bound}
	\log(\kappa_0) \gtrsim - \exp( - c N^{\delta_{ \mathrm{ Bog } }}).
\end{equation} Using the definition of $\lambda_\alpha$ above \eqref{eq:Eigenvalues_Psi_alpha} and \eqref{eq:cutoff_a_remainder_bound_energy}, we find 
\begin{equation*}
	\begin{split}
		\sum_{\alpha \in \mathcal A \setminus \widetilde{ \mathcal A} } \ph( \lambda_\alpha)  = & \sum_{\alpha \in \mathcal A \setminus \widetilde{ \mathcal A} }\lambda_\alpha \log\Big(\sum_{\alpha' \in  \mathcal A } e^{-\beta E_{\alpha'}}\Big) + \beta \sum_{\alpha \in \mathcal A \setminus \widetilde{ \mathcal A} }\lambda_\alpha E_\alpha \\
		\le & \sum_{\alpha \in \mathcal A \setminus \widetilde{ \mathcal A} }\lambda_\alpha \log\Big(\sum_{\alpha' \in  \mathcal A } e^{-\beta E_{\alpha'}}\Big) + C \exp(- c N^{\delta_{ \mathrm{ Bog } }}),
	\end{split}
\end{equation*} with $E_\alpha$ in \eqref{eq:Eigenvalues_Psi_alpha}. A standard computation shows that 
\begin{equation*}
	\begin{split}
		\log\Big(\sum_{\alpha \in  \mathcal A } e^{-\beta E_\alpha}\Big)  = & - \sum_{p\in\Lambda^*_+} \log\left( 1 - \exp(-\beta\eps(p)) \right)  \le - \sum_{p\in\Lambda^*_+} \log\left( 1 - \exp(-\beta|p|^2) \right) \\
		\lesssim & -\beta^{-3/2} \int_{0}^\infty (\beta+x^2) \log(1-\exp(-x^2))\,\mathrm d x \lesssim N.
	\end{split}
\end{equation*} The inequalities above follow from $\eps(p) \ge |p|^2$, the fact that $x\mapsto - \log(1-\exp(-x))$ is monotone decreasing for $x\ge 0$, and Lemma \ref{lem:sum_integral_approx} with $\lambda = 0$. In combination with \eqref{eq:cutoff_a_remainder_bound}, this implies
\begin{equation}
	\label{eq:entropy_lambda_remainder_bound}
	\sum_{\alpha \in \mathcal A \setminus \widetilde{ \mathcal A} }\lambda_\alpha \log\Big(\sum_{\alpha \in  \mathcal A } e^{-\beta E_\alpha}\Big) \lesssim \exp(-c N^{ \delta_{ \mathrm{ Bog } } } ).
\end{equation}
Inserting \eqref{eq:log_k0_bound} and \eqref{eq:entropy_lambda_remainder_bound} into \eqref{eq:entropy_cutoff_removal_1} yields
\begin{equation}
	\label{eq:entropy_error_2_bound}
	S(\widetilde G^{ \mathrm {diag} }) \ge 	S( G^{ \mathrm {diag} }) - C\exp(-c N^{ \delta_{ \mathrm{ Bog } } } ).
\end{equation} 

We also have
\begin{equation}\label{eq:free_energy_cloud_computation_1}
	\tr [ \mathcal H^{ \mathrm{ diag } } G^{ \mathrm{ diag} }  ] -\beta^{-1} S(G^{\mathrm{ diag }}) = \beta^{-1}\sum_{p\in\Lambda_+^*}\log\big( 1-e^{-\beta\eps(p)} \big ) .
\end{equation} The following Lemma, which is proved in \cite[Appendix B]{BocDeuStock2023}, provides us with an asymptotic expansion of the term on the right-hand side of \eqref{eq:free_energy_cloud_computation_1}.
\begin{lemma}
	\label{lem:bogoliubov_free_energy_expansion}
	Let $\beta =\kappa \beta_{\mathrm c}$, with $\beta_{\mathrm c}$ defined in \eqref{eq:Tc_ideal} and $\kappa \in (0,\infty)$. There exists a constant $C>0$ such that, for every $N$, 
	\begin{equation*}
\begin{split}
			\frac 1 \beta \sum_{ p\in P_{ \mathrm B } } & \log \big( 1-e^ {-\beta\eps(p)} \big)  \\
			\le & \frac 1 \beta \sum_{ p\in P_{ \mathrm B } } \log \big( 1- e^{-\beta(|p|^2-\mu_0)}  \big) + 8\pi \mathfrak a_N N_0 \sum_{ p\in P_{ \mathrm B } } \frac 1 { e^ {\beta(|p|^2-\mu_0)}  -1} \\
			& - \frac{1}{2\beta} \sum_{ p\in \Lambda_+^*} \Bigg[ \frac{16\pi  \mathfrak a_N N_0}{|p|^2} - \log\Big( 1 +  \frac{16\pi \mathfrak a_N N_0}{|p|^2}   \Big) \Bigg] + \frac{C N_0^2}{N^2} \Big( N^{ \delta_{ \mathrm{ Bog } } } + \frac 1 {\beta N^{ \delta_{ \mathrm{ Bog } } }  } + \frac{1}{\beta^2 N_0} \Big).
\end{split}
	\end{equation*}
\end{lemma}

An application of Lemma \ref{lem:bogoliubov_free_energy_expansion} shows 
\begin{equation}\label{eq:free_energy_cloud_computation_2}
\begin{split}
		\frac 1 \beta \sum_{p\in\Lambda_+^*} &\log\big( 1-e^{-\beta\eps(p)} \big) \le  \frac 1 \beta \sum_{ p \in \Lambda_+^*} \log \big( 1- e^{-\beta(|p|^2-\mu_0)}  \big ) + 8\pi \mathfrak a_N N_0 \sum_{ p\in P_{ \mathrm B } }\gamma^{ \mathrm{ id } }_p \\
		& - \frac{1}{2\beta} \sum_{ p\in \Lambda_+^*} \Bigg[ \frac{16\pi  \mathfrak a_N N_0}{|p|^2} - \log \Big( 1 +  \frac{16\pi \mathfrak a_N N_0}{|p|^2}   \Big) \Bigg] + C N^{2/3 -\delta_{ \mathrm{ Bog } }},
\end{split}
\end{equation} where we introduced the notation $ \gamma^{ \mathrm{ id } }_p = \big( \exp ( \beta(|p|^2-\mu_0) )  -1 \big)^{-1} $ for $p\in \Lambda_+^*$. 

The sum of the third and fourth term on the right-hand side of \eqref{eq:free_energy_upper_bound_intermediate_1} equals
\begin{equation}\label{eq:free_energy_condesate_computation}
4\pi \mathfrak a _N \int_{\mathbb C } \zeta(z) |z|^4 \, \mathrm d z +\beta^{-1} S^{ \mathrm{ cl } }(\zeta)  = -\beta^{-1} \log\Big( \int_{|z|^2 \le \widetilde N} e^{ -\beta( 4\pi \mathfrak a_N |z|^4 - \widetilde\mu |z|^2 ) } \, \mathrm d z \Big)+\widetilde \mu \widetilde N_0,
\end{equation} with $\widetilde N_0$ defined in \eqref{eq:n0_tilde_def}, and where we used the shorthand notation $\widetilde N = \widetilde c N$. Extending the integral inside the logarithm to the whole complex plane, we generate the error term
\begin{equation*}
		-\beta^{-1} \log\Big( 1- \frac {\int_{|z|^2 > \widetilde N} e^{ -\beta( 4\pi \mathfrak a_N |z|^4 - \widetilde\mu |z|^2 ) } \, \mathrm d z} { \int_{\mathbb C} e^{ -\beta( 4\pi \mathfrak  a_N |z|^4 - \widetilde\mu |z|^2 ) } \, \mathrm d z } \Big).
\end{equation*} Setting $\eta = \widetilde \mu \sqrt{\beta/(4h)}$ and $A = \sqrt{\beta h} \widetilde c N$, and computing the integrals in polar coordinates with the change of variables $t = \sqrt{\beta h} |z|^2 - \eta$, we find 
\begin{equation*}
	\frac {\int_{|z|^2 > \widetilde N} e^{ -\beta( 4\pi \mathfrak  a_N |z|^4 - \widetilde\mu |z|^2 ) } \, \mathrm d z} { \int_{\mathbb C} e^{ -\beta( 4\pi \mathfrak  a_N |z|^4 - \widetilde\mu |z|^2 ) }  \, \mathrm d z } = \frac{ \int_{A-\eta}^\infty e^{-t^2}\,\mathrm dt }{ \int_{-\eta}^\infty e^{-t^2}\,\mathrm dt  }.
\end{equation*} We know from the proof of Lemma \ref{lem:mu_tilde_condensate} that $\eta < A/2$. If $0 \le \eta < A/2$, we estimate the denominator from below by a constant and apply \eqref{eq:erf_asymptotics} to obtain a bound for the numerator. This shows
\begin{equation*}
	\frac{ \int_{A-\eta}^\infty e^{-t^2}\,\mathrm dt }{ \int_{-\eta}^\infty e^{-t^2}\,\mathrm dt} \lesssim  \int_{A/2}^\infty e^{-t^2}\,\mathrm dt \lesssim \exp(-cN^{1/3}). 
\end{equation*} In contrast, if $ \eta < 0$ we apply \eqref{eq:erf_asymptotics} in both the numerator and the denominator, and find
\begin{equation*}
	\frac{ \int_{A-\eta}^\infty e^{-t^2}\,\mathrm dt }{ \int_{-\eta}^\infty e^{-t^2}\,\mathrm dt} \lesssim e^{-A^2 + 2A\eta} \lesssim  \exp(-cN^{1/3}).
\end{equation*}
Hence, 
\begin{equation}\label{eq:free_energy_condensate_cutoff_remainder}
\begin{split}
	-\beta^{-1} \log\Big( \int_{|z|^2 \le \widetilde N} & e^{ -\beta( 4\pi \mathfrak  a_N |z|^4 - \widetilde\mu |z|^2 ) } \, \mathrm d z \Big) \\
	\le & -\beta^{-1} \log\Big( \int_{\mathbb C} e^{ -\beta( 4\pi \mathfrak  a_N |z|^4 - \widetilde\mu |z|^2 ) } \, \mathrm d z \Big) + Ce^{-cN^{1/3}} \\   \le & -\beta^{-1} \log\Big( \int_{\mathbb C} e^{ -\beta( 4\pi \mathfrak  a_N |z|^4 - \mu |z|^2 ) } \, \mathrm d z \Big) - N_0(\widetilde \mu -\mu) + Ce^{-cN^{1/3}}, 
\end{split}
\end{equation} with $\mu$ in \eqref{eq:mu_definition}. In the last step, we used the fact that the first term on the second line of \eqref{eq:free_energy_condensate_cutoff_remainder} is a concave function of $\widetilde \mu$ and that
\begin{equation*}
		\frac 1 \beta \frac{\partial}{\partial \widetilde \mu}  \log\Big( \int_{\mathbb C} e^{ -\beta( 4\pi \mathfrak  a_N |z|^4 - \widetilde \mu |z|^2 ) } \, \mathrm d z \Big) \Big | _{ \widetilde \mu = \mu } = N_0
\end{equation*}
holds.

Inserting \eqref{eq:entropy_error_2_bound}--\eqref{eq:free_energy_condensate_cutoff_remainder} into \eqref{eq:free_energy_upper_bound_intermediate_1}, we get 
\begin{equation}\label{eq:upper_bound_intermediate}
\begin{split}
	 \mathcal F(\Gamma) \le & \frac 1 \beta \sum_{|p|\in \Lambda_+^*} \log( 1- e^{-\beta(|p|^2-\mu_0)}  ) + \mu_0 \sum_{p \in \Lambda_+^*} \gamma_p -\beta^{-1} \log\left( \int_{\mathbb C} e^{ -\beta( 4\pi \mathfrak  a_N |z|^4 - \mu |z|^2 ) }\right) \\
	 &  +   8\pi \mathfrak a_N  \Bigg[  N_0 \sum_{ p\in P_{ \mathrm B } } \gamma_p^{ \mathrm{ id } } + \Big( \sum_{p \in \Lambda_+^*} \gamma_p \Big)^2 + \widetilde N_0  \sum_{ p \in \Lambda_+^* \setminus  P_{\mathrm B} } \gamma_p + \widetilde N_0 \sum_{p \in \Lambda_+^*} \gamma_p \Bigg] \\
	 & + \mu N_0 + \widetilde \mu ( \widetilde N_0 - N_0) - \frac{1}{2\beta} \sum_{ p\in \Lambda_+^*} \Bigg[ \frac{16\pi a_N N_0}{|p|^2} - \log\Big( 1 +  \frac{16\pi a_N N_0}{|p|^2}   \Big) \Bigg] \\
	 &  + C N^{2/3 -\delta_{ \mathrm{ Bog } }} + \mathcal E_{\mathcal H} - \beta^{-1}\mathcal E_{\mathrm{S}}.
\end{split}
\end{equation} 
With Lemma \ref{lem:perturbation_number_particles}, Lemma \ref{lem:Tz_N_bound} and Lemma \ref{lem:cutoff_a_remainder}, we infer that
\begin{equation*}
	\sum_{  p\in \Lambda_+^* } \gamma_p = N - \widetilde N_0 + E,
\end{equation*} with $E \lesssim N^{5/3 + \delta_{ \mathrm{ Bog } }} \ell^2 + N^{3}\ell^4 \lesssim N$. Using this, we see that the terms inside the bracket on the second line of \eqref{eq:upper_bound_intermediate} are bounded from above by
\begin{equation*}
\begin{split}
		2\widetilde N_0 & (N - \widetilde N_0) + (N - \widetilde N_0)^2 + N_0\sum_{p \in P_{ \mathrm{ B } } } \gamma^{ \mathrm {id} }_p - \widetilde N_0\sum_{p \in P_{ \mathrm{ B } } } \gamma_p + CNE \\
		= &  2\widetilde N_0 (N - \widetilde N_0) + (N - \widetilde N_0)^2 + \widetilde N_0 ( \widetilde N_0 - N_0 ) + (N_0 - \widetilde N_0) \sum_{p \in P_{ \mathrm{ B } } } \gamma^{ \mathrm {id} }_p + CNE \\
		= & N^2 - N_0^2 + \widetilde N_0(N_0 - \widetilde N_0) + (N_0 - \widetilde N_0)^2 + (N_0 - \widetilde N_0) \sum_{p \in P_{ \mathrm{ B } } } \gamma^{ \mathrm {id} }_p + CNE. 
\end{split}
\end{equation*} Moreover, an application of Lemma \ref{lem:n0_n0tilde} shows
\begin{equation*}
\begin{split}
		8 \pi \mathfrak a_N \Big[ (N_0 - \widetilde N_0)^2 + (N_0 - \widetilde N_0) \sum_{p \in P_{ \mathrm{ B } } } \gamma^{ \mathrm {id} }_p  + CNE \Big] 
		\lesssim & N^{1/3 + \delta_{ \mathrm{ Bog } }} + N^{5/3 + \delta_{ \mathrm{ Bog } }} \ell^2 + N^{3}\ell^4.
\end{split}
\end{equation*} The above considerations show that the second line of \eqref{eq:upper_bound_intermediate} is bounded from above by 
\begin{equation*}
	8\pi \mathfrak a_N (N^2 - N_0^2) + 8\pi \mathfrak a_N \widetilde N_0 (N_0 - \widetilde N_0) + C [ N^{1/3 + \delta_{ \mathrm{ Bog } }} + N^{5/3 + \delta_{ \mathrm{ Bog } }} \ell^2 + N^{3}\ell^4 ].
\end{equation*}

Next, we consider the second term in the above equation and the second term on the third line of \eqref{eq:upper_bound_intermediate}. Let $\eps\in(0, \eps_0)$. If $ N_0 \gtrsim N^{5/6+\eps}$, we infer from \eqref{eq:asymptotics_mu_tilde_large} in Appendix C that
\begin{equation*}
	(8\pi \mathfrak a_N \widetilde N_0 -\widetilde \mu ) (N_0 - \widetilde N_0) \lesssim e^{-cN^\eps}.
\end{equation*}
On the other hand, if $ N_0 \lesssim N^{5/6+\eps}$, Lemma \ref{lem:n0_n0tilde}, Lemma \ref{lem:mu_tilde_condensate} and the lower bound $N_0 \gtrsim N^{2/3 + \eps_0}$ imply
\begin{equation*}
	|8\pi \mathfrak a_N \widetilde  N_0 ( N_0 - \widetilde N_0) | \lesssim N^{1/3+2\eps} + N^{3/2+\delta_{ \mathrm{ Bog } } + \eps} \ell^2 + N^{17/6 + \eps} \ell^4 
\end{equation*} and
\begin{equation*}
\begin{split}
	 | \widetilde{\mu} (\widetilde N_0 - N_0) | \lesssim & \Big( \frac 1 {\widetilde N_0\beta} + \frac{1}{\sqrt{\beta N}} + \frac{\widetilde N_0}{N} \Big) \Big ( \frac{N_0}{N\beta} + \frac{N_0^2}{N^2}  + N^{3}\ell^4 + N^{5/3 + \delta_{ \mathrm{ Bog } }} \ell^2 \Big) \\
 \lesssim & N^{1/3 + 2\eps} + N^{3}\ell^4 + N^{5/3 + \delta_{ \mathrm{ Bog } }} \ell^2.
\end{split} 
\end{equation*} With the bound $|\mu_0| = \beta^{-1} \log( 1 + N_0^{-1}) \leq 1/(N_0\beta)$, we see that
\begin{equation*}
\begin{split}
		\mu_0 \sum_{p \in \Lambda_+^*} \gamma_p = & \mu_0 (N- N_0) + \mu_0 ( N_0 - \widetilde N_0 + E ) \\
		\le & \mu_0 (N- N_0)  + C ( N^{1/3 } + N^{5/3 + \delta_{ \mathrm{ Bog } }} \ell^2 + N^{3}\ell^4 ).
\end{split}
\end{equation*}

We collect the above estimates, insert the bounds for $\mathcal E_{\mathcal H}$ and $\mathcal E_{\mathrm{S}}$ in \eqref{eq:energy_error_bound} and \eqref{eq:entropy_error_bound}, respectively, choose $\eps$ sufficiently small, and find 
\begin{equation}\label{eq:upper_bound_final}
	\begin{split}
		\mathcal F(\Gamma) \le & \frac 1 \beta \sum_{|p|\in \Lambda_+^*} \log( 1- e^{-\beta(|p|^2-\mu_0)}  ) + \mu_0 (N-N_0) -\frac 1 \beta \log\left( \int_{\mathbb C} e^{ -\beta( 4\pi a_N |z|^4 - \mu |z|^2 ) }\right) + \mu N_0 \\
		& + 8\pi \mathfrak a_N (N^2-N_0^2) - \frac{1}{2\beta} \sum_{ p\in \Lambda_+^*} \Bigg[ \frac{16\pi a_N N_0}{|p|^2} - \log\Bigg( 1 +  \frac{16\pi a_N N_0}{|p|^2}   \Bigg) \Bigg]  \\
		& + C \Big( N^{2/3 - \delta_{ \mathrm{ Bog } }} + N^{2/3 + 3\delta_{ \mathrm{ Bog } }} \ell^{1/2} + N^{5/3 + 3\delta_{ \mathrm{ Bog } }} \ell^2 + N^3 \ell^4 + \ell^{-1} \Big).
	\end{split}
\end{equation} The error size is minimized by choosing $\delta_{\mathrm{Bog}} = 1/18$, $\ell = N^{-11/18}$,  which yields the first upper bound.

To prove an upper bound that is valid in the non-condensed phase, we need a different trial state. In this regime, the Bogoliubov corrections become negligible compared to the size of the error term in Theorem \ref{thm:main_thm}. As undressed trial state, we can thus take the Gibbs state of the ideal gas with an appropriate cutoff for the number of particles: 
\begin{equation*}
	\widetilde G^{ \mathrm{ id } } = \frac{ \mathds 1_{\{ \mathcal N \le \widetilde c N \}} \exp( -\beta( \mathrm d \Gamma( -\Delta - \widetilde \mu_0 )  ) ) }{ \tr_{ \mathfrak{F} }[ \mathds 1_{\{ \mathcal N \le \widetilde c N \}} \exp( -\beta( \mathrm d \Gamma( -\Delta - \widetilde \mu_0 )  ) ) ] },
\end{equation*} where $\widetilde c>1$ is chosen independent of $N$. The chemical potential $\widetilde \mu_0$ is uniquely determined by the condition $ \tr_{ \mathfrak{F} } [ \mathcal N \widetilde G^{ \mathrm{ id } } ] = N$. The final trial state $\widetilde \Gamma$ is obtained by dressing $\widetilde G^{ \mathrm{id} }$ with the correlation structure $F$ in the same way as we did in \eqref{eq:trial_state}. Since the eigenfunctions of $\widetilde G^{ \mathrm {id} }$ can be chosen to be eigenfunctions of $\mathcal N$, it is easy to check that, in contrast to the previous regime, the correlation structure does not alter the expected number of particles in the trial state. Computing the free energy of $\widetilde \Gamma$ is a straightforward replication of Sections \ref{sec:energy} and \ref{sec:entropy} in a simplified setting. We therefore leave the details to the reader and only state the final result: 
\begin{equation}\label{eq:upper_bound_final2}
	\mathcal F( \widetilde \Gamma ) \le F_0(\beta, N) + 8\pi \mathfrak a_N N^2 + CN^{11/18}. 
\end{equation} 

Theorem \ref{thm:main_thm} is a direct consequence of \eqref{eq:upper_bound_final} and \eqref{eq:upper_bound_final2}. To see this, let us make the following remarks. Using Proposition \ref{prop:f_bec} one easily checks that the minimum in \eqref{eq:main_result} is attained by the first term if $N_0 \gtrsim N^{5/6 + \epsilon}$ with $\epsilon > 0$. Since the term on the second line of \eqref{eq:main_result} is negative, it can be pulled out of the minimum in this parameter regime. If $N_0 \lesssim N^{5/6 + \epsilon}$ this term can also be pulled out of the minimum, because it is bounded in absolute value by a constant times $N^{1/3 + 2 \epsilon}$, which is smaller than our remainder term if $\epsilon$ is chosen sufficiently small. Moreover, a short computation that uses Proposition~\ref{prop:f_bec} and the identity $\mu_0 = -\ln(1+N_0^{-1})/\beta$ show that the condition $N_0\lesssim N^{2/3 + \eps_0}$ with $\eps_0< 1/12$ implies 
\begin{equation}\label{eq:bounds_minimum}
	F^{\mathrm{ BEC }} - 8\pi \mathfrak a_N N_0^2 \ge F^{\mathrm{ BEC }}_0 - \frac {\mu_0} 2 + \mathcal O(N^{1/2}).
\end{equation} Using additionally $\mu_0 < 0$, we thus find
\begin{equation*}
	\min\{ F^{\mathrm{ BEC }} - 8\pi \mathfrak a_N N_0^2, F^{\mathrm{ BEC }}_0 \} = F^{\mathrm{ BEC }}_0 + \mathcal O(N^{1/2}).
\end{equation*} This explains why it is sufficient to prove the bound for the trial state $\Gamma$ only in the parameter regime $N_0 \gtrsim N^{2/3+\eps_0}$, concluding our proof of Theorem \ref{thm:main_thm}. \qed


%

\textbf{Acknowledgments.}
A. D. gratefully acknowledges funding from the Swiss National Science Foundation (SNSF) through the Ambizione grant PZ00P2 185851. M. C. gratefully acknowledges support from the European Research Council through the ERC--AdG CLaQS. It is our pleasure to thank Chiara Boccato, Phan Thành Nam, Marcin Napiórkowski, Alessandro Olgiati and Benjamin Schlein for inspiring discussions.

\begin{appendix}
		
	\begin{center}
		\huge \textsc{--- Appendix ---}
	\end{center}

\section{Properties of the solution to the scattering equation} \label{app:scattering}

In this appendix we recall some well-known properties of the solution $f(|x|)$ to the zero energy scattering equation and of $f_{\ell}(x)$ defined above and in \eqref{eq:Definitionfl}, respectively. In the whole section we assume that $V$ is a nonnegative measurable and radial function that vanishes outside the ball with radius $R>0$. All results that we state without proof can be found in \cite[Appendix C]{LiebSeiYng_BoseGas}\footnote{For general interaction potentials, the scattering equation in \cite[Theorem C.1]{LiebSeiYng_BoseGas} has to be understood in the sense of quadratic forms and not in the sense of distributions as claimed. That is, functions used to test the equation should be elements of the form domain of the energy functional $\mathcal E_R$ (notation from the reference). All proofs in the reference apply, with minor adjustments.}. 


Let us introduce the energy functional
\begin{equation*}
	\mathcal E_\ell [ \phi ] = \int_{ B_\ell } \left( | \nabla \phi |^2 + \frac 1 2 V_N(x) | \phi(x) |^2 \right) \mathrm d x
\end{equation*}
with $\ell > R$. The function $f_{\ell}$ is the unique minimizer of $\mathcal E_\ell$ among all $H^1$ functions $\phi$ satisfying the boundary condition $\phi(x) = 1$ for $|x| = \ell$, and its energy is given by
\begin{equation}
	\label{eq:energy_f_ell}
	\mathcal{E}_\ell[f_\ell] = \min_{ \substack{\phi\in H^1(B_\ell) \\ \phi = 1 \text{ on } \partial B_\ell } } \mathcal E_\ell[\phi] =  \frac {4\pi \mathfrak a_N} {1- \mathfrak a_N/\ell}.
\end{equation} 
Here $\mathfrak a_N < \ell$ is a positive number called the scattering length of the potential $V_N$. It is easy to see that, by scaling, $\mathfrak a_N = N^{-1}\mathfrak a$, where $\mathfrak a$ is the scattering length of the unscaled potential $V$. 

The function $f(|x|)$ is monotonically non-decreasing in $|x|$, and it is bounded from above and from below by
	\begin{equation*}
		1 \geq f(r) \ge 1-\frac{\mathfrak a_N}r,
	\end{equation*} 
with equality in the lower bound for $r \ge R$. This, in particular, implies 
	\begin{equation}\label{eq:f_ell_pointwise_bound}
		1 \ge f_\ell(x) \ge \frac { (1 - \mathfrak a_N / |x|)_+ } { 1 - \mathfrak a_N / \ell } \ge ( 1 - \mathfrak a_N/|x| )_+. 
	\end{equation} 
	
	We also need the following integral bounds on $f_\ell$. 
	\begin{lemma} Let $u_\ell = 1- f_\ell$. We have
		\begin{equation}\label{eq:ul_integral}
			\int_{\mathbb R^3} u_\ell(x) \mathrm d x \lesssim \mathfrak a_N \ell^2, \qquad 	\int_{\mathbb R^3} u_\ell(x)^2 \mathrm d x \lesssim \mathfrak a_N^2 \ell, \quad \text{ and } \quad \int_{\mathbb R^3}|\nabla f_\ell(x)| \, \mathrm dx \lesssim \mathfrak{a}_N \ell.
		\end{equation} 
	\end{lemma}
	\begin{proof}
		With \eqref{eq:f_ell_pointwise_bound}, we see that
		\begin{equation*}
			0 \le u_\ell(x) \le \begin{cases}
				0 & |x| \ge \ell, \\
				2 \mathfrak a_N / |x| & \mathfrak a_N \le |x| < \ell, \\
				1 & |x| < \mathfrak  a_N,
			\end{cases}
		\end{equation*} 
		which implies the first two bounds in \eqref{eq:ul_integral}. The last bound follows from $f'\ge 0$ and one integration by parts:
		\begin{equation*}
			\begin{split}
				\int|\nabla f_\ell| = & \frac {4\pi} {f(\ell)} \int_0^\ell f'(r) r^2 \mathrm d r =  \frac {4\pi} {f(\ell)}  \left[ f(\ell) \ell^2 - 2\int_0^\ell f(r) r \mathrm d r \right] \\ 
				\le & 4\pi\ell^2 - 8\pi\int_0^\ell (1- \mathfrak a_N/r)_+ r \mathrm d r = 8\pi\mathfrak a_N\ell -4\pi \mathfrak a_N^2 \le 8\pi\mathfrak a_N\ell.
			\end{split}
		\end{equation*} To come to the second line, we additionally used \eqref{eq:f_ell_pointwise_bound}.
	\end{proof}

\section{Properties of the effective functional for the condensate} \label{app:condensate_functional}

We recall the definition of $\zeta$ in \eqref{eq:condensate_density_cutoff} with the cutoff parameter $\widetilde{c}$. The following is an adaptation of \cite[Lemma C.1]{BocDeuStock2023}.
\begin{lemma}\label{lem:mu_tilde_condensate}
	We assume that $\widetilde{c} > 2$ and consider the combined limit $N\to\infty$, $\beta = \kappa\beta_{\mathrm c}$ with $\kappa\in(0,\infty)$ and $\beta_{\mathrm c}$ in \eqref{eq:Tc_ideal}. Let $1 < \widetilde b < \widetilde c /2 $ and choose a sequence $M= M(N) \in\mathbb R$ such that $0 < M < \widetilde b N$ for every $N\in\mathbb N$. Then there exists a unique $\widetilde \mu \in \mathbb{R}$ such that $\int_{\mathbb C}|z|^2 \zeta(z)\,\mathrm d z = M$, with $\zeta$ in \eqref{eq:condensate_density_cutoff}. Moreover, for every $\eps > 0$, there exists $c>0$ such that:
	\begin{enumerate}
		\item If $M \gtrsim N^{5/6+\eps}$, then
		\begin{equation}\label{eq:asymptotics_mu_tilde_large}
			\left| \widetilde \mu - 8\pi \mathfrak a_N M \right| \lesssim e^{-cN^\eps}.
		\end{equation}
		\item If $M\lesssim N^{5/6 - \eps}$, then
		\begin{equation}\label{eq:asymptotics_mu_tilde_small}
			\left| \widetilde \mu + \frac{1}{\beta M} \right| \lesssim \frac{N^{-2\eps}}{\beta M}.
		\end{equation}
		\item For any $M=M(N)$, we have
		\begin{equation}\label{eq:mu_tilde_general_bound}
			|\widetilde\mu | \lesssim \left( \frac{1}{M\beta} + \frac{1}{\sqrt{\beta N}} + \frac{M}{N} \right).
		\end{equation}
	\end{enumerate}
\end{lemma}
\begin{proof}
	Proceeding as in the proof of \cite[Lemma C.1]{BocDeuStock2023}, we compute 
	\begin{equation} \label{eq:mu_definition_implicit}
		\sqrt{\beta h}\int_{\mathbb C}\zeta(z)|z|^2\,\mathrm d z = \frac{1-\exp(-A^2+2A\eta) + \sqrt \pi \eta \exp(\eta^2) \erf[-\eta,A - \eta]}{ \sqrt \pi  \exp(\eta^2) \erf[-\eta,A - \eta] } \eqqcolon \Upsilon(\eta),
	\end{equation} where $h = 4\pi \mathfrak a_N$, $\eta = \widetilde \mu \sqrt{\beta/(4h)}$, $A=\widetilde c  \sqrt{\beta h} N$ and 
	\[
	\erf[a,b] = \frac 2 {\sqrt \pi} \int_a^be^{-t^2}\, \mathrm d t
	\] with $a, b\in\mathbb R$. From \cite[Eq. 7.1.21]{AbramowitzStegun} we know that
	\begin{equation}\label{eq:erf_asymptotics}
	2\exp(x^2)\int_x^\infty e^{-t^2}\,\mathrm dt = \frac{1}{x} - \frac{1}{2x^3} + Q(x), \qquad |Q(x)|\le \frac{3}{4 x^5},
	\end{equation} for $x > 0$. 

	Using \eqref{eq:erf_asymptotics} we find 
	\begin{equation}\label{eq:upsilon_expansion_negative}
		\Upsilon(\eta) =  \frac{ \frac 1 {2\eta^2} + \eta Q(-\eta)  -\eta \exp(-A^2+2A\eta) \big( \frac 1 \eta + \frac 1 {A-\eta} - \frac{1}{2(A-\eta)^3} + Q(A-\eta) \big)}{ -\frac 1 \eta + \frac 1 {2\eta^3} + Q(-\eta)  -\exp(-A^2+2A\eta) \big( \frac 1 {A-\eta} - \frac{1}{2(A-\eta)^3} + Q(A-\eta) \big) }
	\end{equation} for $\eta < 0$. This, in particular, implies $\Upsilon(\eta)\to 0 $ in the limit $\eta\to -\infty$, for fixed $A$. Moreover, a direct computation shows $\Upsilon(A/2) = A/2$. With
	\begin{equation*}
		\frac{\partial }{\partial \widetilde \mu} \int_{\mathbb C}\zeta(z)|z|^2\,\mathrm d z  = \beta \int_{\mathbb C}\zeta(z)\left ( |z|^2 - \left( \int_{\mathbb C}\zeta(w)|w|^2\,\mathrm d w  \right) \right)^2 \,\mathrm d z > 0,
	\end{equation*} 
	we conclude that for $0 < M < \widetilde b N$ there exists a unique solution $\eta$ to the equation $\sqrt{\beta h}M=\Upsilon(\eta)$. In the following we derive the asymptotic behavior of this solution for large $N$. 
	
	Let us first consider the case $M\lesssim N^{5/6-\eps}$, where we have $\Upsilon(\eta) \lesssim N^{-\eps}$. By comparing this to 
	\begin{equation}\label{eq:upsilon_0}
		\Upsilon(0) = \frac{1-\exp(-A^2)}{2\int_0^A e^{-t^2} \, \mathrm d t} \to \pi^{-1/2} \quad \text{ as } A\to\infty,
	\end{equation} and using the monotonicity of $\Upsilon(\eta)$, we see that $\eta<0$, for sufficiently large $N$. Thus, it follows from \eqref{eq:upsilon_expansion_negative} that $\eta\to-\infty$ as $N\to\infty$, and moreover
	\[
	\sqrt{\beta h} M = -\frac{1}{2\eta} + \mathcal O(|\eta|^{-3}),
	\] which immediately implies \eqref{eq:asymptotics_mu_tilde_small}.
	
	Let us now assume that $M \gtrsim N^{5/6+\eps}$, in which case we have
	\begin{equation}
		\Upsilon(\eta) = \sqrt{\beta h} M \gtrsim N^\eps.
		\label{eq:AppB1}
	\end{equation}
	From the assumption $\widetilde b < \widetilde c/2$ we know that
	\begin{equation*}
		\Upsilon(\eta) = \sqrt{\beta h} M < A/2 = \Upsilon(A/2).
	\end{equation*} In combination with \eqref{eq:upsilon_0} and the monotonicity of $\Upsilon(\cdot)$, this allows us to conclude that $0< \eta < A/2$, for large $N$. An application of \eqref{eq:erf_asymptotics} shows 
	\begin{equation}\label{eq:upsilon_expansion_positive}
		\Upsilon(\eta) = \eta + \frac{ 1- e^{-A(A-2\eta)} }{ 2e^{\eta^2}\sqrt\pi - \frac 1 \eta + \frac 1 {2\eta^3} - Q(\eta) -e^{-A(A-2\eta)}\big( \frac 1 {A-\eta} - \frac 1 {2(A-\eta)^3} + Q(A-\eta) \big) },
	\end{equation}
	which together with \eqref{eq:AppB1} implies $\eta\to+\infty$ as $N\to\infty$, provided $M \gtrsim N^{5/6+\eps}$. Using again $\eta< A/2$ and \eqref{eq:upsilon_expansion_positive}, we see that 
	\begin{equation*}
		| \eta - \sqrt{\beta h} M | \lesssim \exp(-N^{2\eps}),
	\end{equation*} which is \eqref{eq:asymptotics_mu_tilde_large}.

	Finally, the bound \eqref{eq:mu_tilde_general_bound} follows from the asymptotics of $\Upsilon$ for $\eta\to\pm\infty$ and the fact that if $\eta$ is bounded, then $\widetilde{\mu}\lesssim (\beta N)^{-1/2}$. 
\end{proof}

\section{Estimate of the expected number of particles in $\Gamma$} \label{app:perturbation_number_particles}

This section is devoted to the proof of Lemma \ref{lem:perturbation_number_particles}. We start by showing \eqref{eq:difference_particle_number} uniformly in $\widetilde \mu \in \mathbb R$. We will often need to remove cutoffs from the expectation of observables on the Gibbs state $\widetilde G^{\mathrm{ diag}}$, producing errors that are exponentially small in $N$. We will omit such errors because they can be absorbed in the remaining error bounds. Since the proof is analogous to that of Lemma \ref{lem:V1+V2_cancellation} we only sketch it. 

Using \eqref{eq:denominator_bound_first} and \eqref{eq:upperBoundJastrow} without the restriction $(k,h) \neq (i,j)$ and expanding the relevant numerator and denominator, we find
\begin{equation}\label{eq:N_Gamma_upper_bound_1}
	\begin{split}
		\tr [ \mathcal N \Gamma ] = & \tr[\mathcal N \widetilde \Gamma_0] + \frac 1 2 \int_{\mathbb C } \zeta(z) \sum_{\alpha\in \widetilde{ \mathcal A}} \widetilde \lambda_\alpha \int_{\Lambda} \rho_{z,\alpha}^{(1)}(x)  \, \mathrm d x \int_{\Lambda^2} u_{12} \rho_{z,\alpha}^{(2)}(x_1, x_2)  \, \mathrm d x_1 \, \mathrm d x_2 \, \mathrm d z   \\
		& - \frac 1 {2}  \int_{\Lambda^3} u_{12} \rho^{(3)}_{\Gamma_0} (x_1,x_2,x_3)  \, \mathrm d x_1 \, \mathrm d x_2  \, \mathrm d x_3 + \mathcal O(N \ell^2 + N^3\ell^4 ).
	\end{split}
\end{equation} 
From \eqref{eq:action_weyl_momentum} and \eqref{eq:quadratic_expectation_psi_a} we know that
\begin{equation}\label{eq:rho1za_asymptotics}
	\begin{split}
		\int_{\Lambda} \rho_{z,\alpha}^{(1)} (x) \, \mathrm d x = & \langle \phi_{z, \alpha} , \mathcal N  \phi_{z, \alpha}\rangle = |z|^2 + N_\alpha +  \sum_{p\in P_{\mathrm B} } v_p^2 \Big( 1 +  \langle \Psi_\alpha ,a_p^* a_p  \Psi_\alpha \rangle \Big).
	\end{split}
\end{equation} Using \eqref{eq:nabla_rho2_bound}, we see that the second term on the right-hand side of \eqref{eq:N_Gamma_upper_bound_1} equals
\begin{equation}\label{eq:second_term_N_Gamma}
	\frac 1 2 \Big( \int u_\ell \Big) \int_{\mathbb C } \zeta(z) \sum_{\alpha\in \widetilde{ \mathcal A}} \widetilde \lambda_\alpha \int_{\Lambda} \rho_{z,\alpha}^{(1)}(x)  \, \mathrm d x \int_{\Lambda} \rho_{z,\alpha}^{(2)}(x, x)  \, \mathrm d x \, \mathrm d z  + \mathcal E_{\mathrm C}^{(1)},
\end{equation} with 
\begin{equation} \label{eq:EC1_error_bound}
\begin{split}
		\mathcal E_{\mathrm C}^{(1)} = & \frac 1 2 \int_{\mathbb C } \zeta(z) \sum_{\alpha\in \widetilde{ \mathcal A}} \widetilde \lambda_\alpha \int_{\Lambda} \rho_{z,\alpha}^{(1)}(x)  \, \mathrm d x \int_{\Lambda^2} u_{\ell}(x_2-x_1) \\
		& \times \int_0^1 (x_2-x_1) \cdot \nabla_2 \rho_{z,\alpha}^{(2)}(x_1, x_1 + t (x_2- x_1)) \, \mathrm d t \, \mathrm d x_1 \, \mathrm d x_2 \, \mathrm d z. 
\end{split}
\end{equation} It follows from \eqref{eq:nabla_rho2_bound}, \eqref{eq:ul_integral}, \eqref{eq:rho1za_asymptotics} and the cutoff in $\widetilde \lambda_\alpha$ that 
\begin{equation}\label{eq:eq:EC1_error_bound2}
	\begin{split}
		| \mathcal E_{\mathrm C}^{(1)}| \lesssim  & \mathfrak a_N N \ell^3 \int_{\mathbb C } \zeta(z) \sum_{\alpha\in \widetilde{ \mathcal A}} \widetilde \lambda_\alpha \sup_{x_1, x_2 \in \Lambda} \big| \nabla_2 \rho_{z,\alpha}^{(2)}(x_1, x_2) \big | \mathrm d z \\
		\lesssim & \ell^3 \Big( N^{3/2} \big( \tr_{ \mathfrak{F}_+ }[ \mathcal K \widetilde G^{ \mathrm{ diag } } ] \big)^{1/2} + N^{4 \delta_{ \mathrm{ Bog } }} ( \beta^{-2} N^{2 \delta_{ \mathrm{ Bog } }} + 1 ) + N^{1 + 2\delta_{\mathrm{Bog}}} (1 + \beta^{-1}) \Big) \\
		\lesssim & \ell^3 N^2 ( \beta^{-1/2} +  1 ).
	\end{split}
\end{equation} To come to the last line we also used the second bound in \eqref{eq:gamma_Bog_bound_momentum} and the assumption $\delta_{ \mathrm{ Bog } } < 1/12$. 

The same arguments used to prove \eqref{eq:rho2za_bound} show 
\begin{equation}
	\label{eq:rho2za_asymptotics}
	\rho_{z, \alpha}^{(2)}(x, x) =  |z|^4 + 4 |z|^2 N_\alpha + 2N_\alpha(N_\alpha-1)
	+\mathcal O \Big( ( |z|^2 + N_\alpha )   ( N_\alpha^< + N^{ \delta_{ \mathrm{ Bog } } } ) + N^{ 3\delta_{ \mathrm{ Bog } }} ( N_\alpha^{<}  + 1 )^2  \Big), 
\end{equation} uniformly in $x\in\Lambda$. This follows from the observation that, thanks to the translation invariance of the eigenstates $\Psi_\alpha$, the phases in the expansion \eqref{eq:computation_rho_2_z_a} drop out when $x_1 = x_2$. Inserting \eqref{eq:rho2za_asymptotics} into \eqref{eq:second_term_N_Gamma} and using \eqref{eq:ul_integral}, \eqref{eq:rho1za_asymptotics}, \eqref{eq:EC1_error_bound}, \eqref{eq:eq:EC1_error_bound2} and Corollary \ref{cor:expectation_powers_N}, we see that the second term on the right-hand side of \eqref{eq:N_Gamma_upper_bound_1} equals
\begin{equation}\label{eq:N_cancellation_positive}
	\begin{split}
		\frac 1 {2} \Big ( \int u_\ell \Big ) &\int_{\mathbb C } \zeta(z)  \Big\{ |z|^6 + 5|z|^4 \tr_{\mathfrak{F}_+}[ \mathcal N_+  G^{ \mathrm{ diag } }] + 6|z|^2 ( \tr_{\mathfrak{F}_+}[ \mathcal N_+  G^{ \mathrm{ diag } }] )^2  \\ 
		& + 2 ( \tr_{\mathfrak{F}_+}[ \mathcal N_+  G^{ \mathrm{ diag } }] )^3 \Big \} \, \mathrm dz + \mathcal O\Big( \ell^2 N^{1+\delta_{ \mathrm { Bog } } } ( \beta^{-1} + 1)  + \ell^3 N^2 (\beta^{-1/2} + 1) + \ell^4 N^3 \Big).
	\end{split}
\end{equation} 

As for the term on the second line of \eqref{eq:N_Gamma_upper_bound_1}, it is easy to see, using Lemma \ref{lem:rho_k_Gamma_bound}, that 
\begin{equation*}
	\int_{\Lambda^3} u_{12} \rho^{(3)}_{\Gamma_0} (x_1,x_2,x_3)  \, \mathrm d x_1 \, \mathrm d x_2  \, \mathrm d x_3 = \Big ( \int u_\ell \Big )  \int_{\Lambda^2} \rho^{(3)}_{\Gamma_0} (x_1,x_1,x_2)  \, \mathrm d x_1 \, \mathrm d x_2 + \mathcal O( N^2\ell^3\beta^{-1/2} ).
\end{equation*} From here, using Wick's theorem as in the proof of \eqref{eq:V2_final_uper_bound} we arrive at 
\begin{equation}\label{eq:N_cancellation_negative}
	\begin{split}
		- \frac 1 {2}  \int_{\Lambda^3} & u_{12} \rho^{(3)}_{\Gamma_0} (x_1,x_2,x_3)  \, \mathrm d x_1 \, \mathrm d x_2  \, \mathrm d x_3 \\
		= & - \frac 1 {2} \Big ( \int u_\ell \Big ) \int_{\mathbb C } \zeta(z)  \Big\{ |z|^6 + 5|z|^4 \tr_{\mathfrak{F}_+}[ \mathcal N_+  G^{ \mathrm{ diag } }]  \\ 
		& + 6|z|^2 ( \tr_{\mathfrak{F}_+}[ \mathcal N_+  G^{ \mathrm{ diag } }] )^2 + 2 ( \tr_{\mathfrak{F}_+}[ \mathcal N_+  G^{ \mathrm{ diag } }] )^3 \Big \} \, \mathrm dz + \mathcal O\Big( N^2\ell^3\beta^{-1/2} + N\ell^2 \beta^{-1} \Big). 
	\end{split}
\end{equation} Inserting \eqref{eq:N_cancellation_positive} and \eqref{eq:N_cancellation_negative} into \eqref{eq:N_Gamma_upper_bound_1}, we get \eqref{eq:difference_particle_number}. It remains to prove the existence statement of the lemma. 

We denote
\begin{equation*}
	M(\widetilde \mu) = \int_{\mathbb C} |z|^2 \zeta(z) \, \mathrm d z
\end{equation*} and observe that $\tr_{\mathfrak{F}_+}[\mathcal N_+ \widetilde G(z)]$ is independent of $z\in\mathbb C$, and that
\begin{equation*}
	\tr[ \mathcal N \widetilde \Gamma_0 ] =	M(\widetilde \mu) + \tr_{\mathfrak{F}_+}[\mathcal N_+ \widetilde G(z)] 
\end{equation*} holds. Using Lemmas \ref{lem:Tz_N_bound}, \ref{lem:cutoff_a_remainder}, \ref{lem:n_g(z)_ideal_difference} and Equation \eqref{eq:n0_n0tilde_difference_2}, it is easy to see that 
\begin{equation*}
	\left| \tr_{\mathfrak{F}_+}[\mathcal N_+ \widetilde G(z)] - N + N_0 \right| \lesssim \frac{N_0}{N\beta} + \frac{N_0^2}{N^2}.
\end{equation*}
We can thus write 
\begin{equation*}
	\tr[ \mathcal N \Gamma ] - N = M(\widetilde \mu) - N_0 +  \left( 	\tr[ \mathcal N  \widetilde \Gamma_0 ] - 	\tr[ \mathcal N \Gamma ] \right) + \mathcal O \Big( \frac{N_0}{N\beta} +  \frac{N_0^2}{N^2}\Big),
\end{equation*} which, for $\beta \gtrsim \beta_{\mathrm{c}}$, $\delta_{ \mathrm{ Bog } }< 1/12$ and $\ell\le cN^{-7/12}$, for some sufficiently small $c>0$, yields 
\begin{equation} \label{eq:N_Gamma_bound_existence}
	M(\widetilde \mu) - N_0 -  \Big( \frac 1 2 + C \frac{N_0}{N} \Big)N^{2/3} \le 	\tr[ \mathcal N \Gamma ] - N \le 	M(\widetilde \mu) - N_0 + \Big( \frac 1 2 + C \frac{N_0}{N} \Big)  N^{2/3},
\end{equation} for some constant $C>0$ independent of $\widetilde \mu$ and $N$. From the proof of Lemma \ref{lem:mu_tilde_condensate} we know that 
\begin{equation}\label{eq:M(mu)_limits}
	\lim_{\widetilde \mu \to -\infty} M(\widetilde \mu) = 0,  \qquad \lim_{\widetilde \mu \to +\infty} M(\widetilde \mu) = \widetilde  c N ,
\end{equation} for fixed $N\in\mathbb N$. Equations \eqref{eq:N_Gamma_bound_existence}, \eqref{eq:M(mu)_limits} and the assumption $N_0\ge N ^{2/3}$ imply that $ \tr[ \mathcal N \Gamma ] - N $ takes positive and negative values as a function of $\widetilde \mu$, for fixed $N$ large enough. The claim now follows from the continuity of the map 
\begin{equation*}
	\widetilde \mu \mapsto \tr[\mathcal N \Gamma] = \int_{ \mathbb C } \zeta(z) \sum_{\alpha\in \widetilde{ \mathcal A}}\frac {  \widetilde \lambda_\alpha } { \| F \phi_{z,\alpha} \|^2 } \sum_{n =1}^\infty n \int_{\Lambda^n} \left| \big( F_n  \phi_{z,\alpha}^n \big)\right |^2 \, \mathrm d X^n \, \mathrm d z,
\end{equation*} 
which is a consequence of the dominated convergence theorem. \qed

\end{appendix}


\printbibliography

\vspace{0.5cm}
\noindent (Marco Caporaletti) Institute of Mathematics, University of Zurich \\ 
Winterthurerstrasse 190, 8057 Zurich, Switzerland \\ 
E-mail address: \texttt{marco.caporaletti@math.uzh.ch}\\
\\
(Andreas Deuchert) Institute of Mathematics, University of Zurich \\ 
Winterthurerstrasse 190, 8057 Zurich, Switzerland \\ 
E-mail address: \texttt{andreas.deuchert@math.uzh.ch}\\

\end{document}